\documentclass[conference]{IEEEtran}
\IEEEoverridecommandlockouts
\usepackage{graphics}
\usepackage{graphicx}
\usepackage{balance}
\usepackage{epstopdf}
\usepackage{xspace}
\usepackage{subfigure}
\usepackage{amsmath}
\usepackage{amsthm}
\usepackage{amssymb}
\usepackage{bm}
\usepackage{mathrsfs}
\usepackage{array}
\usepackage{textcomp}
\usepackage{weiwAlgorithm}
\usepackage{url}
\usepackage{subfigure}
\usepackage[table]{xcolor}
\usepackage[symbol]{footmisc}

\newtheorem{example}{\textbf{Example}}
\newtheorem{remark}{Remark}
\newtheorem{theorem}{\textbf{Theorem}}
\newtheorem{lemma}{\textbf{Lemma}}

\newtheorem{definition}{\textbf{Definition}}

\newcommand{\lshe}{\textit{LSH-E}\xspace}
\newcommand{\gbkmv}{\textit{GB-KMV}\xspace}
\newcommand{\kmv}{\textit{KMV}\xspace}
\newcommand{\lsh}{\textit{LSH}\xspace}
\newcommand{\gkmv}{\textit{G-KMV}\xspace}
\newcommand{\ds}{$\mathcal{S}$\xspace}
\newcommand{\freset}{\textit{FrequentSet}\xspace}
\newcommand{\ppjoin}{\textit{PPjoin}\xspace}

\def\BibTeX{{\rm B\kern-.05em{\sc i\kern-.025em b}\kern-.08em
    T\kern-.1667em\lower.7ex\hbox{E}\kern-.125emX}}

\renewcommand{\baselinestretch}{0.952} 
\IEEEoverridecommandlockouts

\begin{document}

\title{GB-KMV: An Augmented KMV Sketch for Approximate Containment Similarity Search}

\author{Yang Yang$^{\dagger}$, Ying Zhang$^{\S}$, Wenjie Zhang$^{\dagger}$, Zengfeng Huang$^{\dagger}$ \\
$^\dagger$University of New South Wales, $^\S$University of Technology Sydney\\
		\fontsize{10}{10}\{yang.yang, zhangw\}@cse.unsw.edu.au, ying.zhang@uts.edu.au
	}
\maketitle

\begin{abstract}
In this paper, we study the problem of approximate containment similarity search.
Given two records $Q$ and $X$, the containment similarity between $Q$ and $X$ with respect to $Q$ is $\frac{|Q \cap X |}{|Q|}$.
Given a query record $Q$ and a set of records $\mathcal{S}$,
the containment similarity search finds a set of records from $\mathcal{S}$ whose containment similarity regarding $Q$ is not less than the given threshold.
This problem has many important applications in commercial and scientific fields such as record matching and domain search.
Existing solution relies on the \emph{asymmetric \lsh} method by transforming the containment similarity to well-studied Jaccard similarity.
In this paper, we use a inherently different framework by transforming the containment similarity to set intersection. We propose a novel augmented \kmv sketch technique, namely \gbkmv, which is data-dependent and can achieve a much better trade-off between the sketch size and the accuracy. We provide a set of theoretical analysis to underpin the proposed augmented \kmv sketch technique, and show that it outperforms the state-of-the-art technique \lshe in terms of estimation accuracy under practical assumption. Our comprehensive experiments on real-life datasets verify that \gbkmv is superior to \lshe in terms of the space-accuracy trade-off,  time-accuracy trade-off, and the sketch construction time.
For instance, with similar estimation accuracy (F-$1$ score), \gbkmv is over 100 times faster than \lshe on several real-life datasets.
\end{abstract}


\section{Introduction}
\label{sct:introduction}

In many applications such as information retrieval, data cleaning, machine learning and user recommendation,
an object (e.g., document, image, web page and user) is described by a set of elements (e.g., words, $q$-grams, and items).
One of the most critical components in these applications is to define the set similarity between two objects and develop corresponding similarity query processing techniques.
Given two records (objects) $X$ and $Y$, a variety of  similarity functions/metrics have been identified in the literature for different scenarios (e.g.,~\cite{indyk1998approximate},~\cite{charikar2002similarity}).
Many indexing techniques have been developed to support efficient \emph{exact} and \emph{approximate} lookups and joins based on these similarity functions.

Many of the set similarity functions studied are symmetric functions, i.e., $f(X,Y)$ $= f(Y,X)$, including widely used Jaccard similarity and Cosine similarity.
In recent years, much research attention has been given to the asymmetric set similarity functions, which are more appropriate in some applications.
Containment similarity (a.k.a, Jaccard containment similarity) is one of the representative asymmetric set similarity functions,
where the similarity between two records $X$ and $Y$ is defined as $f(X,Y) =$ $\frac{|X \cap Y|}{|X|}$ in which $|X \cap Y |$ and $|X|$
are intersection size of $X$ and $Y$ and the size of $X$, respectively.

Compared with symmetric similarity such as Jaccard similarity, containment similarity gives special consideration on the query size, which makes it more suitable in some applications.
As shown in~\cite{shrivastava2015asymmetric}, containment similarity is useful in record matching application.
Given two text descriptions of two restaurants
$X$ and $Y$ which are represented by two ``set of words'' records: $\{$\emph{five}, \emph{guys},
\emph{burgers}, \emph{and}, \emph{fries}, \emph{downtown}, \emph{brooklyn}, \emph{new}, \emph{york}$\}$ and $\{$\emph{five}, \emph{kitchen}, \emph{berkeley}$\}$ respectively.
Suppose query $Q$ is \{\emph{five}, \emph{guys}\}, we have that the Jaccard similarity of $Q$ and $X$ (resp. $Y$) is $\frac{2}{9} = 0.22$ ($\frac{1}{4} = 0.25$).
Note the Jaccard similarity is $f(Q,X) = \frac{ |Q \cap X|}{ |Q \cup X|}$.
Based on the Jaccard similarity, record $Y$ matches better to query $Q$, but intuitively $X$ should be a better choice.
This is because the Jaccard similarity unnessesarily favors the short records.
On the other hand, the containment similarity will lead to the desired order with $f(Q,X) = \frac{2}{2}$ $= 1.0$ and $f(Q,Y) = \frac{1}{2}$ $= 0.5$.
Containment similarity search can also support online error-tolerant search for matching user queries against addresses (map service)
and products (product search). This is because the regular keyword search is usually based on the containment search, and containment similarity search provides a natural error-tolerant alternative~\cite{agrawal2010indexing}.
In~\cite{zhu2016lsh}, Zhu et al. show that containment similarity search is essential in domain search which enables users to effectively search Open Data.

The containment similarity is also of interest to applications of computing the fraction of values of one column that are contained in another column.
In a dataset, the discovery of all inclusion dependencies is a crucial part of data profiling efforts. 
It has many applications such as foreign-key detection and data integration(e.g., ~\cite{de2003zigzag, lopes2002discovering, bauckmann2006efficiently, papenbrock2015divide, kruse2017fast}).

\vspace{1mm}
\noindent \textbf{Challenges.}
The problem of containment similarity search has been intensively studied in the literature in recent years (e.g.,~\cite{agrawal2010indexing,shrivastava2015asymmetric,zhu2016lsh}).
The key challenges of this problem come from the following three aspects:
($i$) The number of elements (i.e., vocabulary size) may be very large. For instance, the vocabulary will blow up quickly when the higher-order shingles are used~\cite{shrivastava2015asymmetric}.
Moreover, query and record may contain many elements. To deal with the sheer volume of the data, it is desirable to use sketch technique to provide effectively and efficiently approximate solutions.
($ii$) The data distribution (e.g., record size and element frequency) in real-life application may be highly skewed.
This may lead to poor performance in practice for \emph{data independent} sketch methods. 
($iii$) A subtle difficulty of the approximate solution comes from the asymmetric property of the containment similarity.
It is shown in~\cite{shrivastava2014asymmetric} that there cannot exist any locality sensitive hashing (\lsh) function family for containment similarity search.

To handle the large scale data and provide quick response, most existing solutions for containment similarity search seek to the \emph{approximate} solutions.
Although the use of \lsh is restricted, the novel \emph{asymmetric \lsh} method has been designed in~\cite{shrivastava2014asymmetric} to
address the issue by padding techniques. Some enhancements of asymmetric \lsh techniques are proposed in the following works by introducing different functions (e.g.,~\cite{shrivastava2015asymmetric}).
Observe that the performance of the existing solutions are sensitive to the skewness of the record size, Zhu {\emph et. al} propose a partition-based method based on Minhash \lsh function. 
By using optimal partition strategy based on the size distribution of the records, the new approach can achieve much better time-accuracy trade-off.

We notice that all existing approximate solutions rely on the \lsh functions by transforming the containment similarity to well-studied \textit{Jaccard similarity}.
That is, 
$$
\frac{|Q \cap X|}{|Q|} =  \frac{|Q \cap X|}{|Q \cup X|} \times|Q \cup X| \times \frac{1}{|Q|}
$$
As the size of query is usually readily available, the estimation error come from the computation of Jaccard similarity and union size of $Q$ and $X$.
Note that although the union size can be  derived from jaccard similarity~\cite{zhu2016lsh}, the large variance caused by the combination of two estimations remains. 
This motivates we to use a different framework by transforming the containment similarity to \emph{set intersection size estimation}, 
and the error is only contributed by the estimation of $|Q \cap X|$. 
The well-known \kmv sketch~\cite{beyer2007synopses} has been widely used to estimate the set intersection size, which can be immediately applied to our problem.
However, this method is \textit{data-independent} and hence cannot well handle the skewed distributions of records size and element frequency, which is common in real-life applications.
Intuitively, the record with larger size and the element with high-frequency should be allocated more resources.
In this paper, we theoretically show that the existing \kmv-sketch technique cannot consider these two perspectives by simple heuristics, e.g., explictly allocating more resource to record with large size.
Consequently, we develop an augmented \kmv sketch to exploit both record size distribution and the element frequency distribution for better space-accuracy and time-accuracy trade-offs.
Two technique are proposed: ($i$) we impose a global threshold to \kmv sketch, namely \gkmv sketch, to achieve better estimate accuracy. 
As disscussed in Section~\ref{subsec:alg_motivation}(2), this technique cannot be extended to the Minhash \lsh. 
($ii$) we introduce an extra buffer for each record to take advantage of the skewness of the element frequency. 
A cost model is proposed to carefully choose the buffer size to optimize the accuracy for the given total space budget and data distribution. 

\vspace{1mm}
\noindent \textbf{Contributions.}
Our principle contributions are summarized as follows.

\begin{itemize}
\item We propose a new augmented \kmv sketch technique, namely \gbkmv, for the problem of approximate containment similarity search.
By imposing a global threshold and an extra buffer for \kmv sketches of the records, we significanlty enhance the performance 
as the new method can better exploit the data distributions.

\item We provide theoretical underpinnings to justify the design of \gbkmv method.
We also theoretically show that \gbkmv outperforms the state-of-the-art technique \lshe in terms of accuracy under realistic assumption on data distributions.

\item Our comprehensive experiments on real-life set-valued data from various applications demonstrate the effectiveness and efficiency of our proposed method.
\end{itemize}


\noindent \textbf{Road Map.} The rest of the paper is organized as follows.
Section~\ref{sct:preliminaries} presents the preliminaries.
Section~\ref{sct:exist_solutions} introduces the existing solutions.
Our approach, \gbkmv sketch, is devised in Section~\ref{sct:approach}.
Extensive experiments are reported in Section~\ref{sct:experiment},
followed by the related work in Section~\ref{sct:related}.
Section~\ref{sct:conclusion} concludes the paper.

\vspace{-2mm}
\section{Preliminaries}
\label{sct:preliminaries}
In this section, we first formally present the problem of containment similarity search,
then introduce some preliminary knowledge.
In Table~\ref{tb:notations}, we summarize the important mathematical notations appearing throughout this paper.

\begin{table}
\footnotesize
  \centering

  \vspace{-1mm}
    \begin{tabular}{|c|l|}
      \hline
      \cellcolor{gray!25}\textbf{Notation} & \cellcolor{gray!25}\textbf{Definition}             \\ \hline

      $\mathcal{S}$&  a collection of records\\ \hline
      $X, Q$  &  record, query record\\ \hline
      $x, q$  &  record size of $X$, query size of $Q$ \\ \hline
      $J(Q, X), s$  &  Jaccard similarity between query $Q$ and set $X$ \\ \hline
      $C(Q, X), t$  &  Containment similarity of query $Q$ in set $X$ \\ \hline
      $s^*$  &   Jaccard similarity threshold  \\ \hline
      $\mathcal{L}_X$  & the \kmv signature (i.e., hash values) of record X \\ \hline
      $h(X)$  & all hash values of the elements in record X \\ \hline
      $\mathcal{H}_X$  & the buffer of record X \\ \hline
      $t^*$  &   containment similarity threshold  \\ \hline
      $b$  &   sketch space budget, measured by the number of\\
           &   signatures (i.e., hash values or elements)   \\ \hline
      $\tau$  & the global threshold for hash values  \\ \hline
      $r$  & the buffer size(with \textit{bit} unit) of \gbkmv sketch \\ \hline
      $m$  & number of records in dataset $\mathcal{S}$ \\ \hline
      $n$  & number of distinct elements in dataset $\mathcal{S}$ \\ \hline

      \hline
    \end{tabular}
\vspace{1mm}
\caption{\small The summary of notations}
\label{tb:notations}
\vspace{-4mm}
\end{table}

\subsection{Problem Definition}

In this paper, the element universe is $\mathcal{E}= \{e_1, e_2,...,e_n\}$.
Let $\mathcal{S}$ be a collection of records (sets) $\{X_1, X_2, ..., X_m\}$'
where $X_i$ ( $1\leq i \leq m$) is a set of elements from $\mathcal{E}$.

Before giving the definition of containment similarity, we first introduce the Jaccard similarity.
\begin{definition}[\textbf{Jaccard Similarity}]
Given two records $X$ and $Y$ from $\mathcal{S}$, the Jaccard similarity between $X$ and $Y$ is defined as the size of the intersection divided by the size of the union, which is expressed as
\begin{equation}
\label{eq:jacdef}
J(X, Y) = \frac{|X \cap Y|}{|X \cup Y|}
\end{equation}
\end{definition}

Similar to the Jaccard similarity, the containment similarity (a.k.a Jaccard containment similarity) is defined as follows.

\begin{definition}[\textbf{Containment Similarity}]
Given two records $X$ and $Y$ from $\mathcal{S}$, the containment similarity of $X$ in $Y$, denoted by $C(X, Y)$ is the size of the intersection divided by record size $|X|$,
which is formally defined as
\begin{equation}
\label{eq:condef}
C(X,Y) = \frac{|X \cap Y|}{|X|}
\end{equation}
\end{definition}
Note that by replacing the union size $|X \cup Y|$ in Equation~\ref{eq:jacdef} with size $|X|$, we get the containment similarity. It is easy to see that Jaccard similarity is symmetric while containment similarity is asymmetric.

In this paper, we focus on the problem of containment similarity search which is to look up a set of records
whose containment similarity towards a given query record is not smaller than a given threshold.
The formal definition is as follows.

\begin{definition}[\textbf{Containment Similarity Search}]
Given a query $Q$, and a threshold $t^{*}\in[0,1]$ on the containment similarity, search for records $\{X: X\in \mathcal{S}\}$ from a dataset $\mathcal{S}$ such that:
\begin{equation}
\label{eq:probdef}
C(Q, X) \geq t^{*}
\end{equation}
\end{definition}
Next, we give an example to show the problem of containment similarity search.

\begin{example}
\label{exam:prob_def}
Fig.~\ref{fig:problem example1} shows a dataset with four records $\{ X_1$, $X_2$, $X_3$, $X_4 \}$,
and the element universe is $\mathcal{E}= \{e_1, e_2,...,e_{10}\}$.
Given a query $Q=\{e_1, e_2, e_3, e_5, e_7, e_9\}$ and a containment similarity threshold $t^* = 0.5$, the records satisfying $C(Q,X_i)\geq 0.5$ are $X_1$, $X_2$.
\end{example}

\begin{figure}[!t]
\centering
\vspace{0.2cm}
\begin{tabular}[b]{l l l}
\hline
{\bfseries id} & {\bfseries record} & $C(Q, X_i)$\\
\hline
$X_1$ & $\{e_1, e_2, e_3, e_4, e_7\}$  & $0.67$\\

$X_2$ & $\{e_2, e_3, e_5\}$  & $0.5$\\

$X_3$ & $\{e_2, e_4, e_5\}$ & $0.33$\\

$X_4$ & $\{e_1, e_2, e_6, e_{10}\}$ & $0.33$\\

$Q$ & $\{e_1, e_2, e_3, e_5, e_7, e_9\}$ & \ \ \\
\hline
\end{tabular}
\caption{\small A four-record dataset and a query $Q$; $C(Q,X_i)$ is the containment similarity of $Q$ in $X_i$}
\label{fig:problem example1}
\vspace{-1mm}
\end{figure}

\noindent \textbf{Problem Statement.} %
In this paper, we investigate the problem of \emph{approximate containment similarity search}.
For the dataset $\mathcal{S}$ with a large number of records, we aim to build a synopses of the dataset such that it 
($i$) can efficiently support containment similarity search with high accuracy,
($ii$) can handle large size records, and ($iii$) has a compact index size.

%


\subsection{Minwise Hashing}
\label{subsec:MinHash}
Minwise Hashing is proposed by Broder in~\cite{broder1997resemblance,broder1998min} for estimating the Jaccard similarity of two records $X$ and $Y$.
Let $h$ be a hash function that maps the elements of $X$ and $Y$ to distinct integers, and define $h_{min}(X)$ and $h_{min}(Y)$ to be the minimum hash value of a record $X$ and $Y$, respectively.
Assuming no hash collision, Broder\cite{broder1997resemblance} showed that the Jaccard similarity of $X$ and $Y$ is the probability of two minimum hash values being equal:
$Pr[h_{min}(X) = h_{min}(Y)] = J(X, Y)$.
Applying such $k$ different independent hash functions $h_1, h_2,..., h_k$ to a record $X$($Y$, resp.), the MinHash signature of $X$($Y$, resp.) is to keep $k$ values of $h^i_{min}(X)$( $h_{min}(Y)$, resp.) for  $k$ functions.
Let $\mathbf{n}_i, i = 1,2,...,k$ be the indicator function such that
\begin{equation}
\label{eq:indcfunc}
\mathbf{n}_i :=
\begin{cases}
1 &\text{if } h^i_{min}(X) = h^i_{min}(Y), \\
0 &\text{otherwise}.
\end{cases}
\end{equation}
then the Jaccard similarity between record $X$ and $Y$ can be estimated as
\begin{equation}
\label{eq:jacestimator}
\hat{s} = \hat{J}(X, Y) = \frac{1}{k}\sum\limits_{i=1}^{k} \mathbf{n}_i
\end{equation}
Let $s = J(X, Y)$ be the Jaccard similarity of set $X$ and $Y$, then the expectation of $\hat{J}$ is
\begin{equation}
\label{eq:jaccexpe}
E(\hat{s}) = s
\end{equation}
and the variance of $\hat{s}$ is
\begin{equation}
\label{eq:jacvar}
Var(\hat{s}) = \frac{s(1-s)}{k}
\end{equation}

\subsection{KMV Sketch}
\label{subsec:kmv}

The $k$ minimum values(\kmv) technique introduced by Bayer {\emph et. al} in~\cite{beyer2007synopses} is to estimate the number of distinct elements in a large dataset.
Given a no-collision hash function $h$ which maps elements to range $[0,1]$, a \kmv synopses of a record $X$, denoted by $\mathcal{L}_X$, is to keep $k$ minimum hash values of $X$.
Then the number of distinct elements $|X|$ can be estimated by $\widehat{|X|} = \frac{k-1}{U_{(k)}}$ where $U_{(k)}$ is $k$-th smallest hash value.
By $h(X)$, we denote hash values of all elements in the record $X$.

In \cite{beyer2007synopses}, Bayer {\emph et. al} also methodically analyse the problem of distinct element estimation under multi-set operation. As for union operation, consider two records $X$ and $Y$ with corresponding \kmv synopses $\mathcal{L}_X$ and $\mathcal{L}_Y$ of size $k_X$ and $k_Y$, respectively.
In \cite{beyer2007synopses}, $\mathcal{L}_X \oplus \mathcal{L}_Y$ represents the set consisting of the $k$ smallest hash values in $\mathcal{L}_X \cup \mathcal{L}_Y$ where
\begin{equation}
\label{eq:kmv_1}
 k = min(k_X, k_Y)
\end{equation}
Then the \kmv synopses of $X\cup Y$ is $\mathcal{L} = \mathcal{L}_X \oplus \mathcal{L}_Y$.
An unbiased estimator for the number of distinct elements in $X \cup Y$, denoted by $D_{\cup} = |X\cup Y|$ is as follows.
\begin{equation}
\label{eq:kmvunionest}
\hat{D}_\cup = \frac{k-1}{U_{(k)}}
\end{equation}
For intersection operation, the \kmv synopses is $\mathcal{L} = \mathcal{L}_X \oplus \mathcal{L}_Y$  where $k = min(k_X, k_Y)$.
Let $K_{\cap} = |\{v\in \mathcal{L}: v\in \mathcal{L}_X \cap \mathcal{L}_Y \}|$, i.e., $K_\cap$ is the number of common distinct hash values of $\mathcal{L}_X$ and $\mathcal{L}_Y$ within $\mathcal{L}$. Then the number of distinct elements in $X\cap Y$, denoted by $D_\cap$, can be estimated as follows.
\begin{equation}
\label{eq:kmvinterest}
\hat{D}_\cap = \frac{K_\cap}{k}\times\frac{k-1}{U_{(k)}}
\end{equation}
The variance of $\hat{D}_\cap$, as shown in\cite{beyer2007synopses}, is
\begin{equation}
\label{eq:kmvintervar}
Var[\hat{D}_\cap] = \frac{D_\cap(kD_\cup - k^2 - D_\cup + k + D_\cap)}{k(k-2)}
\end{equation}

\section{Existing Solutions}
\label{sct:exist_solutions}

In this section, we present the state-of-the-art technique for the approximate containment similarity search,
followed by theoretical analysis on the limits of the existing solution.


\subsection{\textbf{LSH Ensemble Method}}
\label{subsec:lshe}

\lsh Ensemble technique, \lshe for short, is proposed by Zhu {\emph et. al} in~\cite{zhu2016lsh} to tackle the problem of approximate containment similarity search.
The key idea is : (1) transform the containment similarity search to the well-studied Jaccard similarity search; and
(2) partition the data by length and then apply the \lsh forest~\cite{bawa2005lsh} technique for each individual partition.

\vspace{2mm}
\noindent \textbf{Similarity Transformation.}
Given a record $X$ with size $x = |X|$, a query $Q$ with size $q = |Q|$, containment similarity $t = C(Q,X)$ and Jaccard similarity $s = J(Q,X)$.
The transformation back and forth are as follows.
\begin{equation}
\label{eq:transform}
s = \frac{t}{\frac{x}{q}+1-t},\  t = \frac{(\frac{x}{q}+1)s}{1+s}
\end{equation}

Given the containment similarity search threshold as $t^{*}$ for the query $q$, we may come up with its corresponding Jaccard similarity threshold $s^{*}$ by Equation~\ref{eq:transform}.
A straightforward solution is to apply the existing approximate Jaccard similarity search technique for
each individual record $X \in \mathcal{D}$ with the Jaccard similarity threshold $s^{*}$
(e.g., compute Jaccard similarity between the query $Q$ and a set $X$ based on their MinHash signatures).
In order to take advantages of the efficient indexing techniques (e.g., \lsh forest~\cite{bawa2005lsh}), \lshe will partition the dataset $\mathcal{S}$.

\vspace{2mm}
\noindent \textbf{Data Partition.}
By partitioning the dataset $\mathcal{S}$ according to the record size, \lshe can replace $x$ in Equation~\ref{eq:transform} with its upper bound $u$
(i.e., the largest record size in the partition) as an approximation.
That is, for the given containment similarity $t^*$ we have
\begin{equation}
\label{eq:upper}
s^{*} = \frac{t^{*}}{\frac{u}{q}+1-t^{*}}
\end{equation}
The use of upper bound $u$ will lead to false positives.
In~\cite{zhu2016lsh}, an optimal partition method is designed to minimize the total number of false positives brought by the use of upper bound in each partition.
By assuming that the record size distribution follows the power-law distribution and similarity values are uniformly distributed,
it is shown that the optimal partition can be achieved by ensuring each partition has the equal number of records (i.e., equal-depth partition).

\vspace{2mm}
\noindent \textbf{Containment Similarity Search.}
For each partition $\mathcal{S}_i$ of the data, \lshe applies the dynamic \lsh technique (e.g., \lsh forest~\cite{bawa2005lsh}).
Particularly, the records in $\mathcal{S}_i$ are indexed by a MinHash \lsh with parameter ($b$, $r$) where $b$ is the number of bands used by the \lsh index
and $r$ is the number of hash values in each band.
For the given query $Q$, the $b$ and $r$ values are carefully chosen by considering their corresponding
number of false positives and false negatives regarding the existing records.
Then the candidate records in each partition can be retrieved from the MinHash index according to the corresponding Jaccard similarity thresholds
obtained by Equation~\ref{eq:upper}.
The union of the candidate records from all partitions will be returned as the result of the containment similarity search.

\subsection{Analysis}
\label{subsec:lshe_analysis}

One of the \lshe's advantages is that it converts the containment similarity problem to Jaccard similarity search problem which can be solved by the mature and efficient MinHash \lsh method.
Also, \lshe carefully considers the record size distribution and partitions the records by record size.
In this sense, we say \lshe is a data-dependent method and it is reported that \lshe significantly outperforms existing asymmetric \lsh based solutions~\cite{shrivastava2014asymmetric,shrivastava2015asymmetric}
(i.e., \emph{data-independent} methods) as \lshe can exploit the information of data distribution by partitioning the dataset.
However, this benefit is offset by the fact that the the upper bound will bring extra false positives, in addition to the error from the MinHash technique.

Below we theoretically analyse the performance of \lshe by studying the expectation and variance of its estimator.

Using the notations same as above, let $s = J(Q, X)$ be the Jaccard similarity between query $Q$ and set $X$ and $t = C(Q, X)$ be the containment similarity of $Q$ in $X$.
By Equation~\ref{eq:jacestimator}, given the MinHash signature of query $Q$ and $X$ respectively, an unbiased estimator $\hat{s}$ of Jaccard similarity $s = J(Q, X)$ is the ratio of collisions in the signature, and the variance of $\hat{s}$ is
$Var[\hat{s}] = \frac{s(1-s)}{k}$
where $k$ is signature size of each record.
Then by transformation Equation~\ref{eq:transform}, the estimator $\hat{t}$ of containment similarity $t = C(Q, X)$ by MinHash \lsh is
\begin{equation}
\label{eq:lshest}
  \hat{t} = \frac{(\frac{x}{q}+1)\hat{s}}{1+\hat{s}}
\end{equation}
where $q = |Q|$ and $x = |X|$.
The estimator $\hat{t}'$ of containment similarity $t = C(Q, X)$ by \lshe is
\begin{equation}
\label{eq:lsheest}
  \hat{t}' = \frac{(\frac{u}{q}+1)\hat{s}}{1+\hat{s}}
\end{equation}
where $q = |Q|$ and $u$ is the upper bound of $|X|$.

Next, we use Taylor expansions to approximate the expectation and variance of a function with one random variable~\cite{esmaili2006probability}. We first give a lemma.
\begin{lemma}
\label{lemma:funcAppr}
Given a random variable $X$ with expectation $E[X]$ and variance $Var[X]$, the expectation of $f(X)$ can be approximated as
\begin{equation}
\label{eq:expFunc}
E[f(X)] = f(E[X] + \frac{f''(E[X])}{2} Var[X]
\end{equation}
and the variance of $f(X)$ can be approximated as
\begin{equation}
\label{eq:varFunc}
Var[f(X)] = [f'(E[X])]^2 Var[X] - \frac{[f''(E[X])]^2}{4} Var^2[X]
\end{equation}
\end{lemma}


According to Equation~\ref{eq:lshest}, let $\hat{t} = f(\hat{s}) = \alpha \frac{\hat{s}}{1+\hat{s}}$ where $\alpha = \frac{x}{q} + 1$. We can see that the estimator $\hat{t}$ is a function of $\hat{s}$, and $f'(\hat{s}) = \alpha \frac{1}{(1+\hat{s})}$ and $f''(\hat{s}) = -2\alpha \frac{1}{(1+\hat{s})}$.
Then based on Lemma~\ref{lemma:funcAppr}, the expectation and variance of $\hat{t}$ are approximated as
\begin{equation}
\label{eq:lsh_expe}
  E[\hat{t}] \approx t(1-\frac{1-s}{k(1+s)^2})
\end{equation}
\begin{equation}
\label{eq:lsh_var}
  Var[\hat{t}] \approx \frac{D_{\cap}^2(1-s)[k(1+s)^2 - s(1-s)]}{q^2k^2s(1+s)^4}
\end{equation}
Similarly, the expectation and variance of \lshe estimator $\hat{t}'$ can be approximated as
\begin{equation}
\label{eq:lshe_expe}
  E[\hat{t}'] \approx t(\frac{u+q}{x+q})(1-\frac{1-s}{k(1+s)^2})
\end{equation}
\begin{equation}
\label{eq:lshe_var}
  Var[\hat{t}'] \approx (\frac{u+q}{x+q})^2\frac{D_{\cap}^2(1-s)[k(1+s)^2 - s(1-s)]}{q^2k^2s(1+s)^4}
\end{equation}
The computation details are in technique report~\cite{misc:technique_report}.
Since $u$ is the upper bound of $x$, the variance of \lshe estimator $Var[\hat{t}']$ is larger than that of MinHash \lsh estimator.
Also, by Equation~\ref{eq:lsh_expe} and Equation~\ref{eq:lshe_expe}, we can see that both estimators are \emph{biased} and \lshe method
 is quite sensitive to the setting of the upper bound $u$ by Equation~\ref{eq:lshe_expe}.
Because the presence of upper bound $u$ will \textit{enlarge} the estimator off true value,
\lshe method favours recall while the precision will be deteriorated.
The larger the upper bound $u$ is, the worse the precision will be.
Our empirical study shows that \lshe cannot achieve a good trade-off between accuracy and space, compared with our proposed method.


\section{Our Approach}
\label{sct:approach}
In this section, we introduce an augmented \kmv sketch technique to achieve better space-accuracy trade-off for approximate containment similarity search.
Section~\ref{subsec:alg_motivation} briefly introduces the motivation and main technique of our method, namely \gbkmv.
The detailed implementation is presented in Section~\ref{subsec:alg}, followed by extensive theoretical analysis in Section~\ref{subsec:analysis}.

\subsection{\textbf{Motivation and Techniques}}
\label{subsec:alg_motivation}

The key idea of our method is to propose a \emph{data-dependent} indexing technique such that we can exploit the distribution of the data 
(i.e., record size distribution and element frequency distribution) for better performance of containment similarity search.
We augment the existing \kmv technique by introducing a \underline{g}lobal threshold for sample size allocation
and a \underline{b}uffer for frequent elements, namely \gbkmv, to achieve better trade-off between synopses size and accuracy.
Then we apply the existing set similarity join/search indexing technique to speed up the containment similarity search.

Below we outline the motivation of the key techniques used in this paper.
Detailed algorithms and theoretical analysis will be introduced in
Section~\ref{subsec:alg} and~\ref{subsec:analysis}, respectively.

\vspace{1mm}
\noindent \textbf{(1) Directly Apply KMV Sketch}

Given a query $Q$ and a threshold $t^*$ on containment similarity, the goal is to find record $X$ from dataset $\mathcal{S}$ such that
\begin{equation}
\label{eq:problemdef}
\frac{|Q\cap X|}{|Q|} \geq t^* ,
\end{equation}
Applying some simple transformation to Equation~\ref{eq:problemdef}, we get
\begin{equation}
\label{eq:inter trans}
|Q\cap X| \geq t^* |Q| ,
\end{equation}
Let $\theta  = t^* |Q|$, then the containment similarity search problem is converted into finding record $X$
whose intersection size with the query $Q$ is not smaller than $\theta$, i.e., $|Q\cap X| \geq \theta $.

Therefore, we can directly apply the \kmv method introduced in Section~\ref{subsec:kmv}.
Given \kmv signatures of a record $X$ and a query $Q$, we can estimate their intersection size ($|Q \cap X|$) according to Equation~\ref{eq:kmvinterest}.
Then the containment similarity of $Q$ in $X$ is immediately available given the query size $|Q|$. 
Below, we show an example on how to apply \kmv method to containment similarity search.

\begin{figure}[t]
\centering
\vspace{0.2cm}
\begin{tabular}[b]{l l l}
\hline
{\bfseries \ \  } & {\bfseries $\mathcal{L}_{KMV}$} & $k_i$\\
\hline
$\mathcal{L}_{X_1}$ & $\{(e_2,0.24), (e_7,0.33), (e_4,0.47)\}$  & $3$\\

$\mathcal{L}_{X_2}$ & $\{(e_5,0.10), (e_2,0.24), (e_3,0.85)\}$  & $3$\\

$\mathcal{L}_{X_3}$ & $\{(e_5,0.10), (e_2,0.24)\}$ & $2$\\

$\mathcal{L}_{X_4}$ & $\{(e_{10},0.18), (e_2,0.24)\}$ & $2$\\

$\mathcal{L}_{Q}$ & $\{(e_5,0.10), (e_2,0.24), (e_7,0.33), (e_9,0.56)\}$ & $4$ \\
\hline
\end{tabular}
\caption{\small The \kmv sketch of the dataset in Example~\ref{fig:problem example1}, each signature consists of element-hash value pairs. $k_i$ is the signature size of $X_i$}
\label{fig:kmvsketch}
\vspace{-1mm}
\end{figure}

\begin{example}
\label{exam:kmv}
Fig.~\ref{fig:kmvsketch} shows the \kmv sketch on dataset in Example~\ref{exam:prob_def}.
Given \kmv signature of $Q$ ($\mathcal{L}_{Q}=\{(e_5,0.10), (e_2,0.24), (e_7,0.33), e_9,0.56)\}$) and $X_1$ ($\mathcal{L}_{X_1} =\{(e_2,0.24), (e_7,0.33), (e_4,0.47)\}$),
we have $k=\min\{k_Q, k_1\} = 3$, then the size-$k$ \kmv synopses of $Q\cup X_1$ is $\mathcal{L} = \mathcal{L}_{Q} \oplus \mathcal{L}_{X_1} =  \{(e_5,0.10), (e_2,0.24), (e_7,0.33)\}$,
the $k$-th smallest hash value $U_{(k)}$ is 0.33 and
 the size of intersection of $\mathcal{L}_{Q}$ and $\mathcal{L}_{X_1}$ within $\mathcal{L}$ is $K_{\cap} = |\{v: v\in \mathcal{L}_{Q}\cap \mathcal{L}_{X_1}, v\in \mathcal{L}\}|=2$. Then the intersection size of $Q$ and $X_1$ is estimated as $\hat{D}_\cap = \frac{K_\cap}{k}\times\frac{k-1}{U_{(k)}} = \frac{2}{3}*\frac{2}{0.33} = 4.04$, and the containment similarity is  $\hat{t} = \frac{\hat{D}_{\cap}}{|Q|} = 0.67$.
Then $X_1$ is returned if the given containment similarity threshold $t^*$ is $0.5$.
\end{example}

\begin{remark}
In~\cite{zhu2016lsh}, the size of the query is approximated by MinHash signature of $Q$, where \kmv sketch can also serve for the same purpose.
But the exact query size is used their implementation for performance evaluation. In practice, the query size is readily available,
we assume query size is given throughout the paper.
\end{remark}

\vspace{1mm}
\noindent \textbf{Optimization of KMV Sketch.} 
Given a space budget $b$, we can keep size-$k_i$ \kmv signatures (i.e., $k_i$ minimal hash values) for each record $X_i$ with $\sum_{i=1}^n k_i = b$.
A natural question is how to allocate the resource (e.g., setting of $k_i$ values)
to achieve the best overall estimation accuracy.
Intuitively, more resources should be allocated to records with more frequent elements or larger record size,
i.e., larger $k_i$ for record with larger size. 
However, \textbf{Theorem~\ref{theo:kmv sig}} (Section~\ref{subsubsec:optimalKMV}) suggests  that, the optimal resource allocation strategy in terms of estimation variance
is to use the same size of signature  for each record. This is because the \textbf{minimal} of two k-values is used in Equation~\ref{eq:kmv_1}, and hence the best solution is to \textbf{evenly} 
allocate the resource. Thus, we have the \kmv sketch based method for approximate containment similarity search.
For the given budget $b$, we keep $k_i=$ $\lfloor \frac{b}{m} \rfloor$ minimal hash values for each record $X_i$.

\vspace{1mm}
\noindent \textbf{(2) Impose a Global Threshold to KMV Sketch (\gkmv)}

The above analysis on optimal \kmv sketch suggests an equal size allocation strategy, that is, each record is associated with the same size signature. Intuitively we should assign more resources (i.e., signature size) to the records with large size because they are more likely to appear in the results.
However, the estimate accuracy of KMV for two sets size intersection is determined by the sketch with smaller size
since we choose $k = \min(k_1, k_2)$ for \kmv signatures of $X_1$ and $X_2$ for $D_\cup$ and $D_\cap$ in Equation~\ref{eq:kmvunionest}, thus it is useless to give more resource to one of the records.
We further explain the reason behind with the following example.

Before we introduce the global threshold to \kmv sketch,
consider the \kmv sketch shown in the Fig.~\ref{fig:kmvsketch}.

\begin{example}
\label{exm:gkmb_motivation}
Suppose we have $\mathcal{L}_{Q} =$ $\{(e_5,0.10), (e_2,0.24), (e_7,0.33), (e_9,0.56)\}$
and  $\mathcal{L}_{X_3} =$ $\{(e_5,0.10), (e_2,0.24)\}$  .
Although there are four hash values in $\mathcal{L}_Q \cup \mathcal{L}_{X_3} = \{(e_5,0.10), (e_2,0.24), (e_7,0.33), (e_9,0.56)\}$,
we can only consider $k=\min\{k_Q, k_{X_3}\}=2$
smallest hash values of $\mathcal{L}_Q \cup \mathcal{L}_{X_3}$ by Equation~\ref{eq:kmv_1},
which is $\{(e_5,0.10), (e_2,0.24)\}$,
and the $k$-th ($k = 2$) minimum hash value used in Equation~\ref{eq:kmvunionest} is $0.24$.
We cannot use $k=4$ (i.e., $U_{(k)}$=$0.56$)
to estimate $|Q \cup X_3|$
because the $4$-th smallest hash value in $\mathcal{L}_Q \cup \mathcal{L}_{X_3}$
\textbf{may not be} the $4$-th smallest hash values in $h(Q \cup X_3)$,
because the unseen $3$-rd smallest hash value of $X_3$ might be $(e_4, 0.47)$ for example,
which is smaller than $0.56$.
Recall that $h(Q \cup X_3)$ denote the hash values of all elements
in  $Q \cup X_3$.

Nevertheless, if we know that all the hash values \emph{smaller than a global threshold}, say $0.6$,
are kept for \emph{every} record, we can safely use the $4$-th hash value of $\mathcal{L}_Q \cup \mathcal{L}_{X_3}$ (i.e., $0.56$) for the estimation.
This is because we can ensure
the $4$-th smallest hash value in $\mathcal{L}_Q \cup \mathcal{L}_{X_3}$
\textbf{must be} the $4$-th smallest hash values in $h(Q \cup X_3)$.

\end{example}
Inspired by the above observation, we can carefully choose a global threshold $\tau$ (e.g., $0.6$ in the above example)
for a given space budget $b$,
and ensure all hash values smaller than $\tau$ will be kept for \kmv sketch of the records.
By imposing a global threshold, we can identify a better (i.e., larger)
$k$ value used for estimation, compared with Equation~\ref{eq:kmv_1}.

Given a record $X$ and a global threshold $\tau$,
the sketch of a record $X$ is obtained as $\mathcal{L}_X =\{h(e): h(e)\leq \tau, e\in X\}$
where $h$ is the hash function.
The sketch of $Q$ ($\mathcal{L}_Q$) is defined in the same way.
In this paper, we say a \kmv sketch is a \gkmv sketch if we impose a global threshold to generate \kmv sketch. 
Then we set $k$ value of the \kmv estimation as follows.
\begin{equation}
\label{eq:kmv_global}
 k = |\mathcal{L}_Q \cup \mathcal{L}_X|
\end{equation}
Meanwhile, we have $K_{\cap} = |\mathcal{L}_Q \cap \mathcal{L}_X|$.
Let $U_{(k)}$ be the $k$-th minimal hash value in $\mathcal{L}_Q \cup \mathcal{L}_X$, then the overlap size of  $Q$ and $X$ can be estimated as
\begin{equation}
\label{eq: gkmv compute}
  \hat{D}_{\cap}^{GKMV} = \frac{K_{\cap}}{k} \frac{k-1}{U_{(k)}}
\end{equation}
Then the containment similarity of $Q$ in $X$ is
\begin{equation}
\label{eq: gkmv estimator}
  \hat{C} = \frac{\hat{D}_{\cap}^{GKMV}}{q}
\end{equation}
where $q$ is the query size.  We remark that, as a by-product, the global threshold favours the record with large size
because all elements with hash value smaller than $\tau$ are kept for each record.

\begin{figure}[t]
\centering
\vspace{0.2cm}
\begin{tabular}[b]{l l}
\hline
{\bfseries \ \  } & {\bfseries $\mathcal{L}_{GKMV}$} \\
\hline
$\mathcal{L}_{X_1}$ & $\{(e_2,0.24), (e_7,0.33), (e_4,0.47)\}$\\

$\mathcal{L}_{X_2}$ & $\{(e_5,0.10), (e_2,0.24)\}$ \\

$\mathcal{L}_{X_3}$ & $\{(e_5,0.10), (e_2,0.24), (e_4, 0.47)\}$ \\

$\mathcal{L}_{X_4}$ & $\{(e_{10},0.18), (e_2,0.24)\}$ \\

$\mathcal{L}_{Q}$ & $\{(e_5,0.10), (e_2,0.24), (e_7,0.33)\}$  \\
\hline
\end{tabular}
\caption{\small The \gkmv sketch of the dataset in Example~\ref{fig:problem example1} with hash value threshold $\tau =0.5$}
\label{fig:gkmv sketch}
\vspace{-1mm}
\end{figure}

Below is an example on how to compute the containment similarity based on \gkmv sketch.
\begin{example}
\label{exam:gkmv}
Fig. ~\ref{fig:gkmv sketch} shows the \kmv sketch  of dataset in Example~\ref{fig:problem example1} with a global threshold $\tau = 0.5$. Given the signature of $Q$($\mathcal{L}_{Q} = \{(e_5,0.10), (e_2,0.24), (e_7,0.33)\}$) and $X_1$($\mathcal{L}_{X_1} = \{(e_2,0.24), (e_7,0.33),(e_4, 0.47)\}$), the \kmv sketch of $Q\cup X_1$ is $\mathcal{L} = \mathcal{L}_Q \cup \mathcal{L}_{X_1} = \{(e_5,0.10),(e_2,0.24), (e_7,0.33), (e_4,0.47)\}$, the $k$-th($k=4$) smallest hash value is $U_{(k)} = 0.47$, and the size of intersection of $\mathcal{L}_{Q}$ and $\mathcal{L}_{X_1}$ within $\mathcal{L}$ is $K_{\cap} = |\{v: v\in \mathcal{L}_{Q}\cap \mathcal{L}_{X_1}, v\in \mathcal{L}\}|=2$. Then the intersection size of $Q$ and $X_1$ is estimated as $\hat{D}_\cap = \frac{K_\cap}{k}\times\frac{k-1}{U_{(k)}} = \frac{2}{4}*\frac{3}{0.47} = 3.19$, and the containment similarity is  $\hat{t} = \frac{\hat{D}_{\cap}}{|Q|} = 0.53$.
Then $X_1$ is returned if the given containment similarity threshold $t^*$ is $0.5$.
\end{example}

\vspace{1mm}
\noindent \textbf{Correctness of \gkmv sketch}. \textbf{Theorem~\ref{theo: gkmv valid}} in Section~\ref{subsubsec:correct_gkmv} shows the correctness of the \gkmv sketch. 

\vspace{1mm}
\noindent \textbf{Comparison with \kmv}. 
In \textbf{Theorem~\ref{theo: gkmv better}} (Section~\ref{subsubsec:better_kmv}), we theoretically show that \gkmv can achieve better accuracy compared with \kmv.

\begin{remark}
Note that the global threshold technique cannot be applied to MinHash based techniques.
In minHash \lsh, the $k$ minimum hash values are corresponding to $k$ different independent hash functions,
while in \kmv sketch, the $k$-value sketch is obtained under one hash function. 
Thus we can only impose this global threshold on the same hash function for the \kmv sketch based method.
\end{remark}

\vspace{1mm}
\noindent \textbf{(3) Use Buffer for KMV Sketch (\gbkmv)}

In addition to the skewness of the record size, it is also worthwhile to exploit the skewness of the element frequency.
Intuitively, more resource should be assigned to high-frequency elements because they are more likely to appear in the records.
However, due to the nature of the hash function used by \kmv sketch,
the hash value of an element is \emph{independent} to its frequency; that is, all elements have the same opportunity contributing to the \kmv sketch.

One possible solution is to divide the elements into multiple disjoint groups according to their frequency (e.g., low-frequency and high-frequency ones), and then apply \kmv sketch for each individual group.
The intersection size between two records $Q$ and $X$ can be computed within each group and then sum up together.
However, our initial experiments suggest that this will lead to poor accuracy because of the summation of the intersection size estimations. In Theorem~\ref{theo: kmv any cut} (Section~\ref{subsubsec:partition_notgood}),
our theoretical analysis suggests that the combination of estimated results are very likely to make the overall accuracy worse.

To avoid combining multiple estimation results, we use a bitmap buffer with size $r$ for each record to \emph{exactly} keep track of the $r$ most frequent elements, denoted by $\mathcal{E}_H$. Then we apply \gkmv technique to the remaining elements,
resulting in a new augmented sketch, namely \gbkmv. 
Now we can estimate $|Q \cap X|$ by combining the intersection of their bitmap buffers (\emph{exact solution}) and \kmv sketches (\emph{estimated solution}). 

As shown in Fig.~\ref{fig:gbkmv sketch}, suppose we have $\mathcal{E}_{H} = \{e_1, e_2 \}$ and the global threshold for hash value is $\tau = 0.5$, then the sketch of each record consists of two parts  $\mathcal{L}_{H}$ and $\mathcal{L}_{GKMV}$;
that is, for each record we use bitmap to keep the elements corresponding to high-frequency elements $\mathcal{E}_H = \{e_1, e_2 \}$, then we store the left elements with hash value less than $\tau = 0.5$.
\begin{figure}[!t]
\centering
\vspace{0.2cm}
\begin{tabular}[b]{l l l}
\hline
{\bfseries \ \ } & {\bfseries $\mathcal{L}_{H}$ } & {\bfseries $\mathcal{L}_{GKMV}$}\\
\hline
$X_1$ & $\{e_1, e_2\}$ & $\{(e_7,0.33), (e_4,0.47)\}$  \\

$X_2$ & $\{e_2\}$ & $\{(e_5,0.10)\}$  \\

$X_3$ & $\{e_2\}$ & $\{(e_5,0.10)\}$ \\

$X_4$ &$\{e_1, e_2\}$ & $\{e_{10},0.18)\}$\\

$Q$ & $\{e_1, e_2\}$ & $\{(e_5,0.10), (e_7,0.33)\}$ \\
\hline
\end{tabular}
\caption{\small The \gbkmv sketch of dataset in Example~\ref{fig:problem example1}}
\label{fig:gbkmv sketch}
\vspace{-1mm}
\end{figure}

\begin{example}
\label{exam:gbkmv}
Given the signature of $Q$($\mathcal{L}_{Q} = \{e_1, e_2\} \cup \{(e_5,0.10), (e_7,0.33)\}$) and $X_1$($\mathcal{L}_{X_1} = \{e_1, e_2\} \cup \{(e_7,0.33), (e_4, 0.47)\}$), the intersection of High-frequency part is $\mathcal{L}_{Q}^{H} \cap \mathcal{L}_{X_1}^{H} = \{e_1, e_2\}$ with intersection size as $2$; next we consider the \gkmv part. Similar to Example ~\ref{exam:gkmv}, we compute the intersection of $\mathcal{L}_{GKMV}$ part.
The \kmv sketch  is $\mathcal{L}' = \mathcal{L}'_Q \cup \mathcal{L}'_{X_1} = \{(e_5,0.10),(e_7,0.33), (e_4,0.47)\}$.
According to Equation~\ref{eq:kmv_global}, the $k$-th($k=3$) smallest hash value is $U_{(k)} = 0.47$, and the size of intersection of $\mathcal{L}_{Q}$ and $\mathcal{L}_{X_3}$ within $\mathcal{L}$ is $K_{\cap} = |\{v: v\in \mathcal{L}_{Q}\cap \mathcal{L}_{X_3}, v\in \mathcal{L}\}|=1$. Then the intersection size of $Q$ and $X_1$ in $\mathcal{L}_{GKMV}$ part is estimated as $\hat{D}_\cap = \frac{K_\cap}{k}\times\frac{k-1}{U_{(k)}} = \frac{1}{3}*\frac{2}{0.47} = 1.4$; together with the High-frequency part, the intersection size of $Q$ and $X_1$ is estimated as $2+1.4=3.4$ and the containment similarity is  $\hat{t} = \frac{\hat{D}_{\cap}}{|Q|} = 0.53$.
Then $X_1$ is returned if the given containment similarity threshold $t^*$ is $0.5$.
\end{example}

\vspace{1mm}
\noindent \textbf{Optimal Buffer Size.}  The key challenge is how to set the size of bitmap buffer for the best expected performance of \gbkmv sketch. In Section~\ref{subsubsec:choose_r}, we provide a theoretical analysis,
which is verified in our performance evaluation. 

\vspace{1mm}
\noindent \textbf{Comparison with \gkmv}.
As the \gkmv is a special case of \gbkmv with buffer size $0$ and we carefully choose the buffer size with our cost model, the accuracy of \gbkmv is \textbf{not worse} than \gkmv.

\vspace{1mm}
\noindent \textbf{Comparison with \lshe}.
Through theoretical analysis, we show that the performance (i.e., the variance of the estimator)
of \gbkmv can always outperform that of \lshe in \textbf{Theorem~\ref{theo:gbkmvBetter}} (Section~\ref{subsubsec:gkmv_better}).

\subsection{\textbf{Implementation of GB-KMV}}
\label{subsec:alg}

In this section, we introduce the technique details of our proposed \gbkmv method. We first show how to build \gbkmv sketch on the dataset $\mathcal{S}$ and
then present the containment similarity search algorithm.

\vspace{1mm}
\noindent \textbf{GB-KMV Sketch Construction}. 
For each record $X \in \mathcal{S}$, its \gbkmv sketch consists of two components:
(1) a buffer which exactly keeps high-frequency elements, denoted by $\mathcal{H}_X$;
and (2) a \gkmv sketch, which is a \kmv sketch with a global threshold value,
denoted by $\mathcal{L}_X$.

\begin{algorithm}[hbt]
\SetVline 
\SetFuncSty{textsf}
\SetArgSty{textsf}
\small
\caption{\textbf{\gbkmv Index Construction}}
\label{alg:gbkmv index}
\Input
{
\ds: dataset; $b$: space budget; \\
$h$: a hash function; $r$: buffer size
}
\Output
{
    $\mathcal{L}_{\mathcal{S}}$, the \gbkmv index of dataset \ds
}
\State{Compute buffer size $r$ based on distribution statistics of \ds and the space budget $b$}
\label{alg:index_c_r}
\State{$\mathcal{E}_H$ $\leftarrow$ Top $r$ most frequent elements; $\mathcal{E}_K \leftarrow $ $\mathcal{E} \setminus \mathcal{E}_H$ }
\label{alg:index_c_h}
\State{$\tau \leftarrow $ compute the global threshold for hash values}
\label{alg:index_c_t}
\ForEach{record $X \in \mathcal{S}$}
{
    \State{$\mathcal{H}_X  \leftarrow $ elements of $X$ in $\mathcal{E}_H$ }
    \label{alg:index_h}
    \State{ $\mathcal{L}_X  \leftarrow$ hash values of elements $\{e\}$  of $X$
     with $h(e) \leq \tau$}
     \label{alg:index_k}
}
\end{algorithm}
Algorithm~\ref{alg:gbkmv index} illustrates the construction of \gbkmv sketch.
Let the element universe be $\mathcal{E} = \{e_1, e_2,..., e_n\}$ and each element is associated with its frequency in dataset \ds. Line~\ref{alg:index_c_r} calculates a buffer size $r$ for all records based on the skewness of record size and elements as well as the space budget $b$ in terms of elements. Details will be introduced in Section~\ref{subsubsec:choose_r}.
We use $\mathcal{E}_H$ to denote the set of top-$r$ most frequent elements (Line~\ref{alg:index_c_h}),
and they will be kept in the buffer of each record.
Let $\mathcal{E}_K$ denote the remaining elements. Line~\ref{alg:index_c_t} identifies maximal possible global threshold $\tau$ for elements in $\mathcal{E}_K$ such that the total size of \gbkmv sketch meets the space budget $b$.
For each record $X$, let $n_X$ denote the number of elements in $\mathcal{E}_K$ with hash values less than $\tau$,
we have $\sum_{X \in \mathcal{S}} (\frac{r}{32} + n_X) \leq b$.
Then Lines~\ref{alg:index_h}-\ref{alg:index_k} build the buffer $\mathcal{H}_X$ and \gkmv sketch
$\mathcal{L}_X$ for every record $X \in \mathcal{S}$.
In section~\ref{theo: gkmv valid}, we will show the correctness of our sketch in Theorem~\ref{theo: gkmv valid}.

\vspace{1mm}
\noindent \textbf{Containment Similarity Search.}
Given the \gbkmv sketch of the query record $Q$ and the dataset $\mathcal{S}$,
we can conduct approximate similarity search as illustrated in Algorithm~\ref{alg:search}.
Given a query $Q$ with size $q$ and the similarity threshold $t^*$,
let $\theta = t^* * q$(Lines 1-2).
With \gbkmv sketch $\{ \mathcal{H}_Q, \mathcal{L}_Q \}$,
we can calculate the containment similarity based on
\begin{equation}\label{eq: gbkmv compute}
\widehat{|Q \cap X|} = |\mathcal{H}_Q \cap \mathcal{H}_X| + \hat{D}_{\cap}^{GKMV}
\end{equation}
where $\hat{D}_{\cap}^{GKMV}$ is the estimation of overlap size of $Q$ and $X$
which is calculated by Equation~\ref{eq: gkmv compute} in Section~\ref{subsec:alg_motivation}.

Note that $|\mathcal{H}_Q \cap \mathcal{H}_X|$ is the number of common elements of $Q$ and $X$ in $\mathcal{E}_H$.
\begin{algorithm}[hbt]
\SetVline 
\SetFuncSty{textsf}
\SetArgSty{textsf}
\small
\caption{\textbf{Containment Similarity Search}}
\label{alg:search}
\Input
{
$Q$, a query set\\
$t^{*}$, containment similarity threshold
}
\Output{$R:$ records $\{X\}$ with $C(Q,X) \geq t^{*}$}
\State{$q$ $\leftarrow$ $|Q|$}
\State{$\theta$ $\leftarrow$ $t^* * q $}
\ForEach{record $X$ $\in$ $\mathcal{S}$}
{
    \State{ $\widehat{|Q\cap X|}$ $\leftarrow$ $|\mathcal{L}^Q_H \cap \mathcal{L}^X_H| + \hat{D}_{\cap}^{GKMV}$}
    \label{alg:search_sim}
    \If{$\widehat{|Q\cap X|} \geq \theta$}{
        \State{$\mathcal{S}_{candidate}= \mathcal{S}_{candidate} \cup X$}
    }
}
\Return{$\mathcal{S}_{candidate}$}
\end{algorithm}

\vspace{1mm}
\noindent \textbf{Implementation of Containment Similarity Search.}
In our implementation, we use a bitmap with size $r$ to keep the elements in buffer
where each bit is reserved for one frequent element.
We can use \emph{bitwise intersection} operator to efficiently compute $|\mathcal{H}_Q \cap \mathcal{H}_X|$
in Line~\ref{alg:search_sim} of Algorithm~\ref{alg:search}.
Note that the estimator of overlap size by \gkmv method in Equation~\ref{eq: gkmv compute} is
$\hat{D}_{\cap}^{GKMV} = \frac{K_{\cap}}{k} \frac{k-1}{U_{(k)}}$.
As to the computation of $\widehat{|Q\cap X|}$,
we apply some transformation to $|\mathcal{L}^Q_H \cap \mathcal{L}^X_H| + \hat{D}_{\cap}^{GKMV} \geq \theta$.
Then we get $K_{\cap} \geq o $ where $o = U_{(k)}(\theta - o_1)$ and $o_1 = |\mathcal{H}_Q \cap \mathcal{H}_X|$. Since $K_{\cap}$ is the overlap size, then we make use of the PPjoin*~\cite{xiao2011efficient} to speed up the search. 
Note that in order to make the PPjoin* which is designed for similarity join problem to be applicable to the similarity search problem, 
we partition the dataset \ds by record size, and in each partition we search for the records which satisfy $K_{\cap} \geq o$, 
where overlap size is modified by the lower bound in corresponding partition.

\begin{remark}
Note that the size-aware overlap set similarity joins algorithm in~\cite{deng2018overlap} can not be applied to our \gbkmv method, 
because we need to online construct $c$-subset inverted list for each incoming query, which results in very inefficient performance.
\end{remark}

\noindent \textbf{Processing Dynamic Data.}
Note that our algorithm can be modified to process dynamic data.
Particularly, when new records come, we compute the new global threshold $\tau$ under the fixed space budget by Line ~\ref{alg:index_c_t} of Algorithm ~\ref{alg:gbkmv index},
and with the new global threshold, we maintain the sketch of each record as shown in Line ~\ref{alg:index_k} of Algorithm ~\ref{alg:gbkmv index}.

\subsection{Theoretical Analysis}
\label{subsec:analysis}
In this section, we provide theoretical underpinnings of the claims and observations in this paper.
\subsubsection{\textbf{Background}}
\label{subsubsec:the_assumption}
We need some reasonable assumptions on the record size distribution, element frequency distribution and query work-load
for a comprehensive analysis.
Following are three popular assumptions widely used in the literature (e.g.,~\cite{albert1999r, jeong2000large,goldstein2004problems,clauset2009power,cho2011friendship, zhu2016lsh,shrivastava2014asymmetric}):
\begin{itemize}
  \item The element frequency in the dataset follows the power-law distribution, with $p_1(x) = c_1 x^{-\alpha_1}$.
  \item The record size in the dataset follows the power-law distribution, with  $p_2(x) = c_2 x^{-\alpha_2}$.
  \item The query $Q$ is randomly chosen from the records.
\end{itemize}
Throughout the paper, we use the variance to evaluate the goodness of an estimator.
Regarding the \kmv based sketch techniques (\kmv, \gkmv and \gbkmv),
we have
\begin{lemma}
\label{lem:eff sig size}
In \kmv sketch based methods, the larger the $k$ value used in Equation~\ref{eq:kmv_1}
and Equation~\ref{eq:kmv_global} is, the smaller the variance will be.
\end{lemma}
It is easy to verify the above lemma by calculating the derivative of Equation~\ref{eq:kmvintervar} with respect to the variable $k$. Thus, in the following analysis of \kmv based sketch techniques. We use the $k$ value (i.e., the sketch size used for estimation) to evaluate the goodness of the estimation, the larger the better.

\subsubsection{\textbf{Optimal KMV Signature Scheme}}
\label{subsubsec:optimalKMV}
In this part, we give an optimal resource allocation strategy for \kmv sketch
method in similarity search.
\begin{theorem}
\label{theo:kmv sig}
Given a space budget $b$, each set is associated with a size-$k_i$ \kmv signature and $\sum_{i=1}^{m} k_i = b$.
For \kmv sketch based containment similarity search, the optimal signature scheme is to keep the $\lfloor \frac{b}{m} \rfloor$ minimal hash values for each set $X_i$.
\end{theorem}
\begin{proof}
  Given a query $Q$ and dataset $\mathcal{S}=\{X_1,...,X_m\}$, an optimal signature scheme for containment similarity search is to minimize the average variance between $Q$ and $X_i, i=1,...,m$. Considering the query $Q$ and set $X_i$ with size-$k_q$ \kmv sketch $\mathcal{L}_Q$ and size-$k_i$ sketch $\mathcal{L}_{X_i}$ respectively, the sketch size is $k = \min\{k_q, k_i\}$ according to Equation~\ref{eq:kmv_1}.  By Lemma~\ref{lem:eff sig size},
an optimal signature scheme is to maximize the total $k$ value(say $T$), then we have the following optimization goal,
  $$\max\  T = \sum_{i=1}^{m} \min\{k_q, k_i\}$$
  $$s.t. \ \ b = \sum_{i=1}^{m} k_i,\ \  k_i>0,i=1,2,...,m$$
  Rank the $k_i$ by increasing order, w.l.o.g., let $k_1, k_2,...,k_m$ be the sketch size sequence after reorder. Let $k_l$ be the first in the sequence such that $k_l = k_q$, then we have $T = k_1+...+k_l + (m-l)k_q = b-\sum_{i=l+1}^{m}(k_i-k_q)$.
  In order to maximize $T$, we set $k_i = k_q, i=l+1,...,m$. Then by $b = \sum_{i=1}^{m} k_i$, we have $k_1 + ...+k_l + k_q(m-l)=b$. Note that $k_i\leq k_q, i=1,...,l$, we must have $k_i = k_q, i=1,...,l$.
  Since $Q$ is randomly selected from dataset \ds, we can get that all the $k_i, i=1,...m$ are equal and $k_i = \lfloor \frac{b}{m} \rfloor$.
\end{proof}


\subsubsection{\textbf{Correctness of GKMV Sketch}}
\label{subsubsec:correct_gkmv}

In this section, we show that the \gkmv sketch is a valid \kmv sketch.

%
\begin{theorem}
\label{theo: gkmv valid}
Given two records $X$ and $Y$, let $\mathcal{L}_X$ and $\mathcal{L}_Y$ be the \gkmv sketch of $X$ and $Y$, respectively. Let $k=|\mathcal{L}_X \cup \mathcal{L}_Y|$, then the size-$k$ \kmv synopses of $X\cup Y$ is $\mathcal{L} = \mathcal{L}_X \cup \mathcal{L}_Y$.
\end{theorem}

\begin{proof}
 We show that the above $\mathcal{L} = \mathcal{L}_X \cup \mathcal{L}_Y$ is a valid \kmv sketch of $X\cup Y$. Let $k=|\mathcal{L}_X \cup \mathcal{L}_Y|$ and $v_k$ is the $k$-th smallest hash value in $\mathcal{L}_X \cup \mathcal{L}_Y$. In order to prove that $\mathcal{L}_X \cup \mathcal{L}_Y$ is valid, we show that $v_k$ corresponds the element with the $k$-th minimal hash value in $X\cup Y$. If not, there should exist an element $e$ such that $h(e')<v_k, e'\in X\cup Y$ and $h(e')\notin \mathcal{L}_X \cup \mathcal{L}_Y$. Note that $v_k\leq \tau$, then $h(e')\leq \tau$, thus $h(e')$ is included in $\mathcal{L}_X \cup \mathcal{L}_Y$, which contradicts to the above statement.
\end{proof}

\subsubsection{\textbf{G-KMV: A Better KMV Sketch}}
\label{subsubsec:better_kmv}

In this part, we show that by imposing a global threshold to \kmv sketch, we can achieve better accuracy.
Let $\mathcal{L}_X^{KMV}$ and $\mathcal{L}_Y^{KMV}$ be the \kmv sketch of $X$ and $Y$ respectively.
Let $k_1 = |\mathcal{L}_X^{KMV}|$ and $k_2 = |\mathcal{L}_Y^{KMV}|$,
then the sketch size $k$ value can be set by Equation~\ref{eq:kmv_1}.
Similarly, let $\mathcal{L}_X^{GKMV}$ and $\mathcal{L}_Y^{GKMV}$ be the \gkmv sketch of $X$ and $Y$ respectively, and the sketch size $k$ value can be set by Equation~\ref{eq:kmv_global}.

\begin{theorem}
\label{theo: gkmv better}
With the fixed index space budget, for containment similarity search the \gkmv sketch method is better than \kmv method in terms of accuracy
when the power-law exponent of element frequency $\alpha_1 \leq 3.4$.
\end{theorem}
\begin{proof}
Let $x_j = |X_j|, j=1,2,...,m$ be the set size
and $k_j$ be the signature size of record $X_j$.
The frequency of element $e_i$ is set to be $f_i$.
The index space budget is $b$.

For \kmv sketch based method, by Theorem~\ref{theo:kmv sig}, the optimal signature scheme is $k= min(k_j, k_l) =  \lfloor \frac{b}{m} \rfloor$ given the index space budget $b$,
then the average $k$ value for all pairs of sets is
\begin{equation}
\label{eq:kmv_k_value}
\bar{k}_{KMV} = \frac{1}{m^2} \sum_{j=1}^{m}\sum_{l=1}^{m} min(k_j, k_l) = \lfloor \frac{b}{m} \rfloor
\end{equation}

For \gkmv sketch based method, let $\tau$ be the hash value threshold.
The probability that hash value $h(e_i)$ is included in signature $\mathcal{L}^{GKMV}_{X_j}$ is
$Pr[h(e_i)\in \mathcal{L}^{GKMV}_{X_j}] = \tau \frac{f_i}{N}x_j$
where $f_i$ is the frequency of element $e_i$, and $N=\sum_{i=1}^{n}f_i$ is the total number of elements.
The size of $\mathcal{L}^{GKMV}_{X_j}$ can be computed by
$l_j = \sum_{i=1}^{n} Pr[h(e_i)\in \mathcal{L}^{GKMV}_{X_j}]  = \tau x_j$
then the total index space is  $b = \sum_{j=1}^{m} l_j = \sum_{j=1}^{m} \tau x_j  = \tau N$.
and the hash value threshold $\tau = \frac{b}{N}$.
Next we compute average sketch size $k$ value of \gkmv method.
The intersection size of $\mathcal{L}_{X_j}$ and $\mathcal{L}_{X_l}$
\begin{equation}
\label{eq:sigInterSize}
|\mathcal{L}_{X_j} \cap \mathcal{L}_{X_l}| = \sum_{i=1}^{n} \tau \frac{f_i}{N}x_j *  \tau \frac{f_i}{N}x_l = \tau^2 x_j x_l f_{n^2}
\end{equation}
where $f_{n^2} = \frac{\sum_{i=1}^{n}f_i^2}{N^2}$.
The $k$ value of \gkmv method according to Equation~\ref{eq:kmv_global} is
\begin{equation}
\label{eq:sigUnionSize}
|\mathcal{L}_{X_j} \cup \mathcal{L}_{X_l}| = \tau x_j + \tau x_l - \tau^2 x_j x_l f_{n^2}
\end{equation}
Then the average $k$ value for all pairs of sets is
\begin{equation}
\bar{k}_{GKMV} = \frac{1}{m^2}\sum_{j=1}^{m}\sum_{l=1}^{m} |\mathcal{L}_{X_j} \cup \mathcal{L}_{X_l}| = \frac{2b}{m} - \frac{b^2}{m^2} f_{n^2}
\end{equation}
Let $\bar{k}_{GKMV}\geq \bar{k}_{KMV}$, we get $\alpha_1\in(0, 0.5] \cup [(1+\frac{m}{b}) - \sqrt{(1+\frac{m}{b})\frac{m}{b}}, (1+\frac{m}{b}) + \sqrt{(1+\frac{m}{b})\frac{m}{b}}]$.
Note that for the common setting $\frac{m}{b} \leq 1 $, we can get $\alpha_1 \leq 3.4$. The result makes sense since the power-law(Zipf's law) exponent of element frequency is usually less than 3.4 for real datasets.
\end{proof}

\subsubsection{\textbf{Partition of KMV Sketch Is Not Promising}}
\label{subsubsec:partition_notgood}

In this part, we show that it is difficult to improve the performance of \kmv by
dividing elements to multiple groups according to their frequency and apply \kmv estimation individually.
W.l.o.g., we consider dividing elements into two groups.

We divide the sorted element universe $\mathcal{E}$ into two disjoint parts $\mathcal{E}_{H_1}$ and $\mathcal{E}_{H_2}$. Let $X$ and $Y$ be two sets from dataset \ds with \kmv sketch $\mathcal{L}_X$ and $\mathcal{L}_Y$ respectively. Let $k_{X} = |\mathcal{L}_X|$ and $k_{Y} = |\mathcal{L}_Y|$.
The estimator of containment similarity is
$\hat{C} = \frac{\hat{D}_{\cap}}{q}$,
where $\hat{D}_{\cap}$ is the estimator of intersection size $D_{\cap}$ and $q$ is the query size($x$ or $y$).

Corresponding to $\mathcal{E}_{H_1}$ and $\mathcal{E}_{H_2}$, we divide $X$($Y$, resp.) to two parts $X_1$ and $X_2$($Y_1$ and $Y_2$, resp.). We know that $X_1\cap X_2=\Phi$ and $Y_1\cap Y_2 = \Phi$.
Also, let $D_{\cap}=|X\cap Y|$,$D_{\cup} = |X\cup Y|$, we have $D_{\cap}=|X_1\cap Y_1| + |X_2\cap Y_2|$ and $D_{\cup} = |X_1\cup Y_1| + |X_2\cup Y_2|$ since $\mathcal{E}_{H_1}$ and $\mathcal{E}_{H_2}$ are disjoint.
For simplicity, let $D_{\cap 1}=|X_1\cap Y_1|$, $D_{\cup 1} = |X_1\cup Y_1|$, $D_{\cap 2}=|X_2\cap Y_2|$ and $D_{\cup 2} = |X_2\cup Y_2|$. For $X_1$, $X_2$, $Y_1$ and $Y_2$, the \kmv sketches are $\mathcal{L}_{X_1}$, $\mathcal{L}_{X_2}$, $\mathcal{L}_{Y_1}$ and $\mathcal{L}_{Y_2}$ with size $k_{X_1}$, $k_{X_2}$, $k_{Y_1}$ and $k_{Y_2}$, respectively.
Based on this, we give another estimator as
$\hat{C}' = \frac{\hat{D}_{\cap 1} + \hat{D}_{\cap 2}}{q}$,
where $\hat{D}_{\cap 1}$($\hat{D}_{\cap 2}$, resp.) is the estimator of intersection size $D_{\cap 1}$($D_{\cap 2}$, resp.).
 Next, we compare the variance of $\hat{C}$  and  $\hat{C}'$.
\begin{theorem}
\label{theo: kmv any cut}
After dividing the element universe into two groups and applying \kmv sketch in each group, with the same index space budget, the variance of $\hat{C}'$ is larger than that of $\hat{C}$.
\end{theorem}
\begin{proof}
  Recall the \kmv sketch, we have $E(\hat{C}') = E(\hat{D}_{\cap 1}) + E(\hat{D}_{\cap 2}) = D_{\cap 1} + D_{\cap 2} = D_{\cap}$. Because of the two disjoint element groups,  $\hat{D}_{\cap 1}$ and $\hat{D}_{\cap 2}$ are independent. Thus the variance
  $Var[\hat{C}'] = \frac{Var[\hat{D}_{\cap 1}] + Var[\hat{D}_{\cap 2}]}{q^2}$.
  Next, we will show
  $$Var[\hat{D}_{\cap 1}] + Var[\hat{D}_{\cap 2}] \geq Var[\hat{C}].$$
  Consider the \kmv sketch for set $X$ and $Y$, the sketch size according to Equation~\ref{eq:kmv_1} is $k=\min\{k_X, k_Y\}$. Similarly, for $X_1$ and $Y_1$, we have the sketch size $k_1 = \min\{k_{X_1}, k_{Y_1}\}$; for $X_2$ and $Y_2$, we have the sketch size $k_2 = \min\{k_{X_2}, k_{Y_2}\}$. Since the index is fixed, we have $k_X = k_{X_1} + k_{X_2}$ and $k_Y = k_{Y_1} + k_{Y_2}$. Then, $k_1 + k_2 = \min\{k_{X_1}, k_{Y_1}\} + \min\{k_{X_2}, k_{Y_2}\} \leq \min\{k_X, k_Y\} = k$.

  Let $\Delta =  Var[\hat{D}_{\cap 1}] + Var[\hat{D}_{\cap 2}] - Var[\hat{C}]$, after some calculation, we have
  $\Delta = \frac{D_{\cap 1}^2}{k_1^2} + \frac{D_{\cap 2}^2}{k_2^2} - \frac{D_{\cap}^2}{k^2} + \frac{D_{\cap 1} D_{\cup 1}}{k_1} + \frac{D_{\cap 2} D_{\cup 2}}{k_2} - \frac{D_{\cap} D_{\cup}}{k}$.
Next we show that $\frac{D_{\cap 1}^2}{k_1^2} + \frac{D_{\cap 2}^2}{k_2^2} - \frac{D_{\cap}^2}{k^2} \geq 0$.
Let $k_1 = \frac{1}{\alpha} k$ and $k_2 = \frac{1}{\beta}$ where $\frac{1}{\alpha} + \frac{1}{\beta} = 1$ and $\alpha,\ \beta >1$.
Then we have $\frac{D_{\cap 1}^2}{k_1^2} + \frac{D_{\cap 2}^2}{k_2^2} - \frac{D_{\cap}^2}{k^2} 
= \frac{\alpha^2 D_{\cap 1}^2 + \beta^2 D_{\cap 2}^2 - D_{\cap}^2}{k^2} = \frac{(\alpha^2 - 1) D_{\cap 1}^2 + (\beta^2 - 1) D_{\cap 2}^2 - 2 D_{\cap 1} D_{\cap 2}}{k^2}$.
As for the upper part in the above equation, by  inequality of arithmetic and geometric means, we get 
$(\alpha^2 - 1) D_{\cap 1}^2 + (\beta^2 - 1) D_{\cap 2}^2 - 2 D_{\cap 1} D_{\cap 2} \geq 2 (\sqrt{(\alpha - 1)(\alpha + 1)(\beta - 1)(\beta + 1)} - 1) D_{\cap 1} D_{\cap 2}$.
Since $(\alpha-1)(\beta - 1) = 1 $, we get $\sqrt{(\alpha - 1)(\alpha + 1)(\beta - 1)(\beta + 1)} - 1 = \sqrt{(\alpha + 1)(\beta + 1)} - 1 \geq 0$,
thus $\frac{D_{\cap 1}^2}{k_1^2} + \frac{D_{\cap 2}^2}{k_2^2} - \frac{D_{\cap}^2}{k^2} \geq 0$.

Let $ \Delta_1 = \frac{D_{\cap 1} D_{\cup 1}}{k_1} + \frac{D_{\cap 2} D_{\cup 2}}{k_2} - \frac{D_{\cap} D_{\cup}}{k}$, after some computation, we have $\Delta_1 = \frac{(k_1 D_{\cup 2} - k_2 D_{\cup 1})(k_1 D_{\cap 2} - k_2 D_{\cap 1})}{k k_1 k_2}$. As for the numerator(upper) of $\Delta _1$, consider the two parts after dividing the element universe, if the union size in one part, say $D_{\cup 2 }$, is larger, meanwhile the corresponding intersection size $D_{\cap 2}$ is larger, we have $\Delta_1 \geq 0$. This case can be realized since one of the two groups divided from element universe is made of high-frequency elements, which will result in large intersection size and large union size under the proper choice of $k, k_1, k_2$.
\end{proof}

\subsubsection{\textbf{Optimal Buffer Size $r$} }
\label{subsubsec:choose_r}

In this part, we show how to find optimal buffer size $r$ by analysing the variance for \gbkmv method.
Given the space budget $b$, we first show that the variance for \gbkmv sketch is a function of $f(r, \alpha_1, \alpha_2, b)$ and then we give a method to appropriately choose $r$.
Below are some notations first.

Given two sets $X$ and $Y$ with \gkmv sketch $\mathcal{L}_X$ and $\mathcal{L}_Y$ respectively, the containment similarity of $Q$ in $X$ is computed by Equation~\ref{eq: gkmv estimator} as
$\hat{C}^{GKMV} = \frac{\hat{D}_{\cap}^{GKMV}}{q}$,
where $\hat{D}_\cap^{GKMV} = \frac{K_\cap}{k}\times\frac{k-1}{U_{(k)}}$ is the overlap size.

As for the \gbkmv method of set $X$ an $Y$ with sketch $\mathcal{H}_X \cup \mathcal{L}_X$ and $\mathcal{H}_Y \cup \mathcal{L}_Y$ respectively, the containment similarity of $Q$ in $X$ is computed by Equation~\ref{eq: gbkmv compute} as
 $\hat{C}^{GBKMV} = \frac{|\mathcal{H}_Q \cap \mathcal{H}_X| + \hat{D}_{\cap}^{GKMV}}{q}$,
where $|\mathcal{H}_Q \cap \mathcal{H}_X|$ is the number of common elements in $\mathcal{E}_H$ part. It is easy to verify that
$\hat{C}_{GBKMV}$ is an unbiased estimator. Also, the variance of \gbkmv method estimator is
$Var[\hat{C}_{GBKMV}] = \frac{Var[\hat{D}_{\cap}^{GKMV}]}{q^2}$,
where $Var[\hat{D}_{\cap}^{GKMV}]$ corresponds to the variance of the \gkmv sketch in the \gbkmv sketch.

Next, with the same space budget $b$, we compute the average variance of \gbkmv method.

Consider the \gbkmv index construction which is introduced in Section~\ref{subsec:alg} by Algorithm~\ref{alg:gbkmv index}.
Let $N$ be the total number of elements and $b$ the space budget in terms of elements for index construction.
Assume that we keep $r$ high-frequency elements by bitmap in the buffer, which have $N_1 = \sum_{j=1}^{m}|\mathcal{H}_{X_j}| = \sum_{i=1}^{r} f_i$ elements and occupy $T_1 = m*r/32$ index space. Then the total number of elements left for \gkmv sketch is $N_2 = N - N_1$ and the index space for \gkmv sketch is $T_2 = b - T_1$.

Given two sets $X_j$ and $X_l$, the variance of overlap size estimator in Equation~\ref{eq:kmvintervar} is as follows
\begin{equation}
\label{eq:kmvVariacne}
Var[\hat{D}_\cap] = \frac{D_\cap(kD_\cup - k^2 - D_\cup + k + D_\cap)}{k(k-2)}
\end{equation}
where $D_{\cup}=|X_j\cup X_l|$, $D_{\cap}=|X_j\cap X_l|$ and $k$ is the sketch size.
Since the variance is concerned with the union size $D_{\cup}$, the intersection size $D_{\cap}$ and the signature size $k$, we first calculate these three formulas, then compute the variance.

Consider the two sets $X_j$, $X_l$ from dataset $\mathcal{S}$ with \gbkmv sketch $\mathcal{H}_{X_j}\cup \mathcal{L}_{X_j}$ and $\mathcal{H}_{X_j}\cup \mathcal{L}_{X_j}$ respectively.  The element $e_i$ is associated with frequency $f_i$, and the probability of element $e_i$ appearing in record $X_j$ is
$Pr[h(e_i)\in \mathcal{L}_{X_j}] =  \frac{f_i}{N}x_j$.
Given a hash value threshold $\tau$, the \gkmv signature size of set $X_j$ is computed as $k_j = \tau(x_j - |\mathcal{H}_{X_j}|)$.
The total index space in \gkmv sketch is $\sum_{j=1}^{m} k_j = T_2 = b - T_1 = b - \frac{r}{32}*m$, then we get $\tau = \frac{b-r/32*m}{N-N_1}$.

Similar to Equation~\ref{eq:sigInterSize},~\ref{eq:sigUnionSize}, the sketch size $k$ value for \gbkmv sketch is
$k = \tau(x_j + x_l)-\tau^2 x_1 x_2(f_{n^2} - f_{r^2})$
where $f_{n^2} = \frac{\sum_{i+1}^{n}f_i^2}{N^2}$, $f_{r^2} = \frac{\sum_{i+1}^{r}f_i^2}{N^2}$.
The intersection size and union size of $X_j$ and $X_l$ are
$D_{\cap}  =  x_j x_l (f_{n^2}- f_{r^2})$
and
$D_{\cup}  = (x_j + x_l)(1-f_r) - x_j x_l(f_{n^2} - f_{r^2})$
where $f_r = \frac{\sum_{i=1}^{r}f_i}{N}$, then the variance of \gbkmv method by Equation~\ref{eq:kmvVariacne} is
$$Var[\hat{C}_{GBKMV}] = \frac{(x_j + x_l)x_l}{k x_j} F_1 + \frac{x_l^2}{k} F_2 + \frac{x_l}{x_j} F_3$$
where $F_1 = f_{n^2}-f_{r^2}$, $F_2 = - (f_{n^2} - f_{r^2})^2$ and $F_3 = -(f_{n^2}-f_{r^2})$,
and the average variance of \gbkmv method $Var_{GBKMV} = \frac{1}{m^2} \sum_{j=1}^{m}\sum_{l=1}^{m}  Var[\hat{C}_{GBKMV}]$ is
$$Var_{GBKMV} = L_1 F_1 +  L_2 F_2 +  L_3 F_3$$
where $L_1 = \frac{1}{m^2}\sum_{j=1}^{m}\sum_{l=1}^{m} \frac{(x_j + x_l)x_j x_l}{kx_j^2}$,
$L_2 =\frac{1}{m^2} \sum_{j=1}^{m}\sum_{l=1}^{m}\frac{(x_j x_l)^2}{k x_j^2}$
and $L_3 = \frac{1}{m^2}\sum_{j=1}^{m}\sum_{l=1}^{m} \frac{x_l}{ x_j}$.

Note that $F_1, F_2, F_3$ is concerned with the element frequency which can be computed by using the distribution $p_1(x) = c_1 x^{-\alpha_1}$; $L_1, L_2, L_3$ is related to the record size which can be computed by using $p_2(x) = c_2 x^{-\alpha_2}$ and $k$ is related to the index budget size $b$ and buffer size $r$, then $Var_{GBKMV}$ can be restated as
$Var_{GBKMV} = L_1 F_1 +  L_2 F_2 +  L_3 F_3 = \frac{1}{m^2} [A \frac{(d^{1-\alpha_1} - r^{1-\alpha_1}) (d^{1-2\alpha_1} - r^{1-2\alpha_1})}{b - \frac{m}{32} r}
- B \frac{(d^{1-\alpha_1} - r^{1-\alpha_1}) (d^{1-2\alpha_1} - r^{1-2\alpha_1})^2}{b - \frac{m}{32} r} ]
- C (d^{1-2\alpha_1} - r^{1-2\alpha_1})$
where $A = \frac{N(\alpha_1 - 1)^2}{(1-2\alpha_1) d^{1-\alpha_1} (d^{1-\alpha_1} - 1)^2} \frac{(\alpha_2-1)^2}{-\alpha_2 (2-\alpha_2)}
 \frac{(x_t^{2-\alpha_2} - x_1^{2-\alpha_2}) (x_t^{-\alpha_2} - x_1^{-\alpha_2})}{(x_t^{-\alpha_2+1} - x_1^{-\alpha_2+1})^2}$,
 $B = \frac{N(\alpha_1 - 1)^4}{(1-2\alpha_1)^2 d^{1-\alpha_1} (d^{1-\alpha_1} - 1)^4} \frac{(\alpha_2-1)^2}{-\alpha_2 (3-\alpha_2)}
 \frac{(x_t^{3-\alpha_2} - x_1^{3-\alpha_2}) (x_t^{-\alpha_2} - x_1^{-\alpha_2})}{(x_t^{-\alpha_2+1} - x_1^{-\alpha_2+1})^2}$
 and $C = \frac{(\alpha_1-1)^2}{(1-2\alpha_1)(d^{1-\alpha_1} -1)^2} \frac{(\alpha_2-1)^2}{-\alpha_2 (2-\alpha_2)}
 \frac{(x_t^{2-\alpha_2} - x_1^{2-\alpha_2}) (x_t^{-\alpha_2} - x_1^{-\alpha_2})}{(x_t^{-\alpha_2+1} - x_1^{-\alpha_2+1})^2}$
Moreover, we have $Var_{GBKMV} = \frac{a_1 r^{5\alpha_1+1} + a_2 r^{5\alpha_1} +  a_3 r^{4\alpha_1+1} +  a_4 r^{3\alpha_1+2} + a_5  r^{3\alpha_1+1} +  a_6 r^{2\alpha_1+2}
+ a_7 r^{\alpha_1+2} + a_8 r^{3}}{(b-\frac{m}{32} r) r^{5\alpha_1}}$
where $a_1 = C \frac{m}{32} d^{1-2\alpha_1}$, $a_2 = A d^{2-2\alpha_1} - B d^{3-5\alpha_1} - b C d^{1-2\alpha_1}$,
$a_3 = -A d^{1-\alpha_1} + B d^{2-4\alpha_1}$, $a_4 = -C \frac{m}{32}$, $a_5 = -A d^{1-\alpha_1} + 2B d^{2-2\alpha_1} + bC$, $a_6 = A - 2B d^{1-2\alpha_1}$,
$a_7 = -B d^{1-\alpha_1}$ and $a_8 = B$.

We can see that the variance $Var_{GBKMV}$ can be regarded as a function of $f(r, \alpha_1, \alpha_2, b)$, i.e.,
\begin{equation}
\label{eq:optimalGoal}
Var_{GBKMV} = f(r, \alpha_1, \alpha_2, b)
\end{equation}
Similarly, for the \gkmv sketch based method, the variance can be calculated as
$$Var[\hat{C}_{GKMV}]  = \frac{(x_j + x_l)x_j x_l}{kx_j^2} F'_1 + \frac{(x_j x_l)^2}{k x_j^2} F'_2 + \frac{ x_j x_l}{x_j^2} F'_3 $$
where $F'_1 = f_{n^2}$, $F'_2 = - f_{n^2}^2$, $F'_3 = -f_{n^2}$ and $k = \frac{b}{N}(x_j + x_l) - (\frac{b}{N})^2 x_j x_l f_{n^2}$.
Let $\Delta Var = Var[\hat{C}_{GBKMV}] - Var[\hat{C}_{GKMV}]$, then for all pairs of $X_j$, $X_l$, the average of $\Delta Var$ is $V_{\Delta} =  \frac{1}{m^2} \sum_{j=1}^{m}\sum_{l=1}^{m} \Delta Var $.
Moreover, we can rewrite $V_{\Delta}$ as
$V_{\Delta} = L_1 (F'_1 - F_1) +  L_2 (F'_2 - F_2) +  L_3 (F'_3 - F_3)$.

Eventually, in order to find the optimal $r$, i.e., the number of high-frequency elements in \gbkmv method, we give the optimization goal as
$\max_r\ \ V_{GBKMV} = f(r, \alpha_1, \alpha_2, b)$, $s.t. \ \ V_{\Delta} < 0$.

In order to compute the above optimization problem, we try to extract the roots of the first derivative function of Equation~\ref{eq:optimalGoal}( i.e., $f(r, \alpha_1, \alpha_2, b)$) with respect to $r$.
However, the derivative function is a polynomial function with degree of $r$ larger than four. According to Abel's impossibility theorem~\cite{weissteinabel}, there is no algebraic solution, thus we try to give the numerical solution.

Recall that we use bitmap to keep the $r$ high-frequency elements, given the space budget $b$, the element frequency and record size distribution with power-law exponent $\alpha_1$ and $\alpha_2$ respectively, the optimization goal $\max_r\  V_{GBKMV}$ can be considered as a function $\max_r\ f(r, b, \alpha_1, \alpha_2)$.
Given a dataset \ds and the space budget $b$, we can get the power-law exponent $\alpha_1, \alpha_2$. Then we assign $8, 16, 24,...$ to $r$ and calculate the $f(r, b, \alpha_1, \alpha_2)$. In this way, we can give a good guide to the choice of $r$.

\subsubsection{\textbf{GB-KMV Sketch provides Better Accuracy than LSH-E Method}}
\label{subsubsec:gkmv_better}
In Section~\ref{subsec:lshe_analysis}, we have shown that the variance of \lshe estimator(Equation~\ref{eq:lshe_var}) is larger than that of MinHash \lsh estimator(Equation~\ref{eq:lsh_var}).
Note that \gkmv sketch is a special case of \gbkmv sketch when the buffer size $r = 0$.
By choosing an optimal buffer size $r$ in~\ref{subsubsec:choose_r}, it can guarantee that
the performance of \gbkmv is not worse than \gkmv.
Below, we show that \gkmv outperforms MinHash \lsh in terms of estimate accuracy.

\begin{theorem}
\label{theo:gbkmvBetter}
 The variance of \gkmv method is smaller than that of minHash \lsh method given the same sketch size.
\end{theorem}
\begin{proof}
Suppose that the minHash \lsh method uses $k'$ hash functions to the dataset, then the total sketch size is $T = mk'$. Let $\tau$ be the global threshold of \gkmv method, we have $\tau = \frac{mk'}{N}$ where $N$ is the total number of elements in dataset.

We first consider the \gkmv method.
Similar to Equation~\ref{eq:sigInterSize},~\ref{eq:sigUnionSize},
the intersection size of $X_j$ and $X_l$ is $D_{\cap} = x_j x_l \sum_{i=1}^{n}\frac{f_i^2}{N^2}$, and the union size is $D_{\cup} = x_j + x_l - x_j x_l \sum_{i=1}^{n}\frac{f_i^2}{N^2}$.
Then by Equation~\ref{eq:kmvintervar} the variance of the \gkmv method to estimate the containment similarity of $X_j$ in $X_l$ can be rewritten as
\begin{equation}
\label{eq:gkmv var}
  V_{\gkmv} = \frac{(x_j + x_l)x_j x_l}{k x_j^2} F_1 + \frac{(x_j x_l)^2}{k x_j^2} F_2 + \frac{x_j x_l}{x_j^2} F_3
\end{equation}
where $F_1 = f_{n^2}$, $F_2 = -(f_{n^2})^2$, $F_3 = -f_{n^2}$ and $f_{n^2} = \sum_{i=1}^{n} \frac{f_i^2}{N^2}$.

Next we compute the $k$ value of the sketch. Note that $\tau$ is the global threshold of \gkmv method. The $k$ value corresponding the intersection size of $X_j$ and $X_l$ by Equation~\ref{eq:kmv_global} is $k = \tau(x_j + x_l) - \tau^2 x_j x_l f_{n^2}$.
Then the average variance  $V_1 = \frac{1}{m^2} \sum_{j=1}^{m}\sum_{l=1}^{m} V_{\gkmv}$ is
$$ V_1 = \frac{1}{m^2}(L_1 F_1 + L_2 F_2 + L_3 F_3) $$
where $L_1 = \sum_{j=1}^{m}\sum_{l=1}^{m} \frac{(x_j + x_l)x_jx_l}{x_j^2 k}$, $L_2 = \sum_{j=1}^{m}\sum_{l=1}^{m} \frac{(x_j x_l)^2}{k x_j^2}$ and $L_3 = \sum_{j=1}^{m}\sum_{l=1}^{m} \frac{x_j x_l}{x_j^2}$. 
After some computation,
$ V_1 = \frac{1}{k'}[ \frac{(\alpha_2 - 1)^3}{-\alpha_2(2-\alpha_2)} W_1 f_{n^2} + \frac{(\alpha_2-1)^3}{\alpha_2(2-\alpha_2)(3-\alpha_2)} W_2 (f_{n^2})^2 + k' \frac{(\alpha_2-1)^2}{\alpha_2(2-\alpha_2)} W_3 f_{n^2} ]$
where $W_1 = \frac{(x_t^{2-\alpha_2} - x_1^{2-\alpha_2})^2 (x_t^{-\alpha_2} - x_1^{-\alpha_2})}{(x_t^{-\alpha_2+1} - x_1^{-\alpha_2+1})^3}$,
$W_2 = \frac{(x_t^{2-\alpha_2} - x_1^{2-\alpha_2}) (x_t^{-\alpha_2} - x_1^{-\alpha_2}) (x_t^{3-\alpha_2} - x_1^{3-\alpha_2})}{(x_t^{-\alpha_2+1} - x_1^{-\alpha_2+1})^3}$,
$W_3 = \frac{(x_t^{2-\alpha_2} - x_1^{2-\alpha_2}) (x_t^{-\alpha_2} - x_1^{-\alpha_2})}{(x_t^{-\alpha_2+1} - x_1^{-\alpha_2+1})^2}$
and $f_{n^2} = \frac{(1-\alpha_1)^2}{1-2\alpha_1} \frac{d^{1-2\alpha_1} - 1}{(d^{1-\alpha_1} - 1)^2}$. Note that $x_t$($x_1$, resp.) is the largest(smallest, resp.) set size and $d$ is the distinct number of elements.


Next we take into account the minHash \lsh method.

Given two sets $X_j$ and $X_l$, by Equation~\ref{eq:lsh_var}, the variance of minHash \lsh method to estimate the containment similarity of $X_j$ in $X_l$ is
$V_{minH} = \frac{1}{k'}[a_1 f_{n^2} + a_2 (f_{n^2})^2 + a_3 (f_{n^2})^3 + a_4 (f_{n^2})^4]$
where $a_1 = x_l + \frac{x_l^2}{x_j}$, $a_2 = -4x_l^2$, $a_3 = 5\frac{x_j x_l^3}{x_j + x_l}$ and $a_4 = -2\frac{x_j^2 x_l^4}{(x_j + x_l)^2}$.
Then the average variance $V_2 = \frac{1}{m^2} \sum_{j=1}^{m} \sum_{l=1}^{m} V_{minH}$ is
$$V_2 = \frac{1}{k' m^2}[A_1 f_{n^2} +  A_2 (f_{n^2})^2 + A_3 (f_{n^2})^3 + A_4 (f_{n^2})^4]$$

where $A_1 = \frac{\alpha_2-1}{2-\alpha_2} \frac{x_t^{2-\alpha_2} - x_1^{2-\alpha_2}}{x_t^{-\alpha_2+1} - x_1^{-\alpha_2+1}}
 + \frac{(\alpha_2 -1)^2}{-\alpha_2(3-\alpha_2)} \frac{(x_t^{3-\alpha_2} - x_1^{3-\alpha_2})(x_t^{-\alpha_2} - x_1^{-\alpha_2})}{(x_t^{-\alpha_2+1} - x_1^{-\alpha_2+1})^2}$,
$A_2 = -4 \frac{\alpha_2 -1}{3-\alpha_2} \frac{x_t^{3-\alpha_2} - x_1^{3-\alpha_2}}{x_t^{-\alpha_2+1} - x_1^{-\alpha_2+1}}$,
$A_3 = 5 \frac{\alpha_2 -1}{4-\alpha_2} \frac{x_t^{4-\alpha_2} - x_1^{4-\alpha_2}}{x_t^{-\alpha_2+1} - x_1^{-\alpha_2+1}}$
and $A_4 = -2[(\frac{\alpha_2 -1}{3-\alpha_2} \frac{x_t^{3-\alpha_2} - x_1^{3-\alpha_2}}{x_t^{-\alpha_2+1} - x_1^{-\alpha_2+1}})^2
- 2\frac{(\alpha_2-1)^2}{(4-\alpha_2)(2-\alpha_2)} \frac{(x_t^{2-\alpha_2} - x_1^{2-\alpha_2})(x_t^{4-\alpha_2} - x_1^{4-\alpha_2})}{(x_t^{-\alpha_2+1} - x_1^{-\alpha_2+1})^2}
+ 3 \frac{\alpha_2 -1}{5-\alpha_2} \frac{x_t^{5-\alpha_2} - x_1^{5-\alpha_2}}{x_t^{-\alpha_2+1} - x_1^{-\alpha_2+1}}]$


Note that $f_{n^2}$ is computed by the distribution $p_1(x)=c_1 x^{-\alpha_1}$ and the sum over set size is computed by the set size distribution $p_2(x)=c_2 x^{-\alpha_2}$, and the variance $V_1$ and $V_2$ is dependent on $\alpha_1$ and $\alpha_2$. Compare the variance $V_1$ and $V_2$, we get that $V_1 < V_2$ for all $\alpha_1 >0$ and $\alpha_2 > 0$.

Next, we analyse the performance of the two methods with the dataset following uniform distribution(i.e., $\alpha_1 = 0$, $\alpha_2 = 0$).

For \gkmv method, the average variance is 
$$ V'_1 = \frac{1}{m^2}(L_1 F_1 + L_2 F_2 + L_3 F_3) $$
where $L_1 = \sum_{j=1}^{m}\sum_{l=1}^{m} \frac{(x_j + x_l)x_jx_l}{x_j^2 k}$, $L_2 = \sum_{j=1}^{m}\sum_{l=1}^{m} \frac{(x_j x_l)^2}{k x_j^2}$ and $L_3 = \sum_{j=1}^{m}\sum_{l=1}^{m} \frac{x_j x_l}{x_j^2}$. 
After some computation,
$ V'_1 = \frac{1}{k'}[ \frac{(\alpha_2 - 1)^3}{2-\alpha_2} W_1 f_{n^2} - \frac{(\alpha_2-1)^3}{(2-\alpha_2)(3-\alpha_2)} W_2 (f_{n^2})^2 - k' \frac{(\alpha_2-1)^2}{2-\alpha_2} W_3 f_{n^2} ]$
where $W_1 = \frac{(x_t^{2-\alpha_2} - x_1^{2-\alpha_2})^2 (\ln{x_t} - \ln{x_1})}{(x_t^{-\alpha_2+1} - x_1^{-\alpha_2+1})^3}$,
$W_2 = \frac{(x_t^{2-\alpha_2} - x_1^{2-\alpha_2}) (\ln{x_t} - \ln{x_1}) (x_t^{3-\alpha_2} - x_1^{3-\alpha_2})}{(x_t^{-\alpha_2+1} - x_1^{-\alpha_2+1})^3}$,
$W_3 = \frac{(x_t^{2-\alpha_2} - x_1^{2-\alpha_2}) (\ln{x_t} - \ln{x_1})}{(x_t^{-\alpha_2+1} - x_1^{-\alpha_2+1})^2}$
and $f_{n^2} = \frac{(1-\alpha_1)^2}{1-2\alpha_1} \frac{d^{1-2\alpha_1} - 1}{(d^{1-\alpha_1} - 1)^2}$. Note that $x_t$($x_1$, resp.) is the largest(smallest, resp.) set size and $d$ is the distinct number of elements.

For \lshe method, the average variance is
$$V'_2 = \frac{1}{k' m^2}[A_1 f_{n^2} +  A_2 (f_{n^2})^2 + A_3 (f_{n^2})^3 + A_4 (f_{n^2})^4]$$
where $A_1 = \frac{\alpha_2-1}{2-\alpha_2} \frac{x_t^{2-\alpha_2} - x_1^{2-\alpha_2}}{x_t^{-\alpha_2+1} - x_1^{-\alpha_2+1}}
 + \frac{(\alpha_2 -1)^2}{3-\alpha_2} \frac{(x_t^{3-\alpha_2} - x_1^{3-\alpha_2})(\ln{x_t} - \ln{x_1})}{(x_t^{-\alpha_2+1} - x_1^{-\alpha_2+1})^2}$,
$A_2 = -4 \frac{\alpha_2 -1}{3-\alpha_2} \frac{x_t^{3-\alpha_2} - x_1^{3-\alpha_2}}{x_t^{-\alpha_2+1} - x_1^{-\alpha_2+1}}$,
$A_3 = 5 \frac{\alpha_2 -1}{4-\alpha_2} \frac{x_t^{4-\alpha_2} - x_1^{4-\alpha_2}}{x_t^{-\alpha_2+1} - x_1^{-\alpha_2+1}}$
and $A_4 = -2[(\frac{\alpha_2 -1}{3-\alpha_2} \frac{x_t^{3-\alpha_2} - x_1^{3-\alpha_2}}{x_t^{-\alpha_2+1} - x_1^{-\alpha_2+1}})^2
- 2\frac{(\alpha_2-1)^2}{(4-\alpha_2)(2-\alpha_2)} \frac{(x_t^{2-\alpha_2} - x_1^{2-\alpha_2})(x_t^{4-\alpha_2} - x_1^{4-\alpha_2})}{(x_t^{-\alpha_2+1} - x_1^{-\alpha_2+1})^2}
+ 3 \frac{\alpha_2 -1}{5-\alpha_2} \frac{x_t^{5-\alpha_2} - x_1^{5-\alpha_2}}{x_t^{-\alpha_2+1} - x_1^{-\alpha_2+1}}]$.

Similarly, we can get that $V'_1 < V'_2$.

\end{proof} 

\begin{remark}
We have illustrated that the variance of \gbkmv is smaller than that of \lshe.
Then by Chebyshev's inequality, i.e., $\Pr(|X-\mu|\geq \epsilon\sigma) \leq \frac{1}{\epsilon^2}$ 
where $\mu$ is the expectation, $\delta$ is the standard deviation and $\epsilon >1$ is a constant, 
we consider the probability that values lie outside the interval $[\mu - \epsilon\delta, \mu + \epsilon\delta]$, 
that is, values deviating from the expectation.
By Theorem 5, we get that the standard deviation $\delta_1$ of \gbkmv is smaller than $\delta_2$ of \lshe, 
then with the same interval $[\mu - \epsilon\delta, \mu + \epsilon\delta]$, the constant $\epsilon_1$ for \gbkmv is larger than $\epsilon_2$ for \lshe, thus
the probability that values lie outside the interval for \gbkmv is smaller than that for \lshe,
which means that the result of \gbkmv is more concentrated around the expected value than that of \lshe.
\end{remark}
\section{Performance Studies}
\label{sct:experiment}
In this section, we empirically evaluate the performance of our proposed \gbkmv method and
compare \lsh Ensemble~\cite{zhu2016lsh} as baseline.
We also compare our approximate \gbkmv method with the exact containment similarity search method.
All experiments are conducted on PCs with Intel Xeon $2\times 2.3GHz$ CPU and $128GB$ RAM running Debian Linux,
and the source code of \gbkmv is made available ~\cite{code:gbkmv}.
\begin{table*}[hbt]
\centering
\small
\begin{tabular}{|l|l|l|l|l|r|r|r|l|l} \hline
\cellcolor{gray!25}\textbf{Dataset} & \cellcolor{gray!25}\textbf{Abbrev}  & \cellcolor{gray!25}\textbf{Type}
 & \cellcolor{gray!25}\textbf{Record}   & \cellcolor{gray!25}\textbf{\#Records}
   & \cellcolor{gray!25}\textbf{AvgLength} & \cellcolor{gray!25}\textbf{\#DistinctEle}   & \cellcolor{gray!25}\textbf{$\alpha_1$-eleFreq} &   \cellcolor{gray!25}\textbf{$\alpha_2$-recSize} \\ \hline

Netflix~\cite{bouros2016set} & NETFLIX & Rating & Movie  & 480,189 & 209.25 & 17,770  & 1.14 & 4.95 \\ \hline

Delicious~\cite{Dataset:delic} & DELIC   & Folksonomy & User & 833,081 & 98.42 & 4,512,099 & 1.14 & 3.05  \\ \hline

CaOpenData~\cite{zhu2016lsh} & COD   & Folksonomy & User  & 65,553 & 6284 & 111,011,807  & 1.09 & 1.81 \\ \hline

Enron~\cite{Dataset:enron} & ENRON   & Text & Email  & 517,431 & 133.57 &1,113,219  & 1.16 & 3.10 \\ \hline

Reuters~\cite{Dataset:ruts} & REUTERS   & Folksonomy & User  & 833,081 & 77.6 & 283,906  & 1.32 &  6.61 \\ \hline

Webspam~\cite{webb2006introducing} & WEBSPAM &Text &Text & 350,000 & 3728 & 16,609,143 & 1.33 & 9.34 \\ \hline

WDC Web Table~\cite{zhu2016lsh} & WDC   & Text & Text & 262,893,406 & 29.2 & 111,562,175 & 1.08  & 2.4  \\ \hline

\end{tabular}
\vspace{1mm}
\caption{\small Characteristics of datasets}
\label{tb:datasets}
\vspace{-4mm}
\end{table*}

\vspace{-2mm}
\subsection{Experimental Setup}
\label{experimental setup}

\begin{figure}[hbt]
\centering
\subfigure[NETFLIX]{\includegraphics[width=0.48\linewidth]{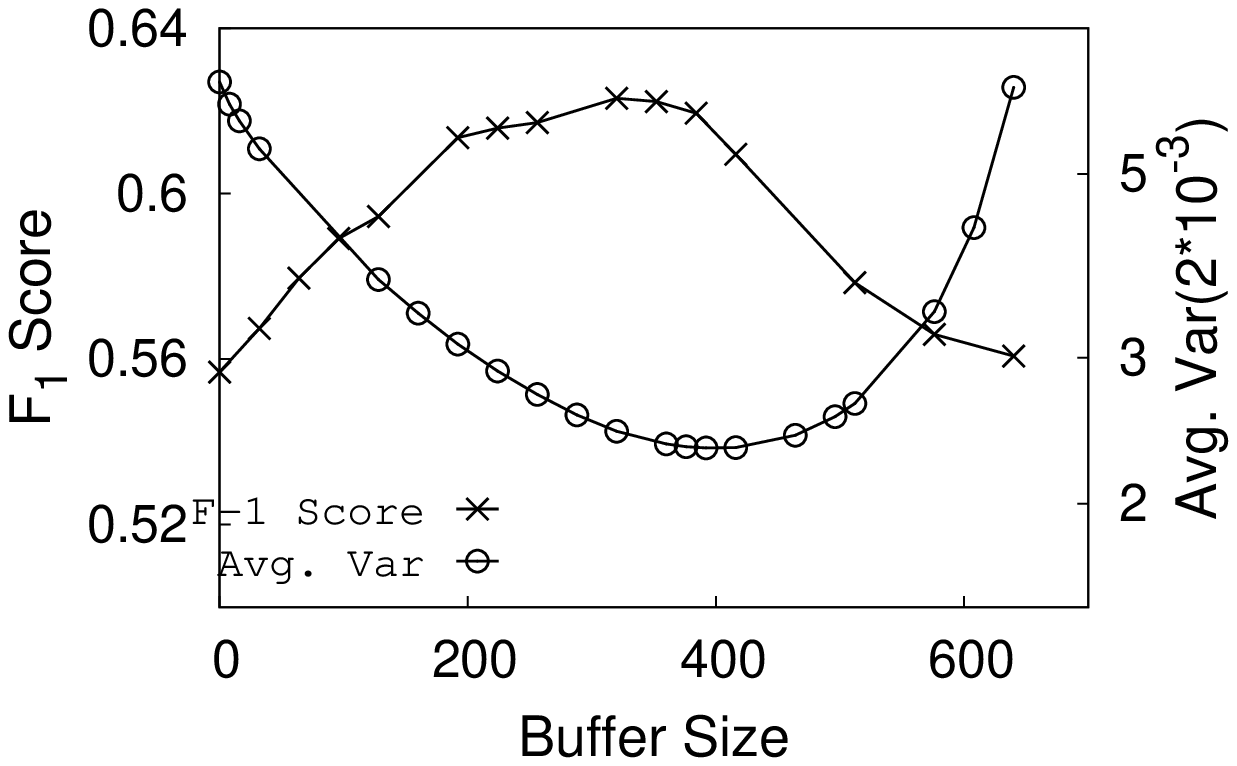}}
\subfigure[ENRON]{\includegraphics[width=0.48\linewidth]{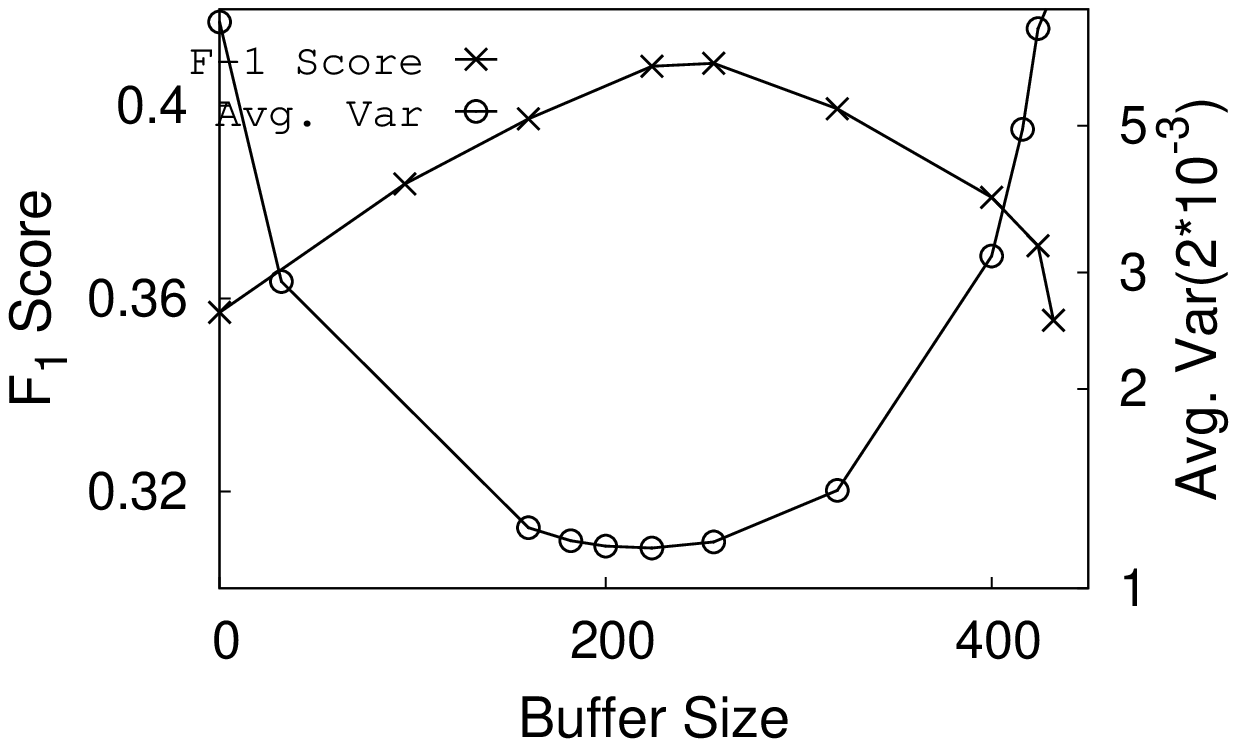}}
\vspace{-0.3cm}
\caption{\small Effect of Buffer Size}
\label{fig:tuning_pgkmv1}
\end{figure}

\begin{figure}[hbt]
\centering
\subfigure[NETFLIX]{\includegraphics[width=0.48\linewidth]{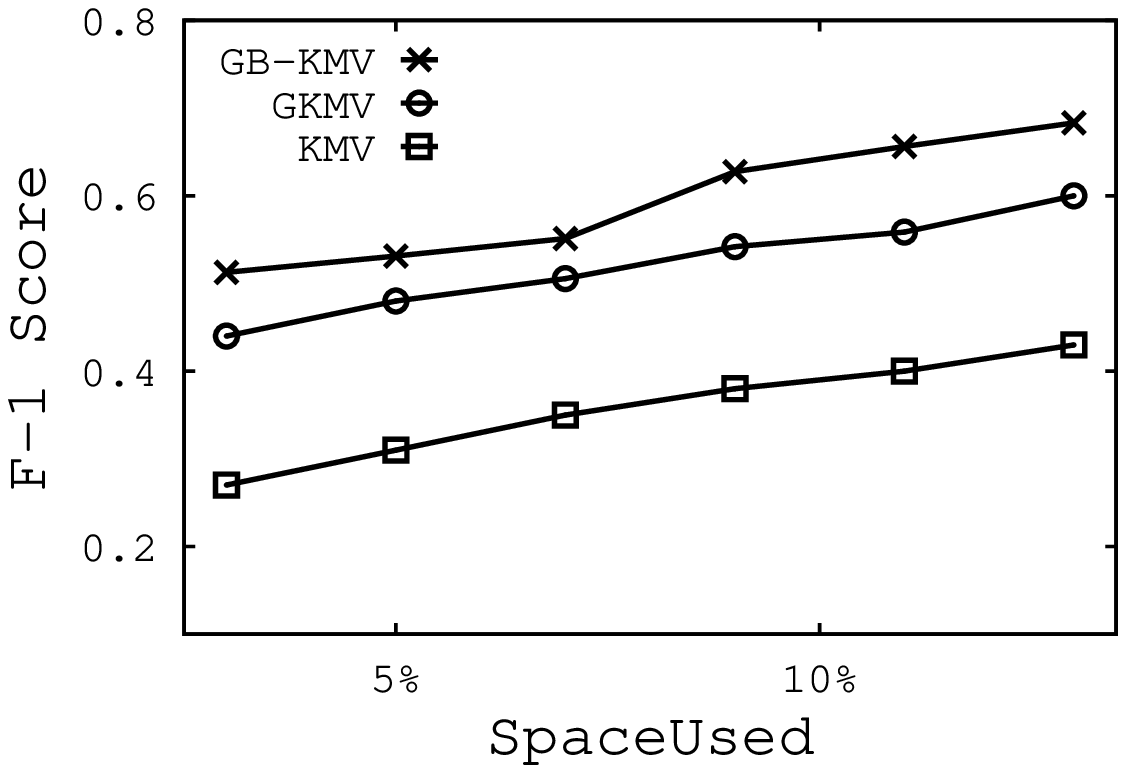}}
\subfigure[DELIC]{\includegraphics[width=0.48\linewidth]{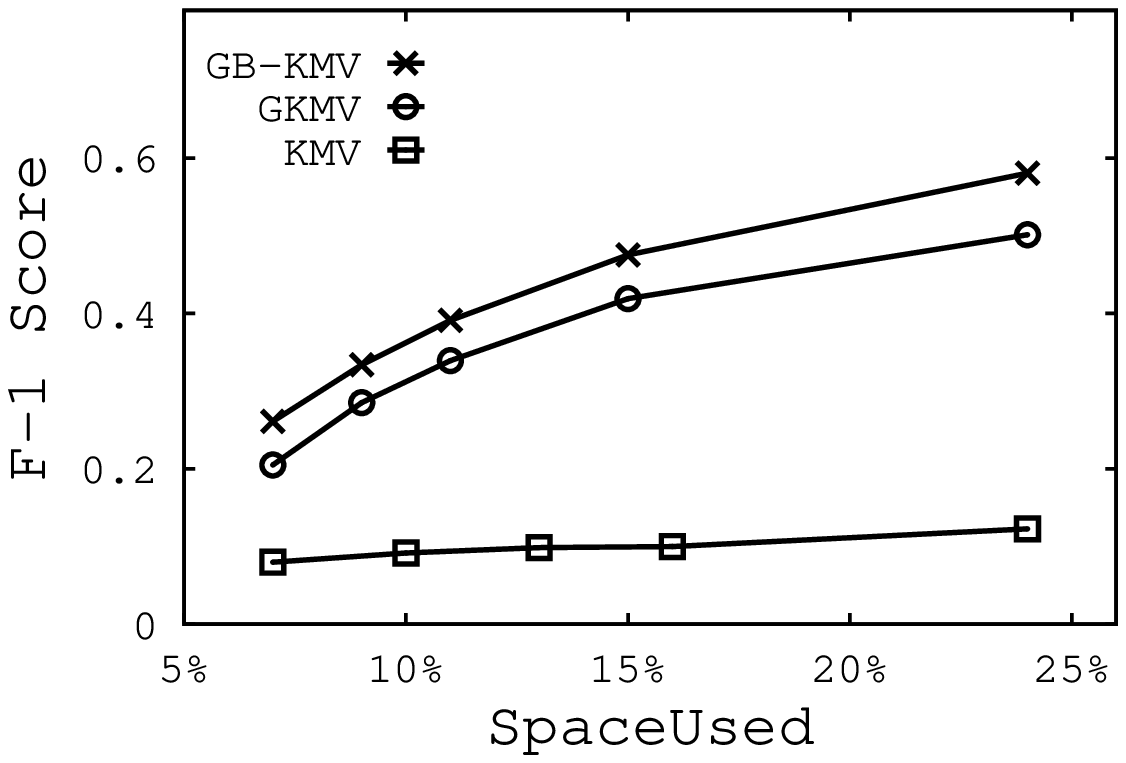}}
\subfigure[COD]{\includegraphics[width=0.48\linewidth]{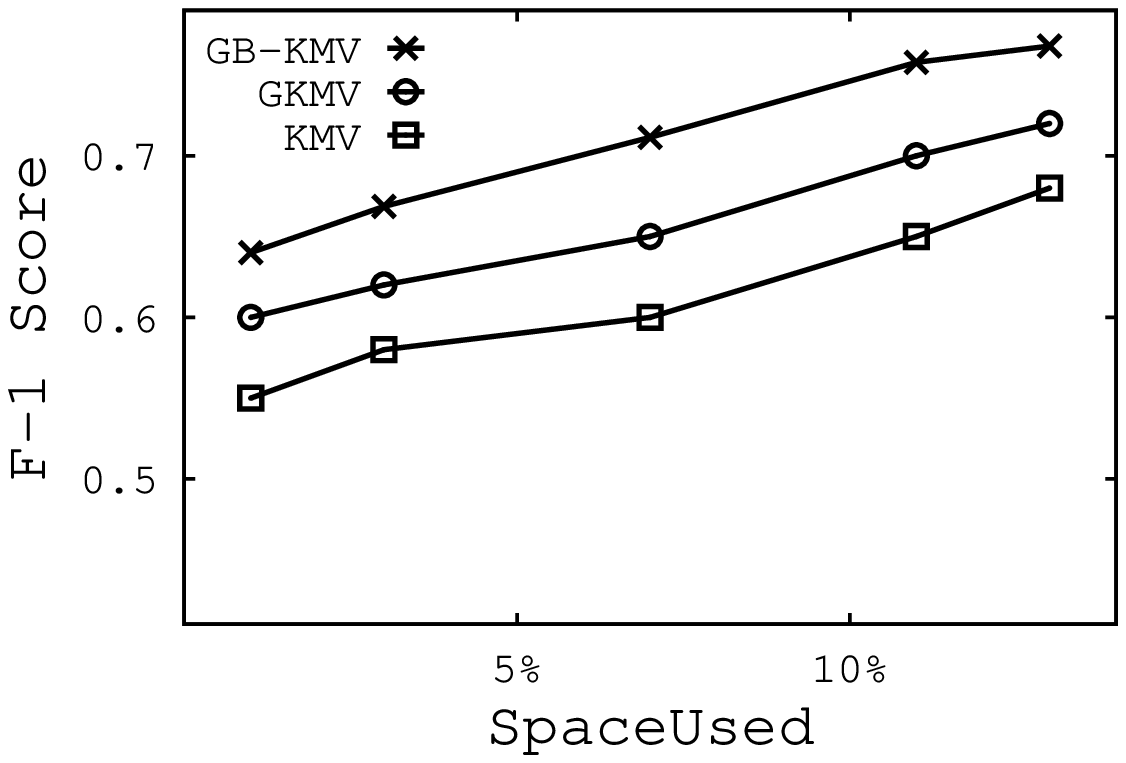}}
\subfigure[ENRON]{\includegraphics[width=0.48\linewidth]{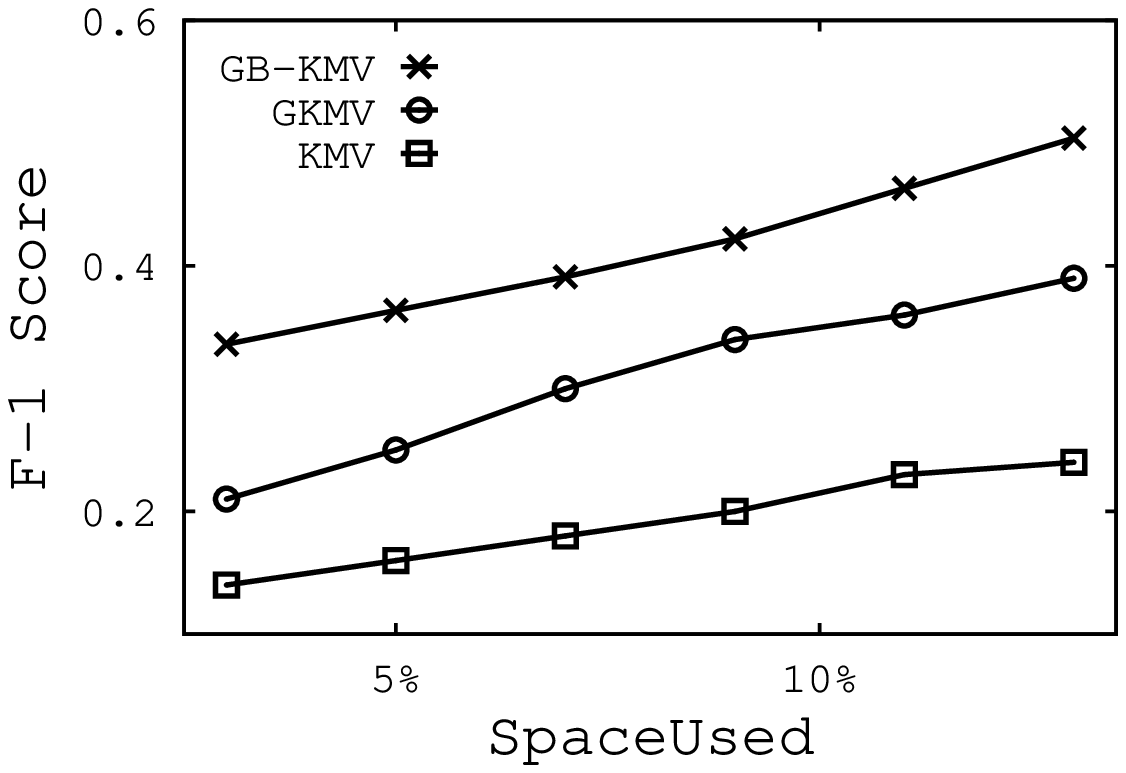}}
\subfigure[REUTERS]{\includegraphics[width=0.48\linewidth]{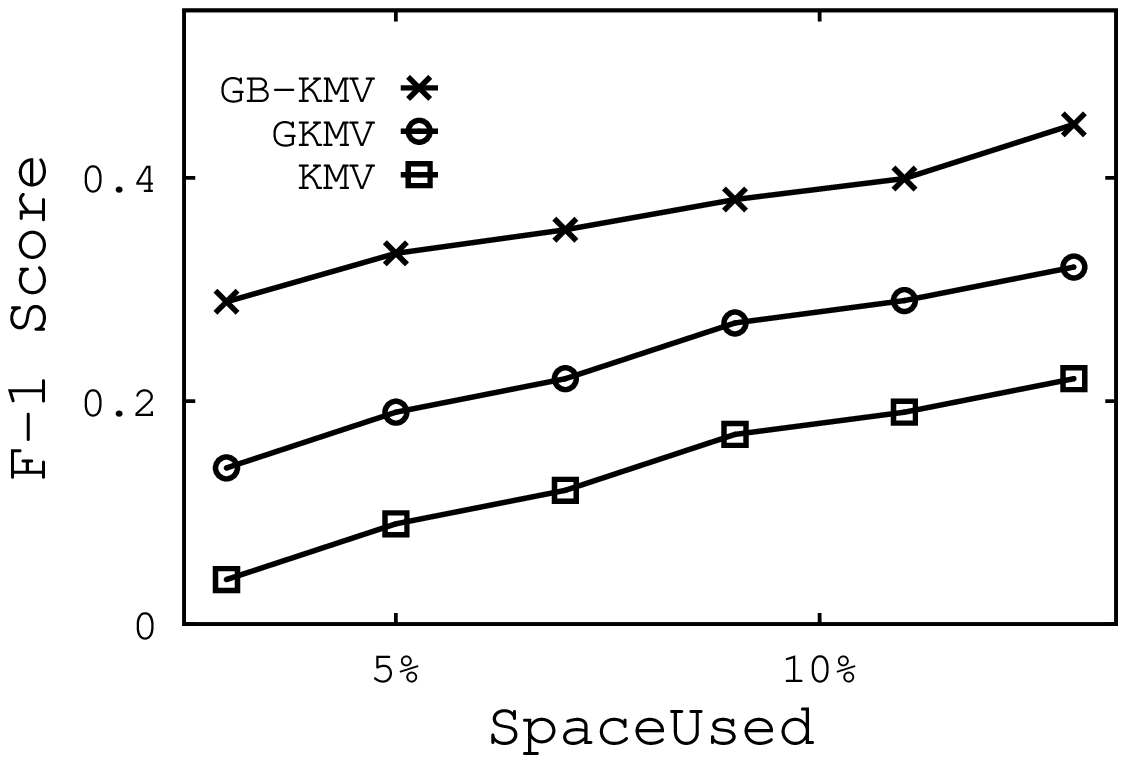}}
\subfigure[WEBSPAM]{\includegraphics[width=0.48\linewidth]{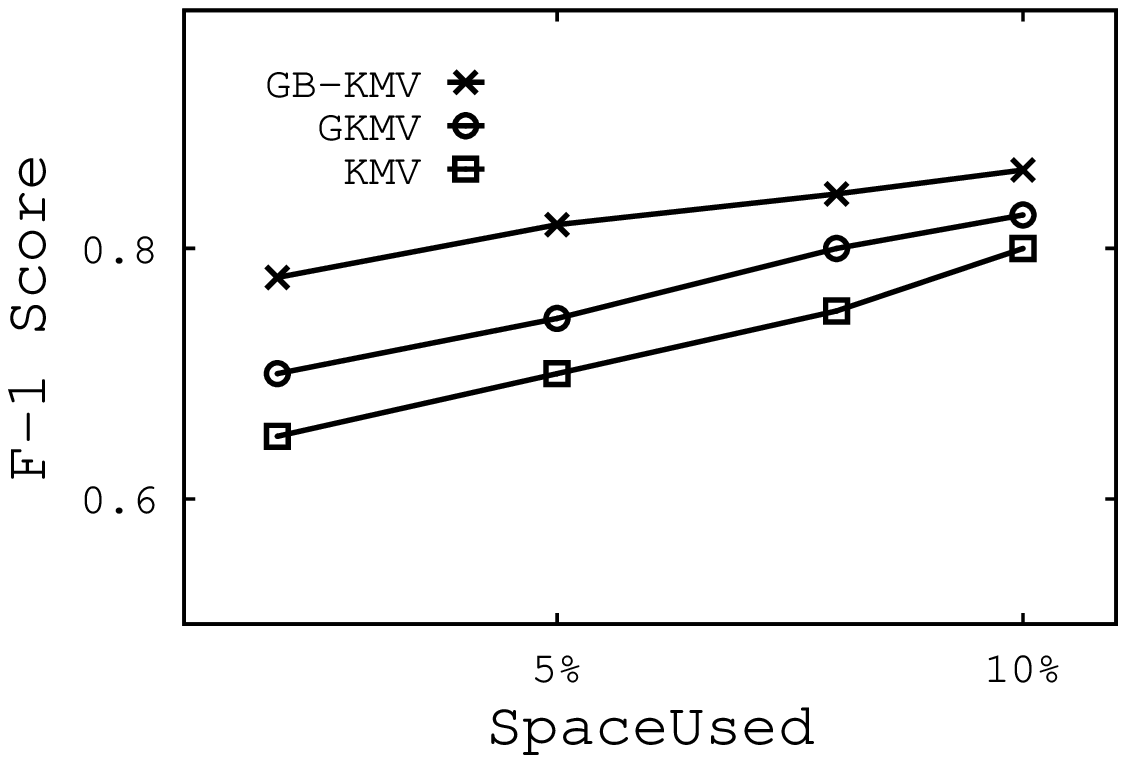}}
\subfigure[WDC]{\includegraphics[width=0.48\linewidth]{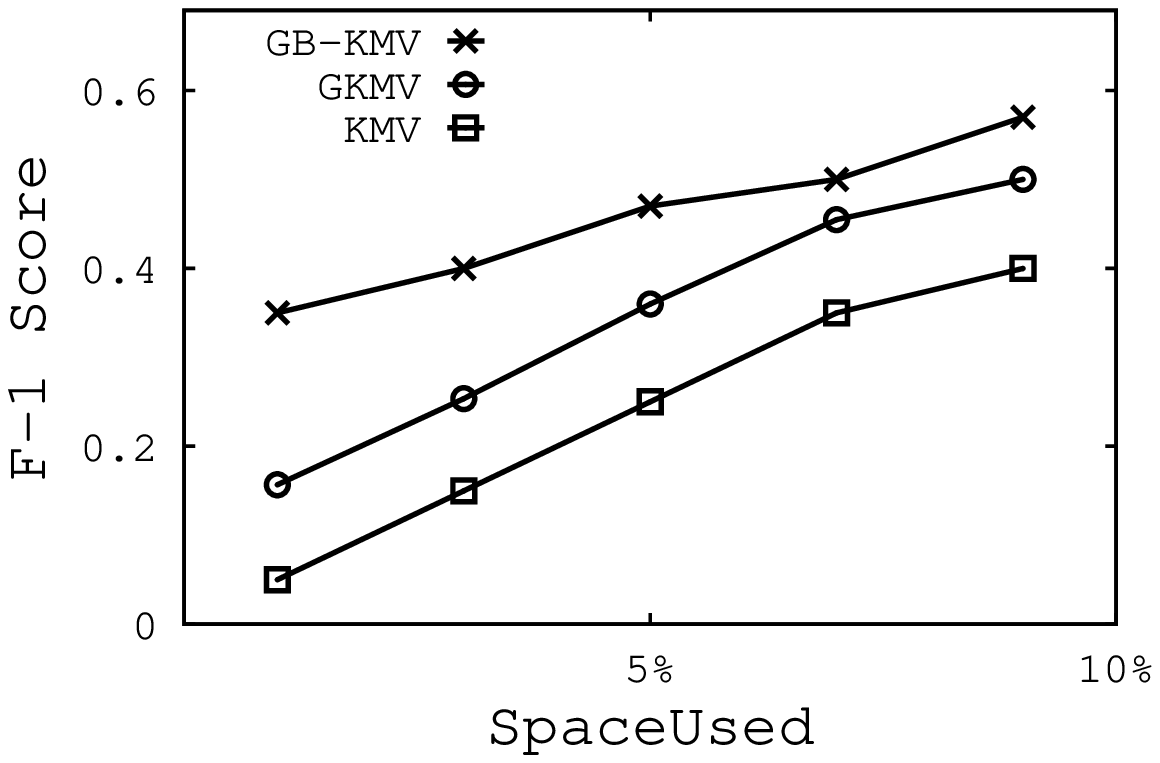}}

\vspace{-0.3cm}
\centering
\caption{\small \gbkmv, \gkmv, \kmv comparison }
\vspace{-2mm}
\label{fig:tuning_pgkmv2}
\end{figure}

\begin{figure*}[hbt]
\centering
\subfigure{\includegraphics[width=0.24\linewidth]{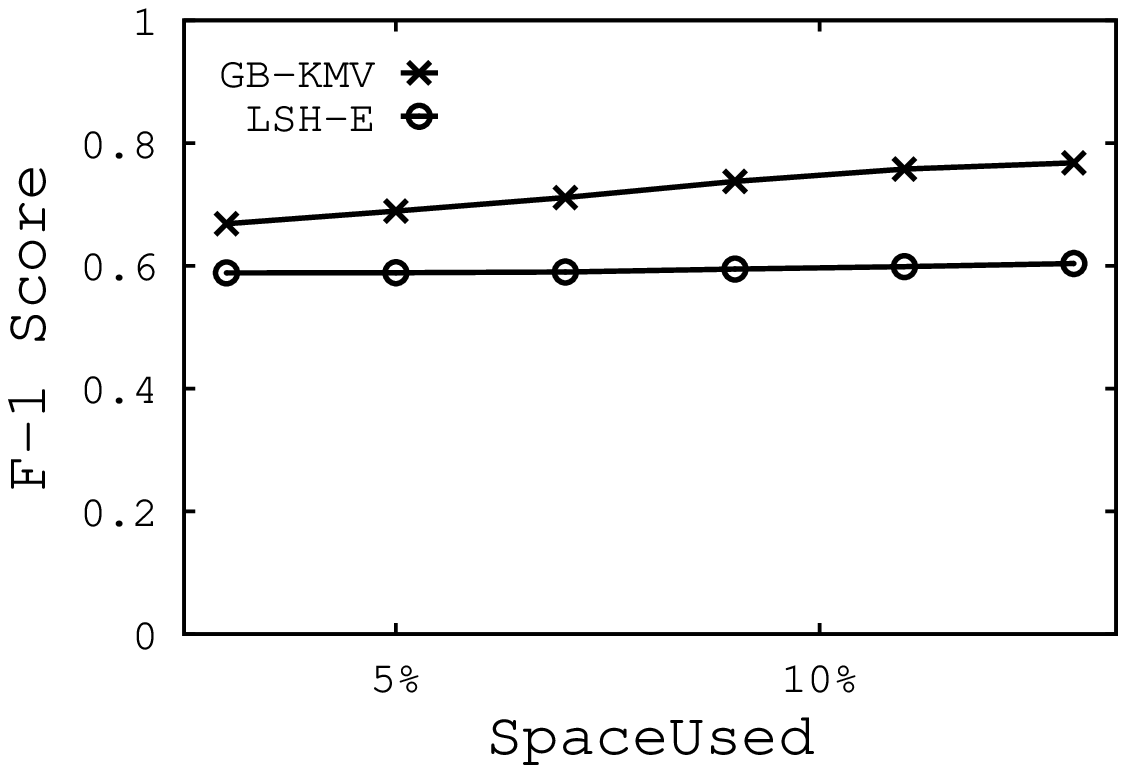}}
\subfigure{\includegraphics[width=0.24\linewidth]{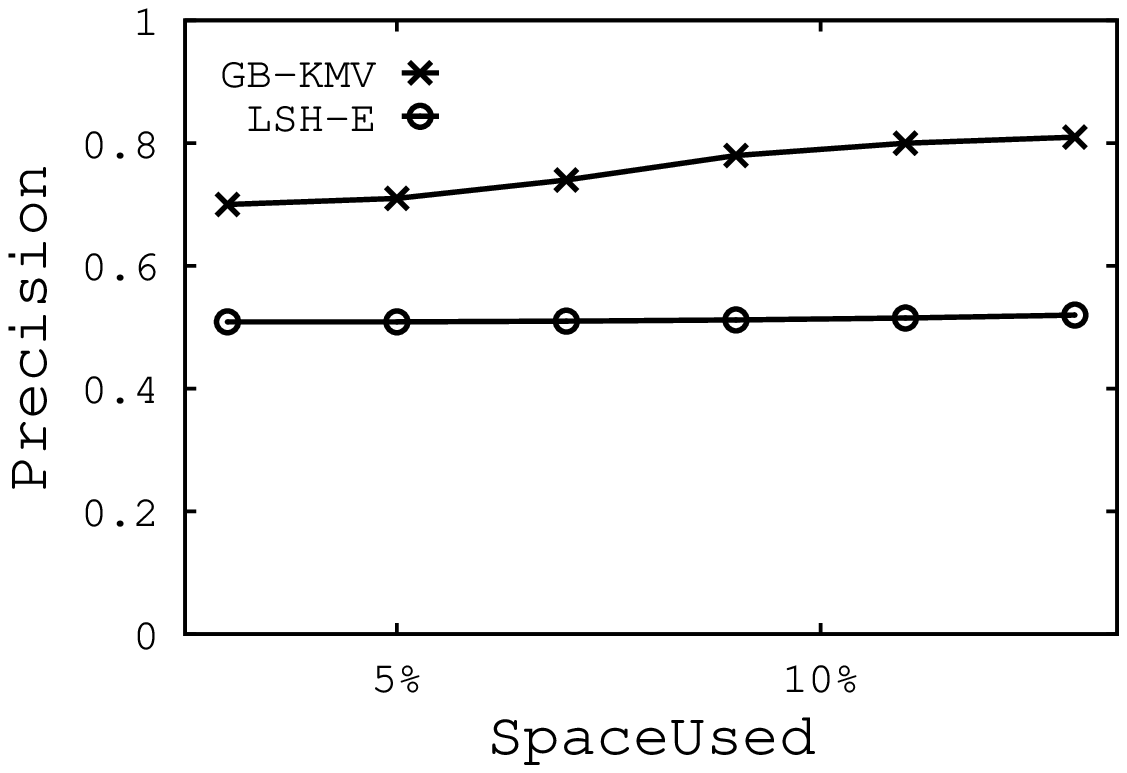}}
\subfigure{\includegraphics[width=0.24\linewidth]{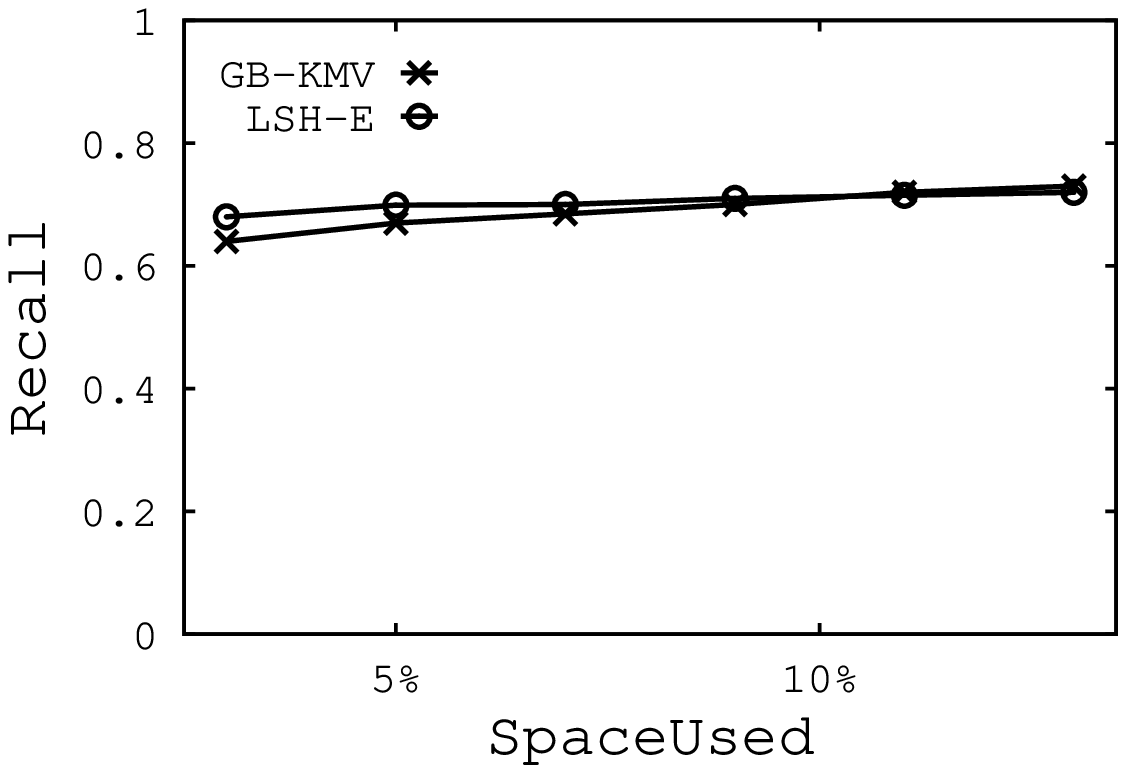}}
\subfigure{\includegraphics[width=0.24\linewidth]{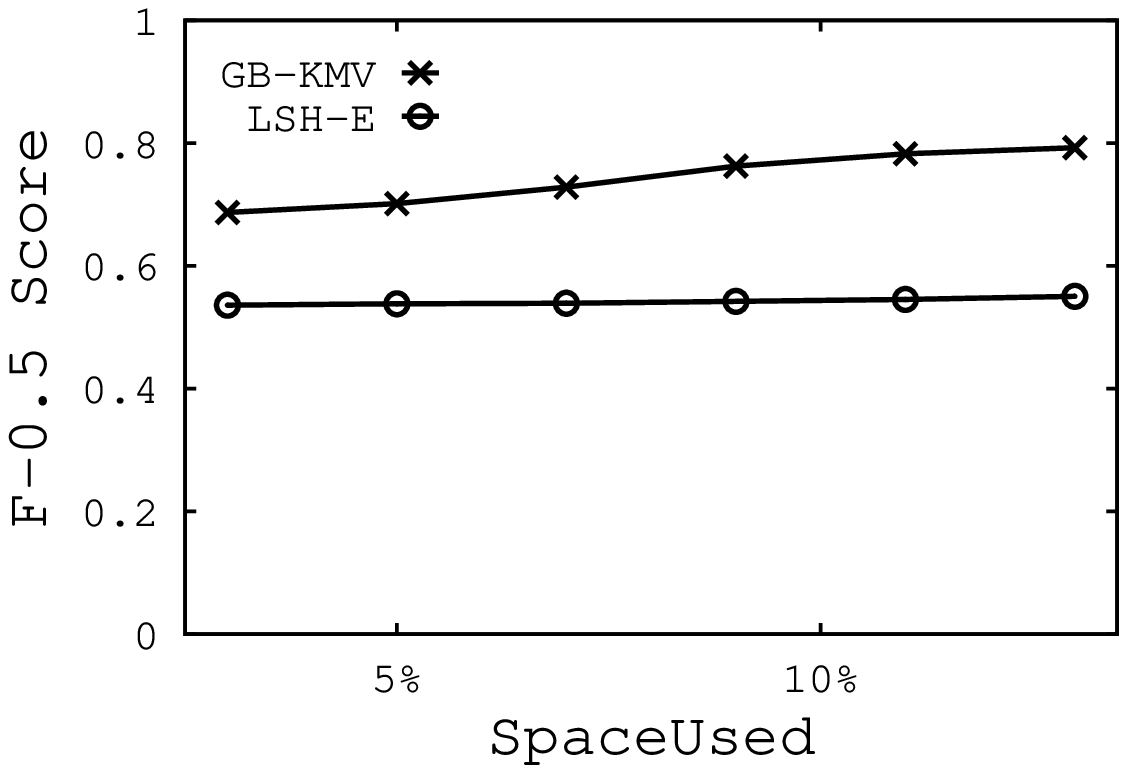}}
\vspace{-0.3cm}
\centering
\caption{\small Accuracy versus Space on COD}
\vspace{-4mm}
\label{fig:space_accu_cod}
\end{figure*}

\begin{figure*}[hbt]
\centering
\subfigure{\includegraphics[width=0.24\linewidth]{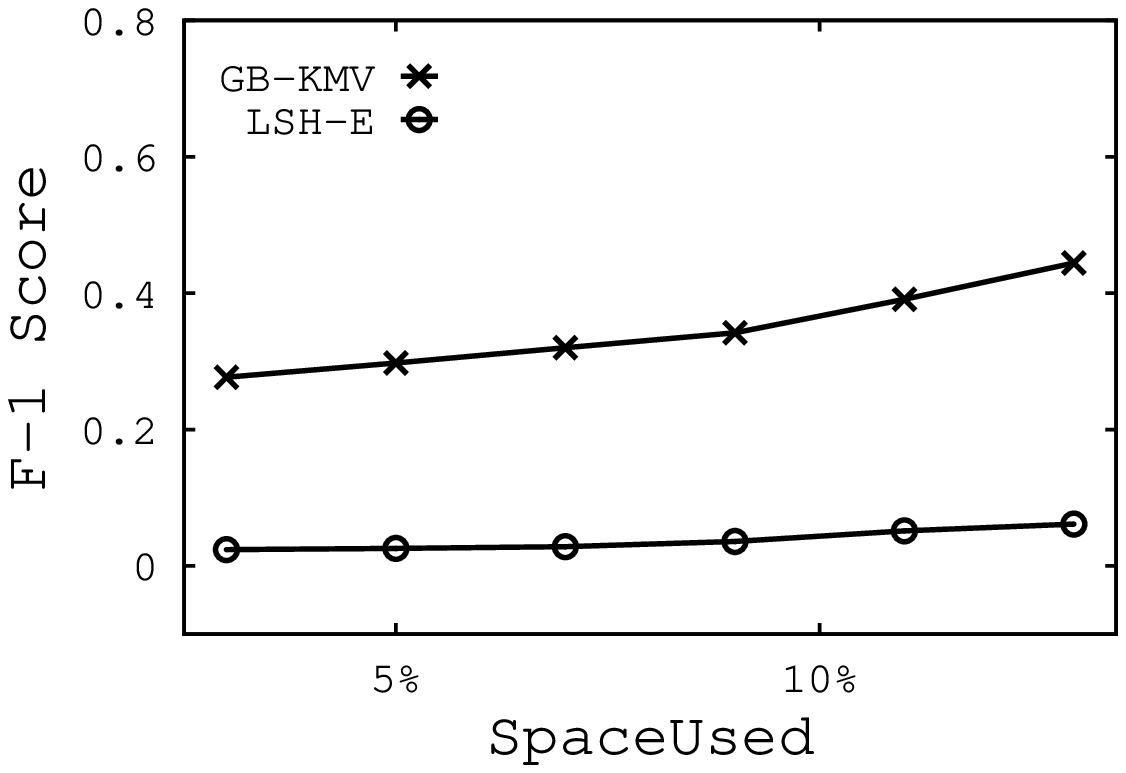}}
\subfigure{\includegraphics[width=0.24\linewidth]{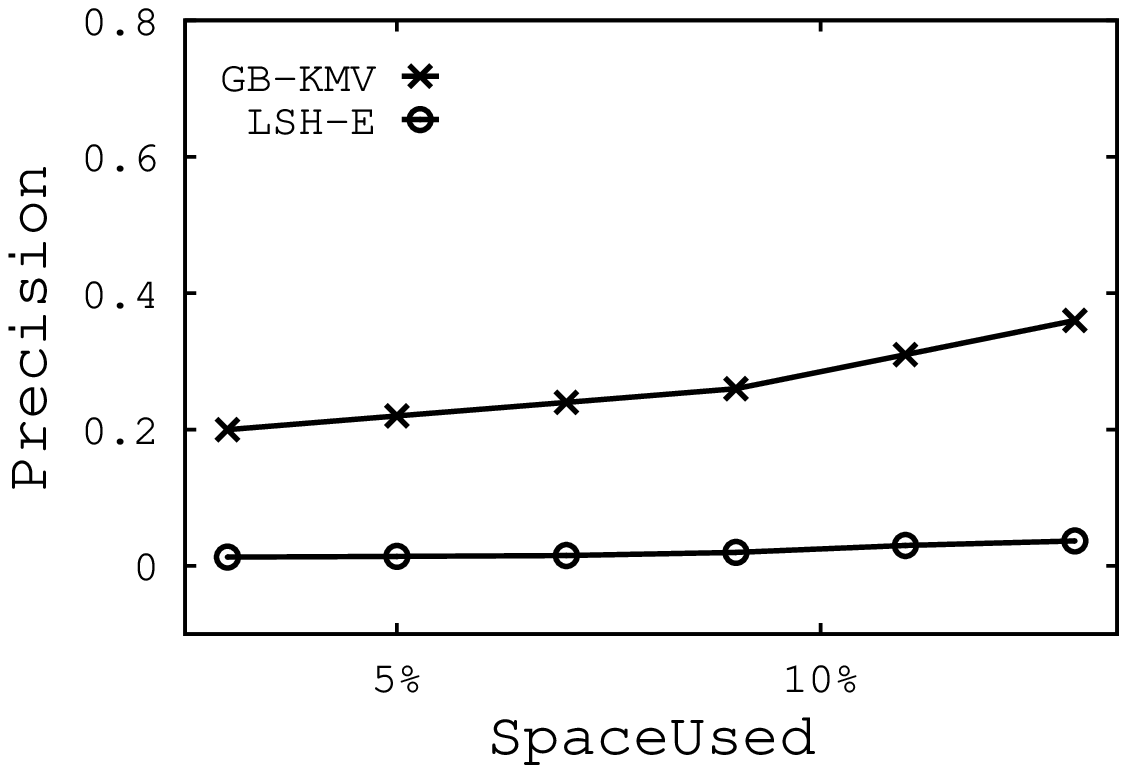}}
\subfigure{\includegraphics[width=0.24\linewidth]{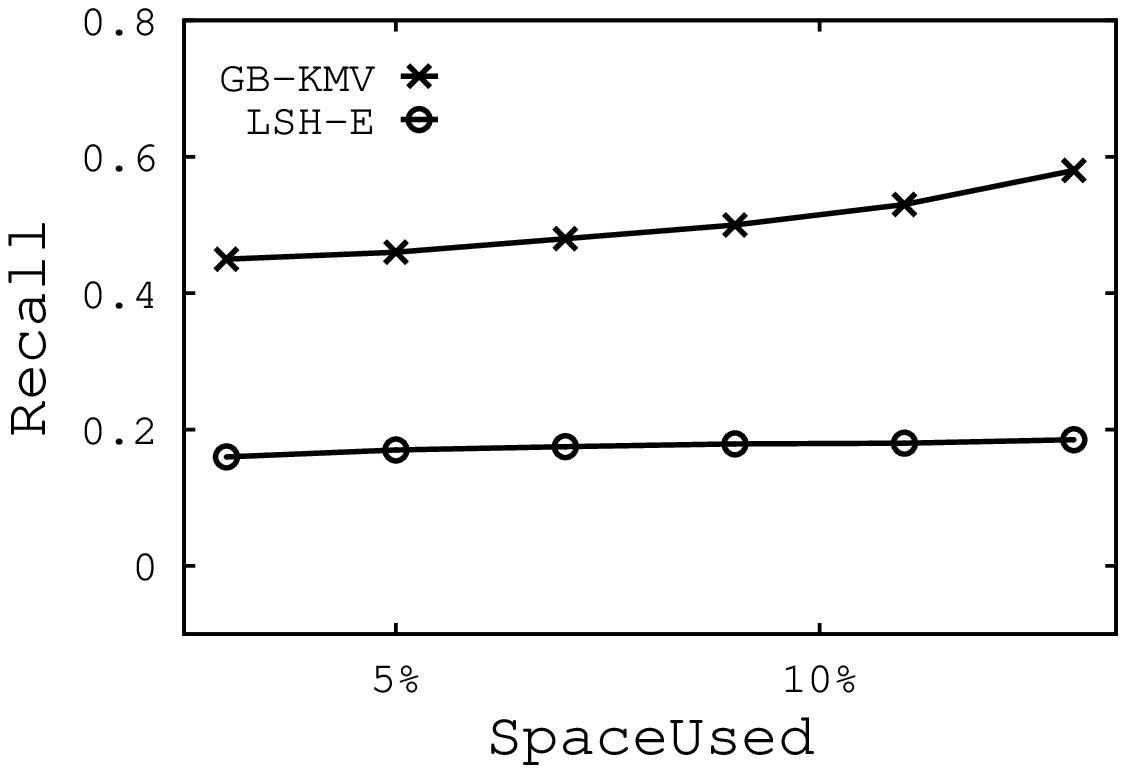}}
\subfigure{\includegraphics[width=0.24\linewidth]{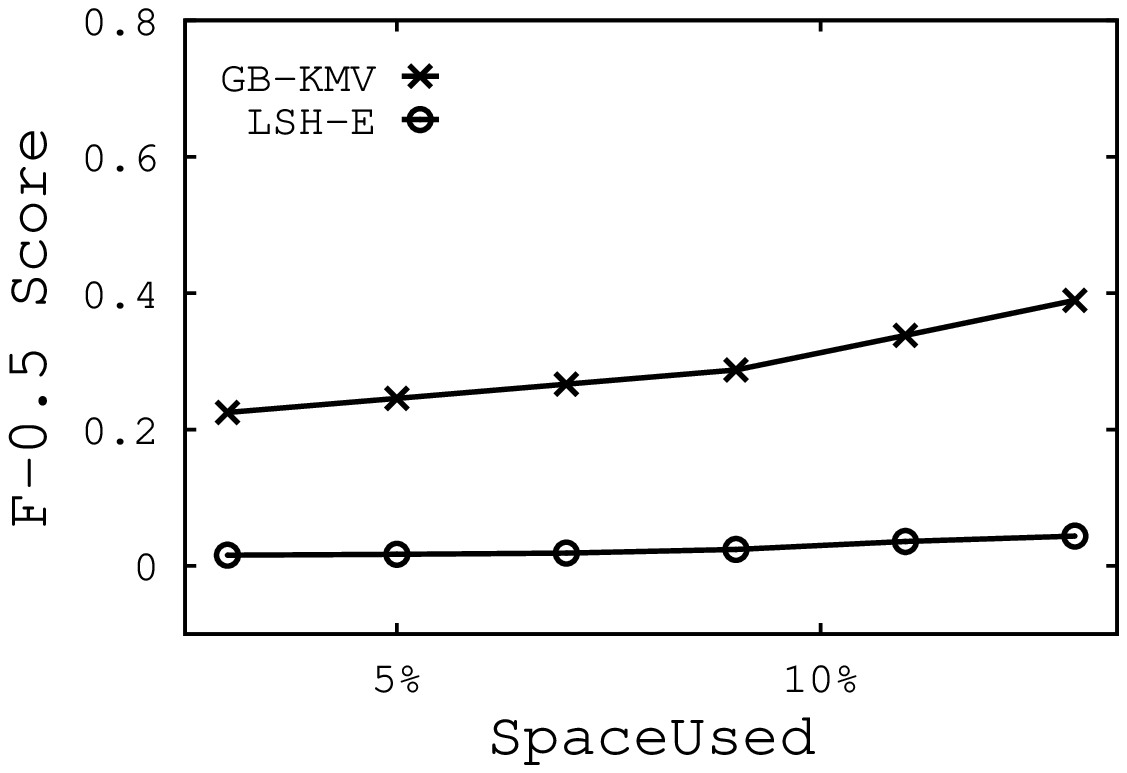}}
\vspace{-0.3cm}
\centering
\caption{\small Accuracy versus Space on DELIC}
\vspace{-4mm}
\label{fig:space_accu_delic}
\end{figure*}

\begin{figure*}[hbt]
\centering
\subfigure{\includegraphics[width=0.24\linewidth]{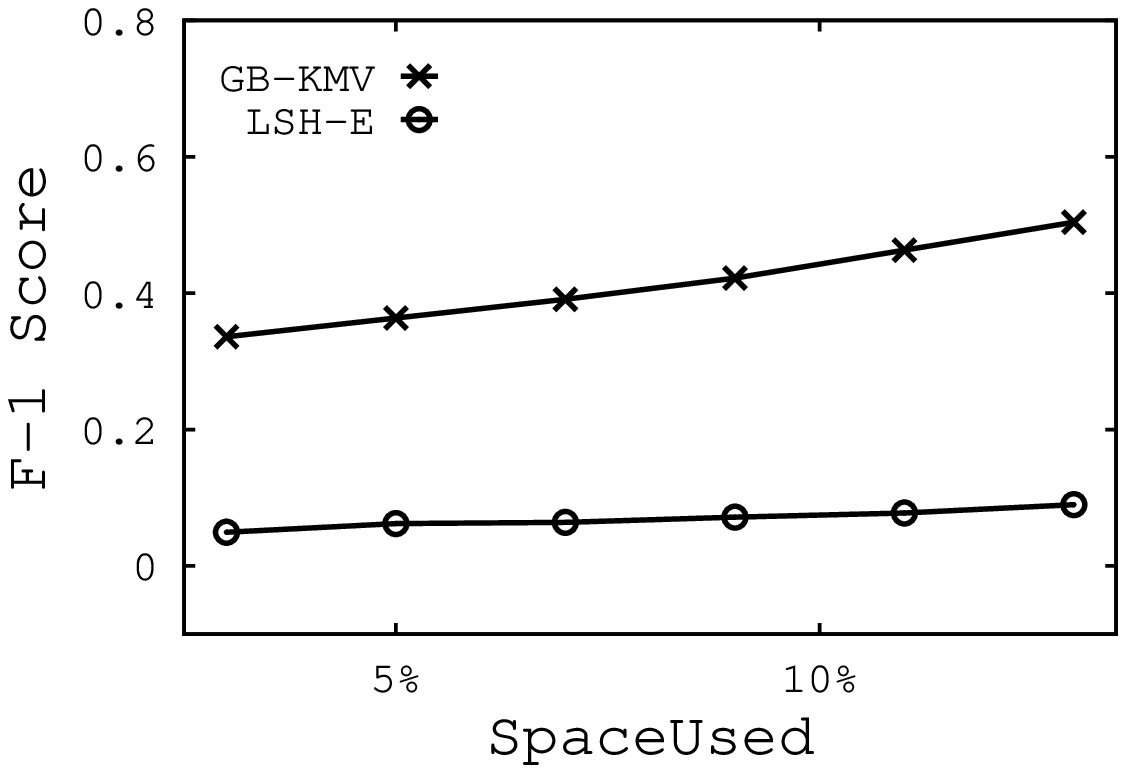}}
\subfigure{\includegraphics[width=0.24\linewidth]{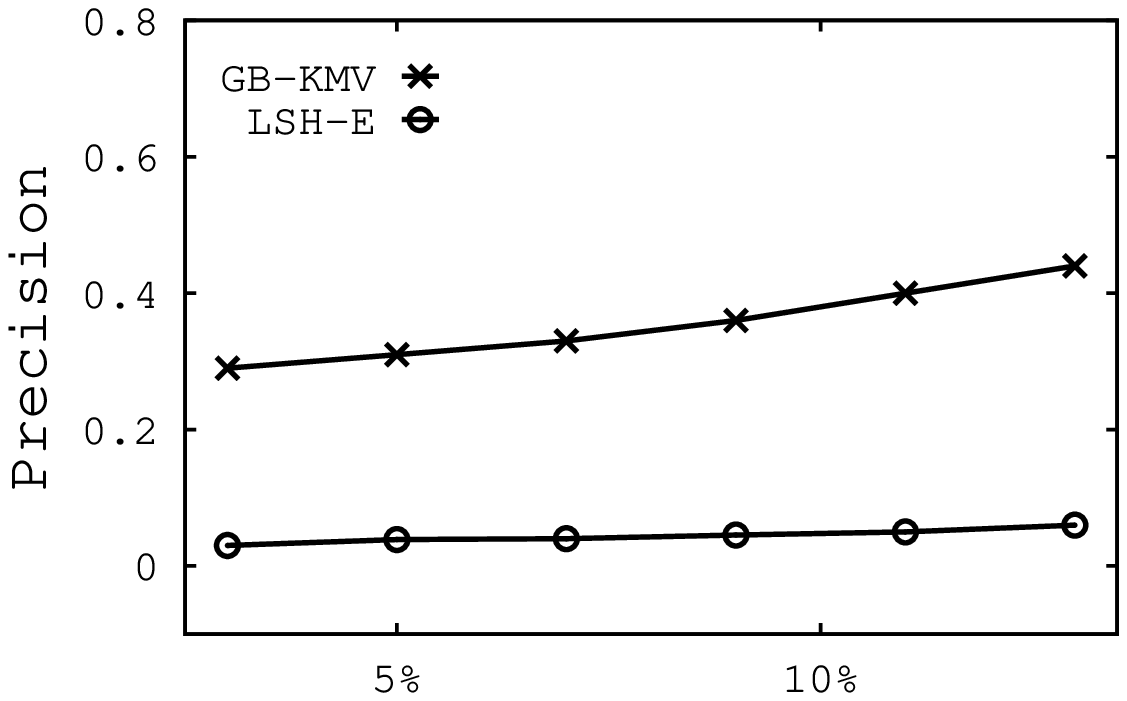}}
\subfigure{\includegraphics[width=0.24\linewidth]{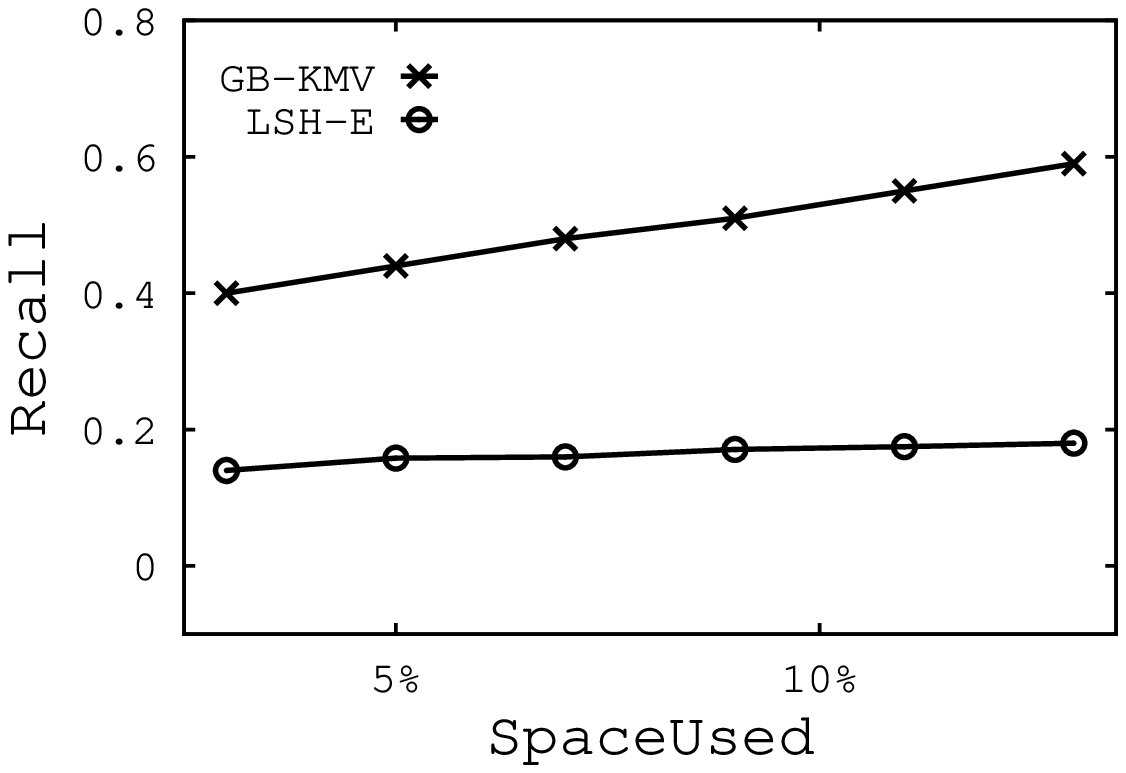}}
\subfigure{\includegraphics[width=0.24\linewidth]{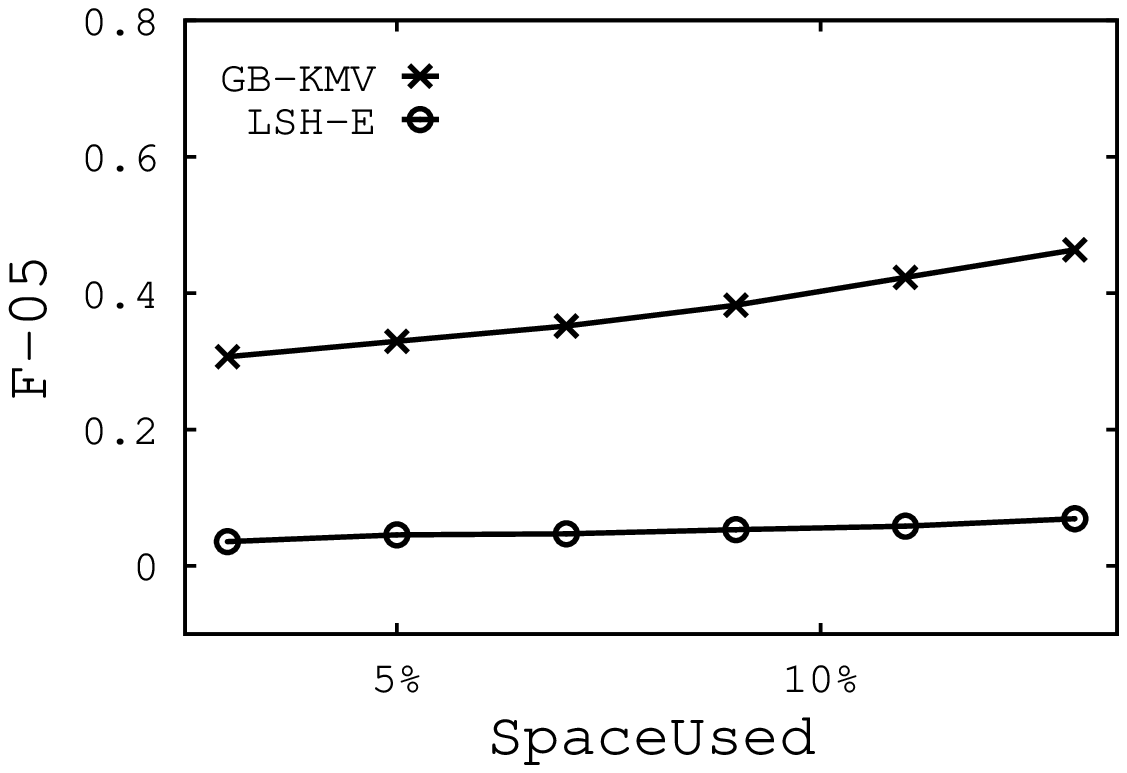}}
\vspace{-0.3cm}
\centering
\caption{\small Accuracy versus Space on ENRON}
\vspace{-4mm}
\label{fig:space_accu_enron}
\end{figure*}

\begin{figure*}[hbt]
\centering
\subfigure{\includegraphics[width=0.24\linewidth]{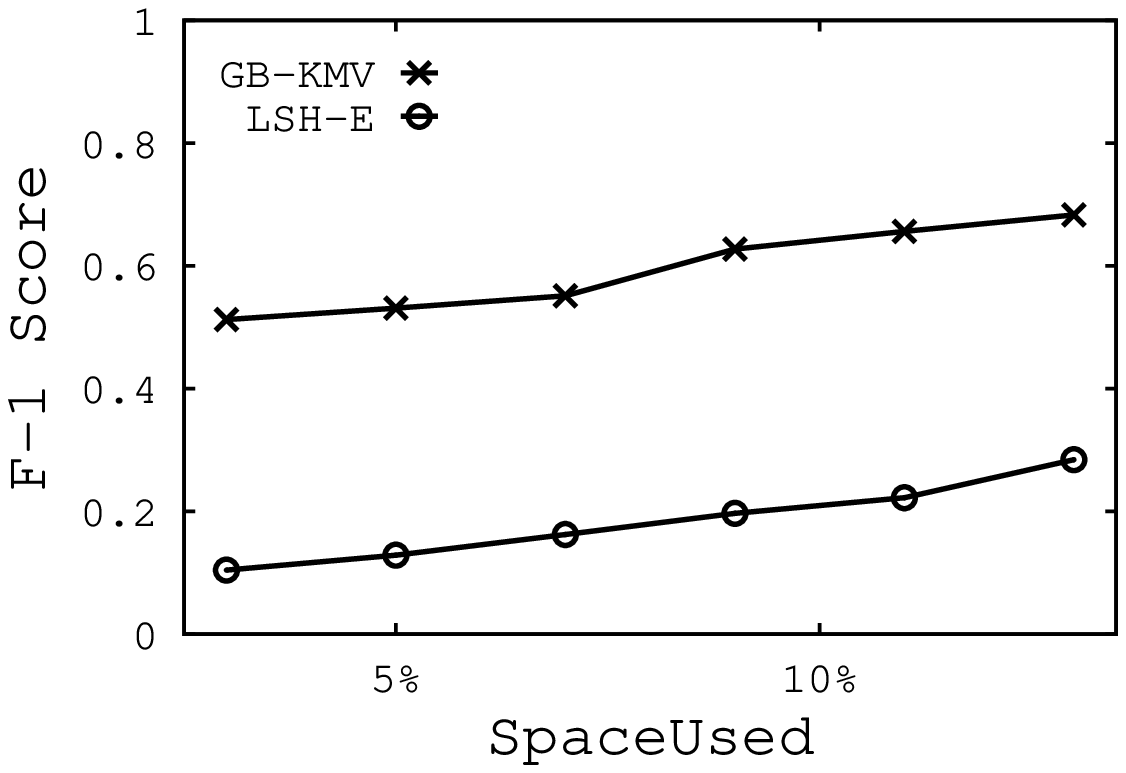}}
\subfigure{\includegraphics[width=0.24\linewidth]{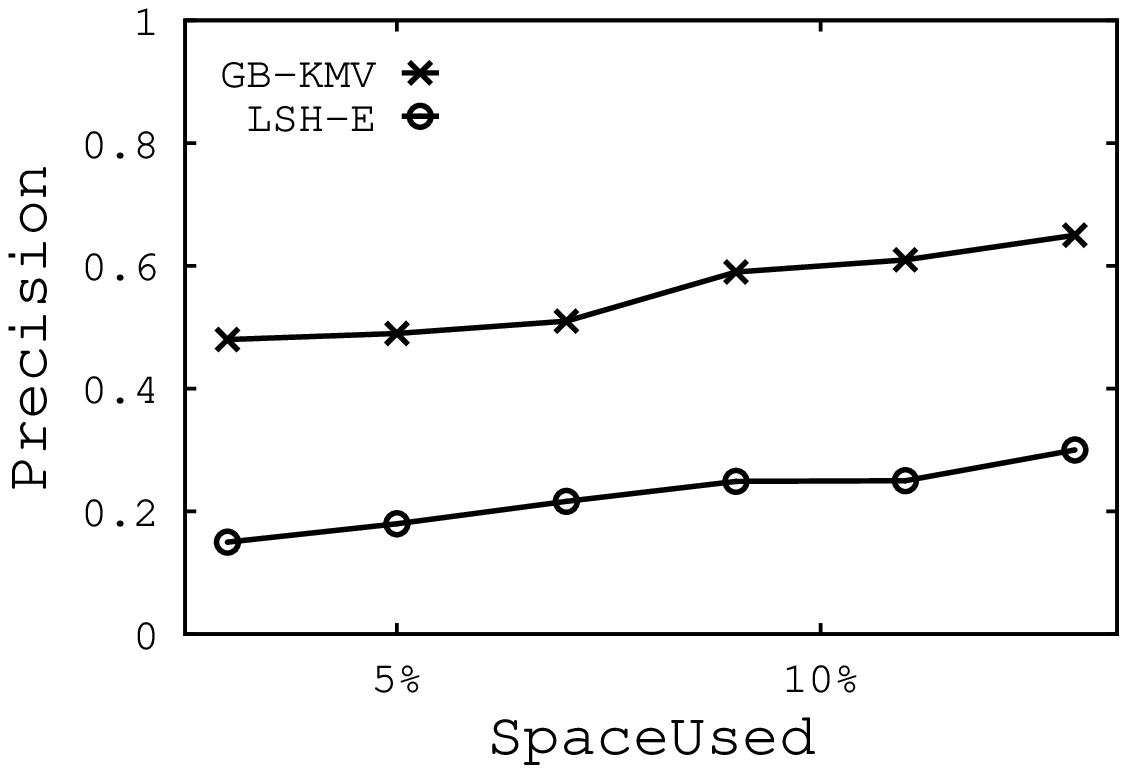}}
\subfigure{\includegraphics[width=0.24\linewidth]{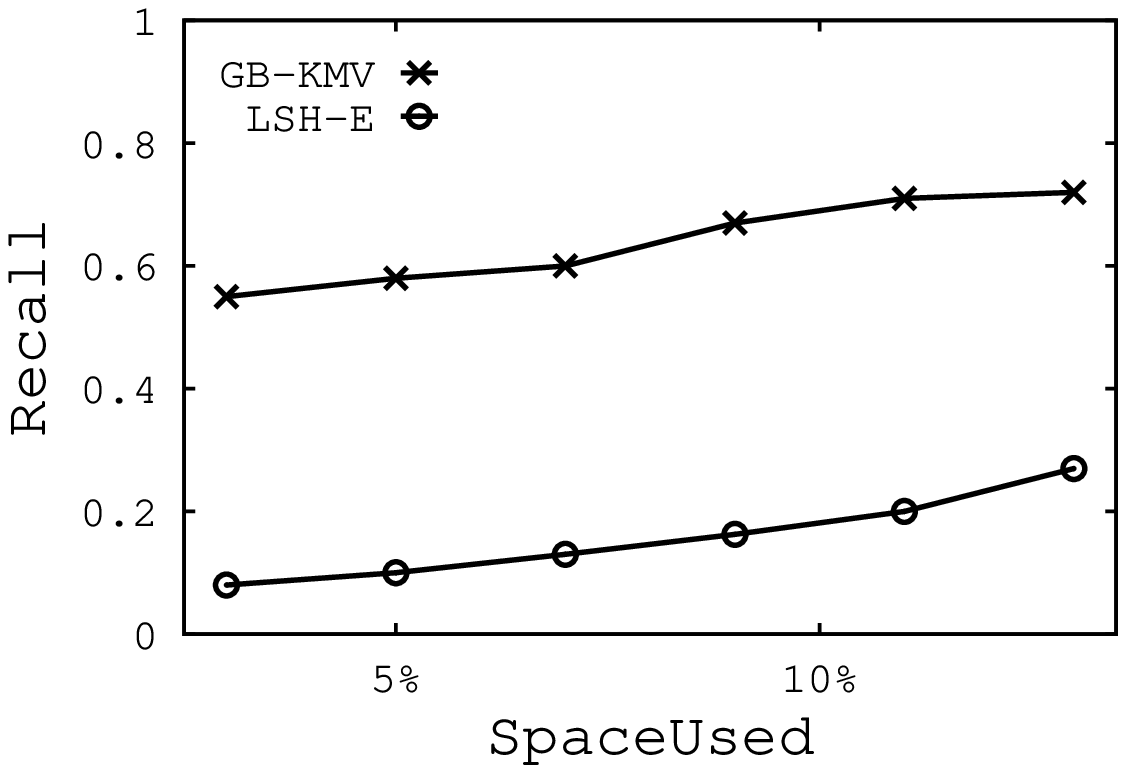}}
\subfigure{\includegraphics[width=0.24\linewidth]{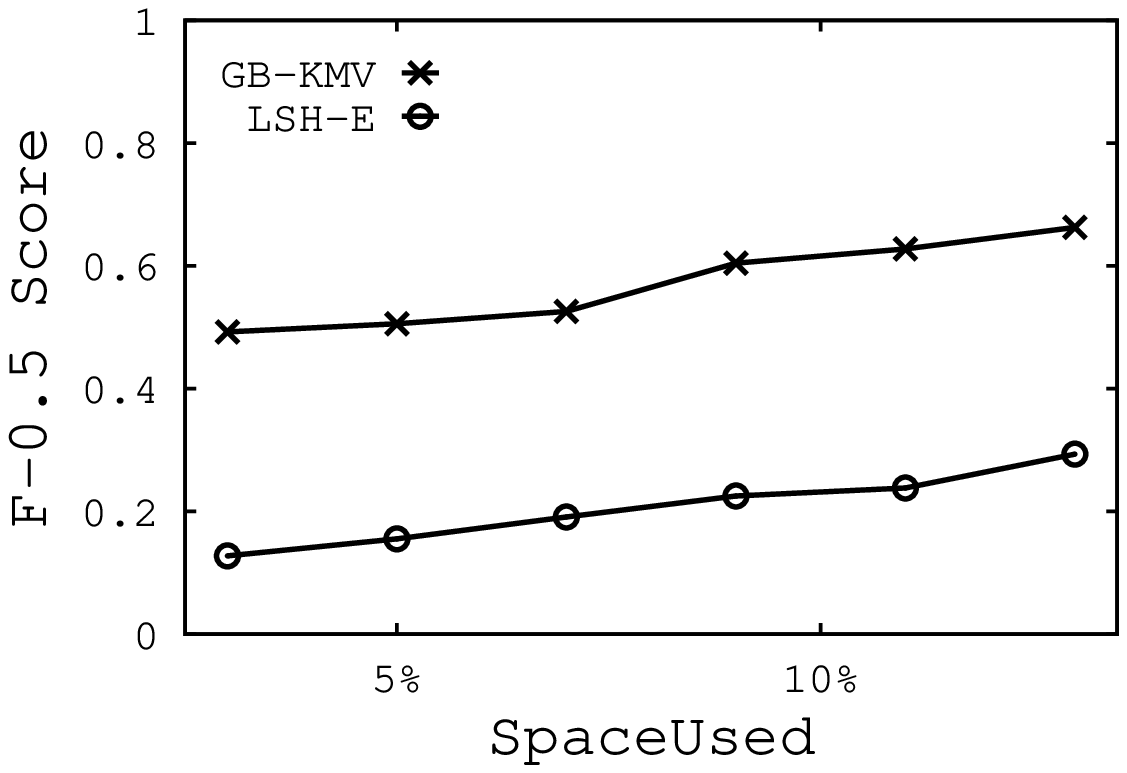}}
\vspace{-0.3cm}
\centering
\caption{\small Accuracy versus Space on NETFLIX}
\vspace{-4mm}
\label{fig:space_accu_netf}
\end{figure*}

\begin{figure*}[hbt]
\centering
\subfigure{\includegraphics[width=0.24\linewidth]{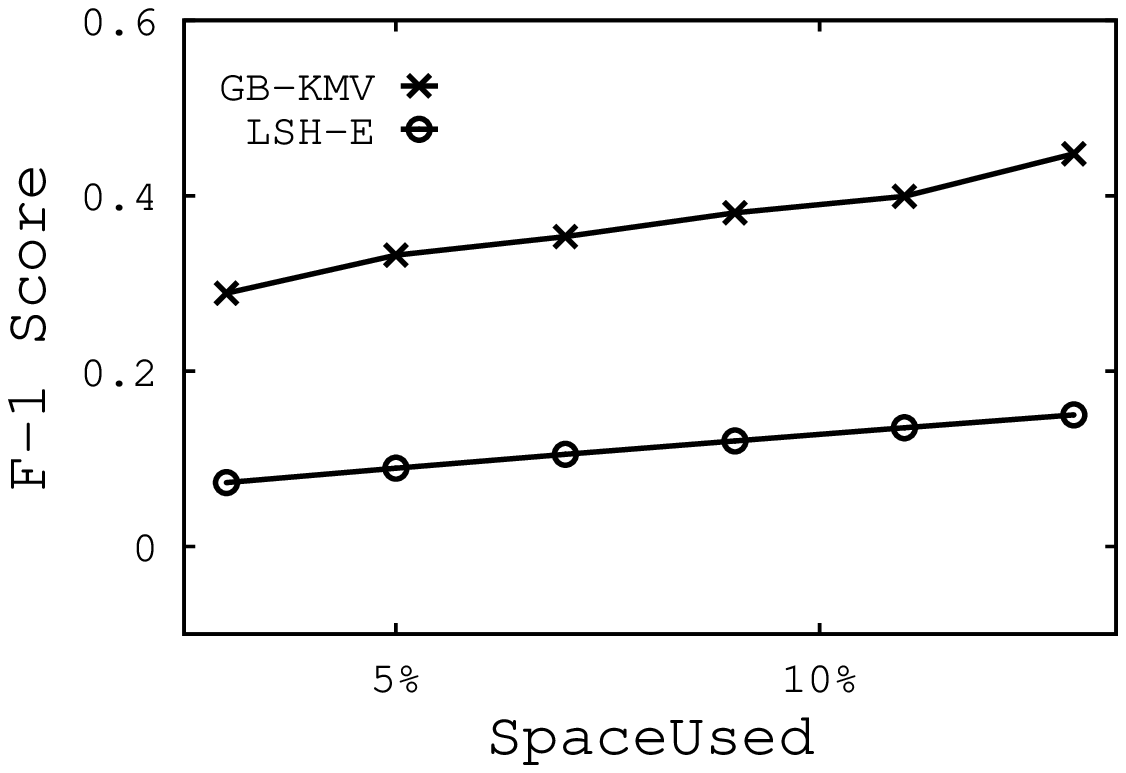}}
\subfigure{\includegraphics[width=0.24\linewidth]{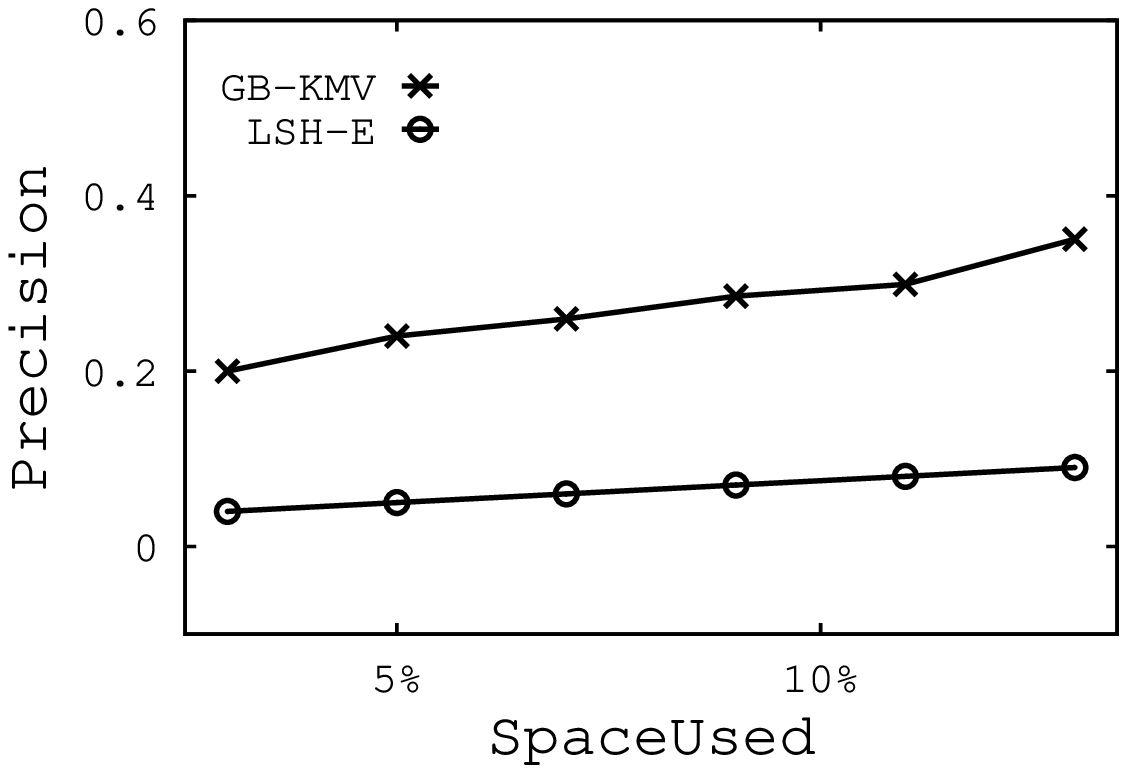}}
\subfigure{\includegraphics[width=0.24\linewidth]{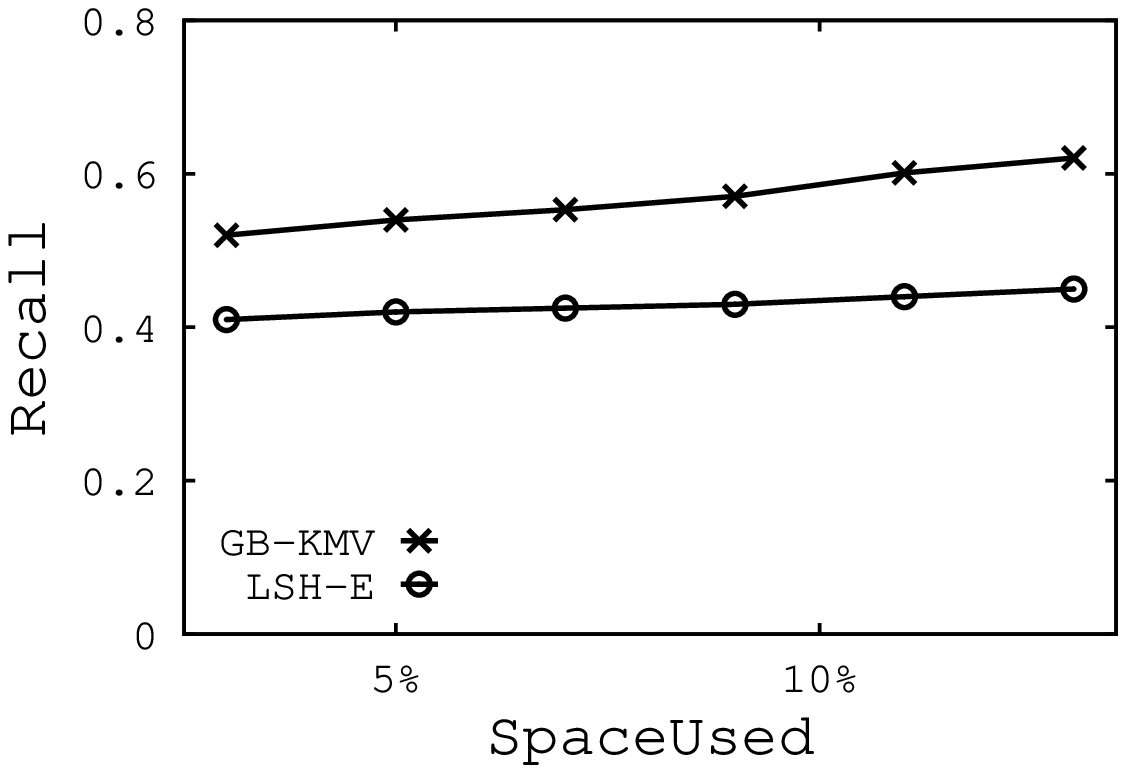}}
\subfigure{\includegraphics[width=0.24\linewidth]{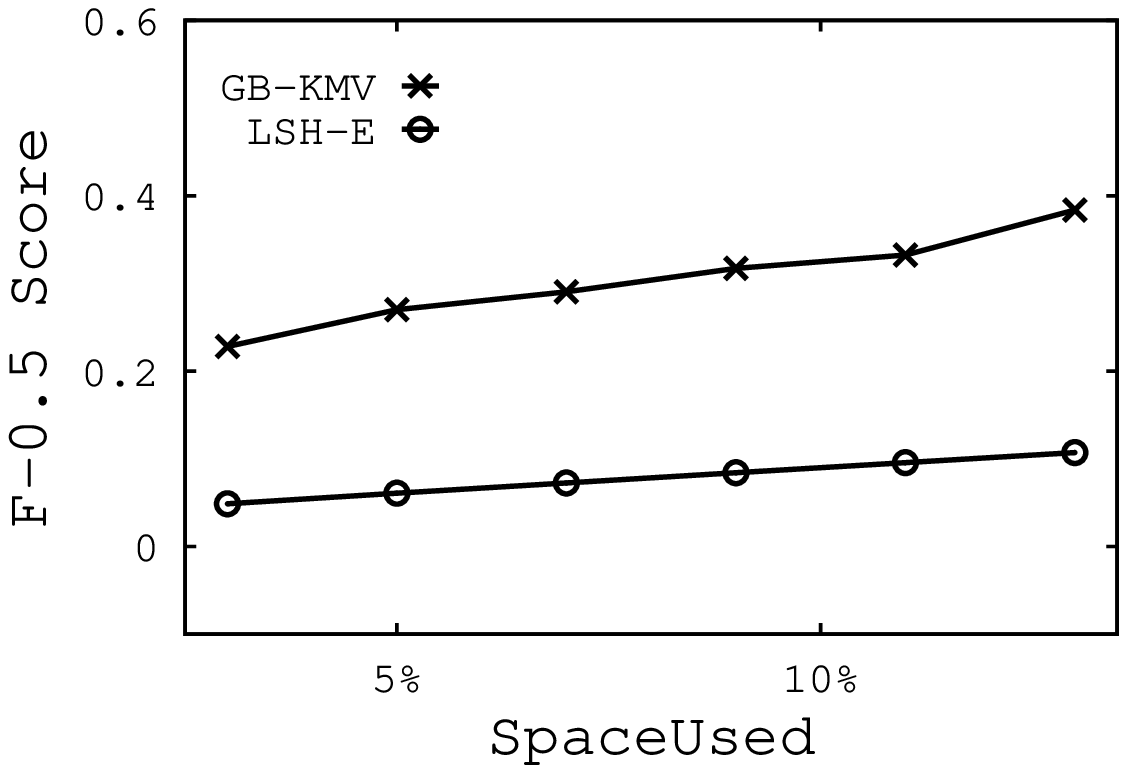}}
\vspace{-0.3cm}
\centering
\caption{\small Accuracy versus Space on REUTERS}
\vspace{-4mm}
\label{fig:space_accu_reuters}
\end{figure*}

\begin{figure*}[hbt]
\centering
\subfigure{\includegraphics[width=0.24\linewidth]{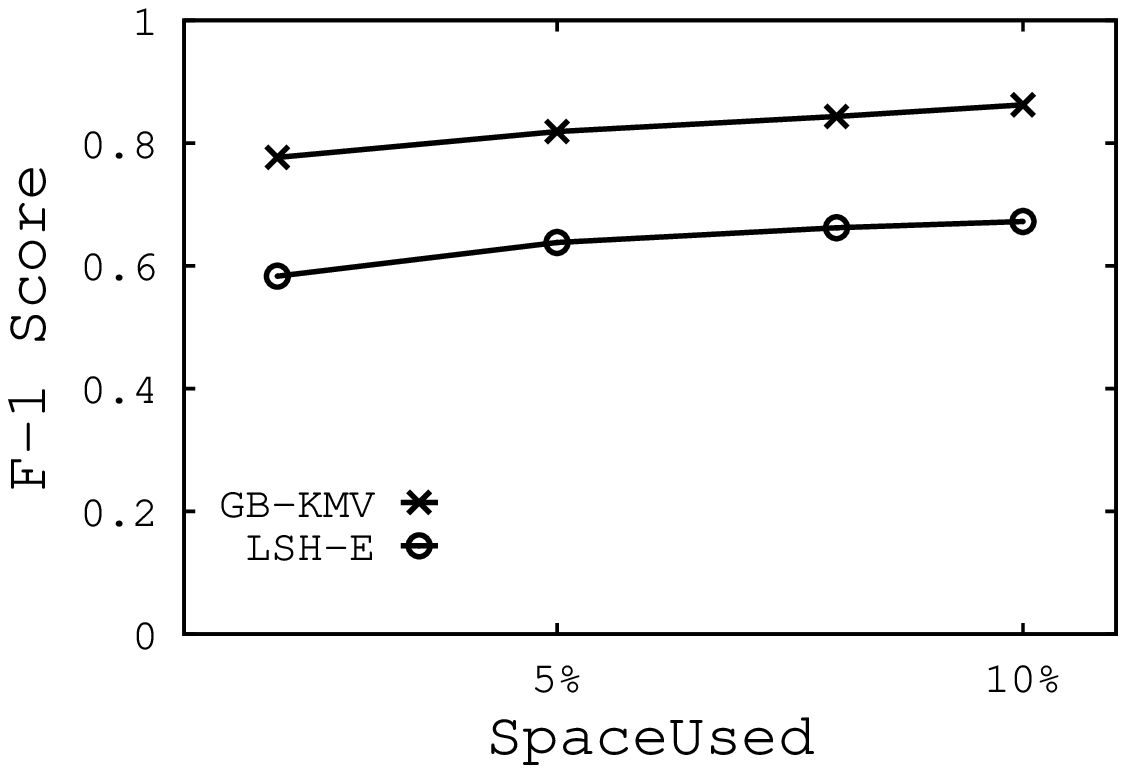}}
\subfigure{\includegraphics[width=0.24\linewidth]{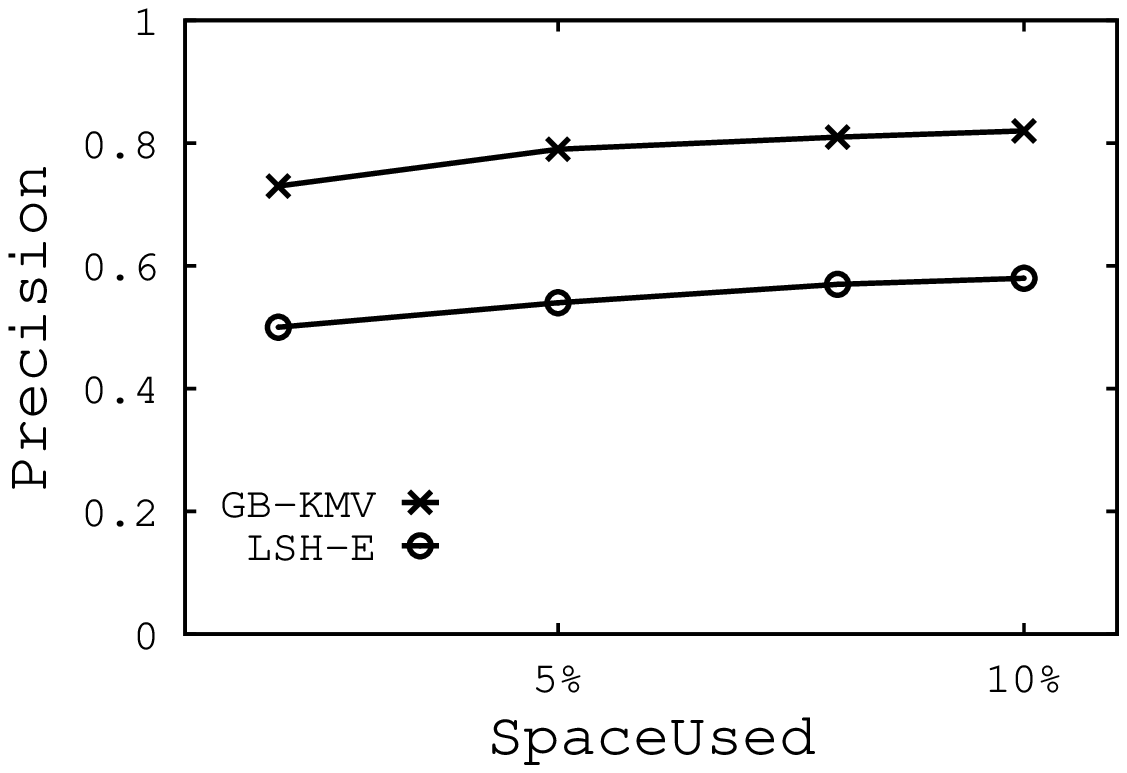}}
\subfigure{\includegraphics[width=0.24\linewidth]{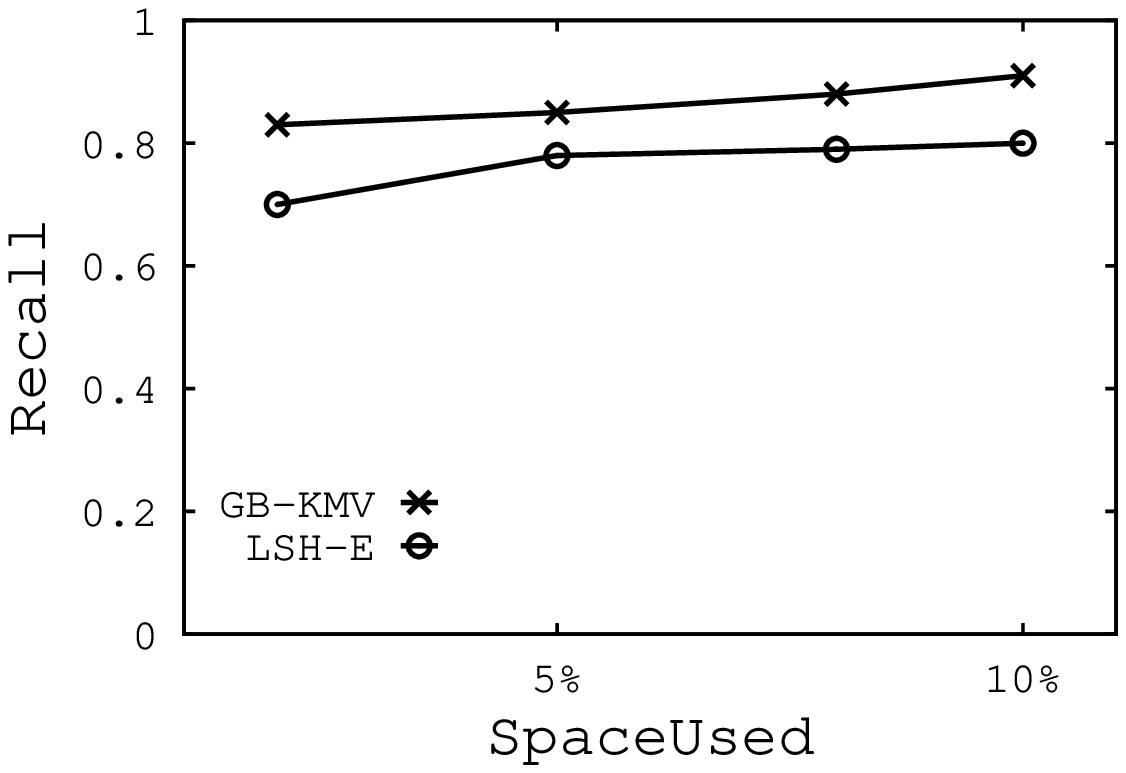}}
\subfigure{\includegraphics[width=0.24\linewidth]{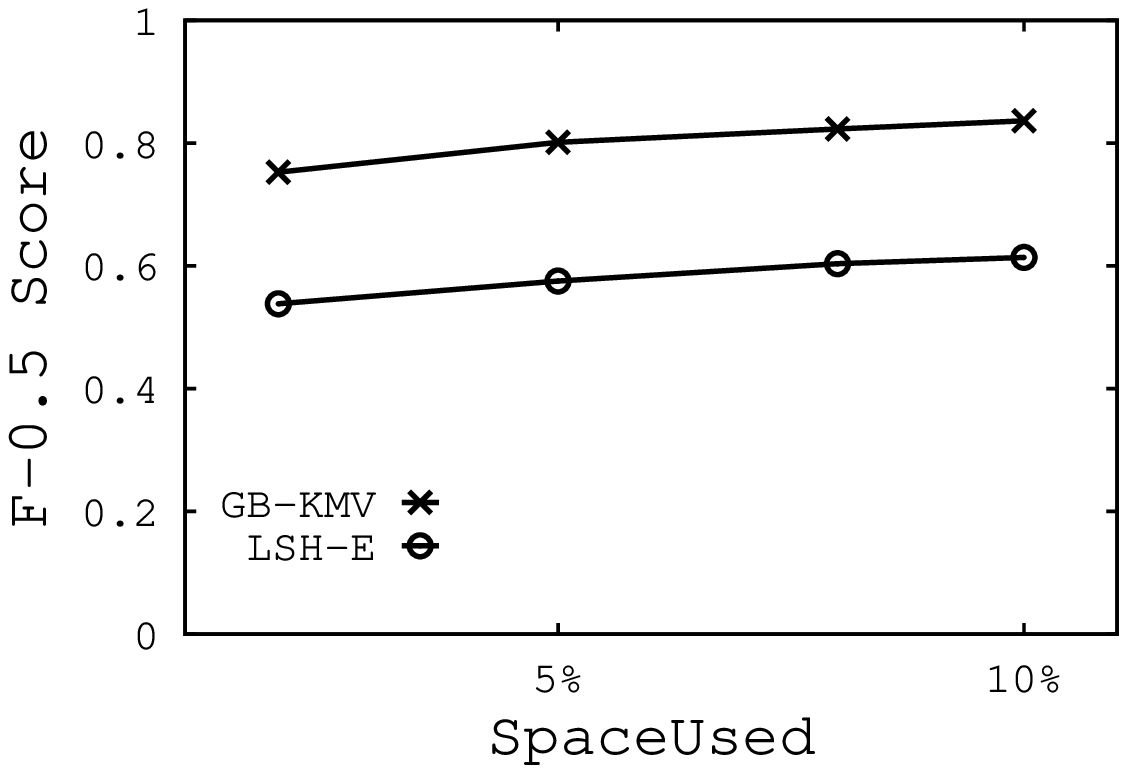}}
\vspace{-0.3cm}
\centering
\caption{\small Accuracy versus Space on WEBSPAM}
\vspace{-4mm}
\label{fig:space_accu_webspam}
\end{figure*}

\begin{figure*}[hbt]
\centering
\subfigure{\includegraphics[width=0.24\linewidth]{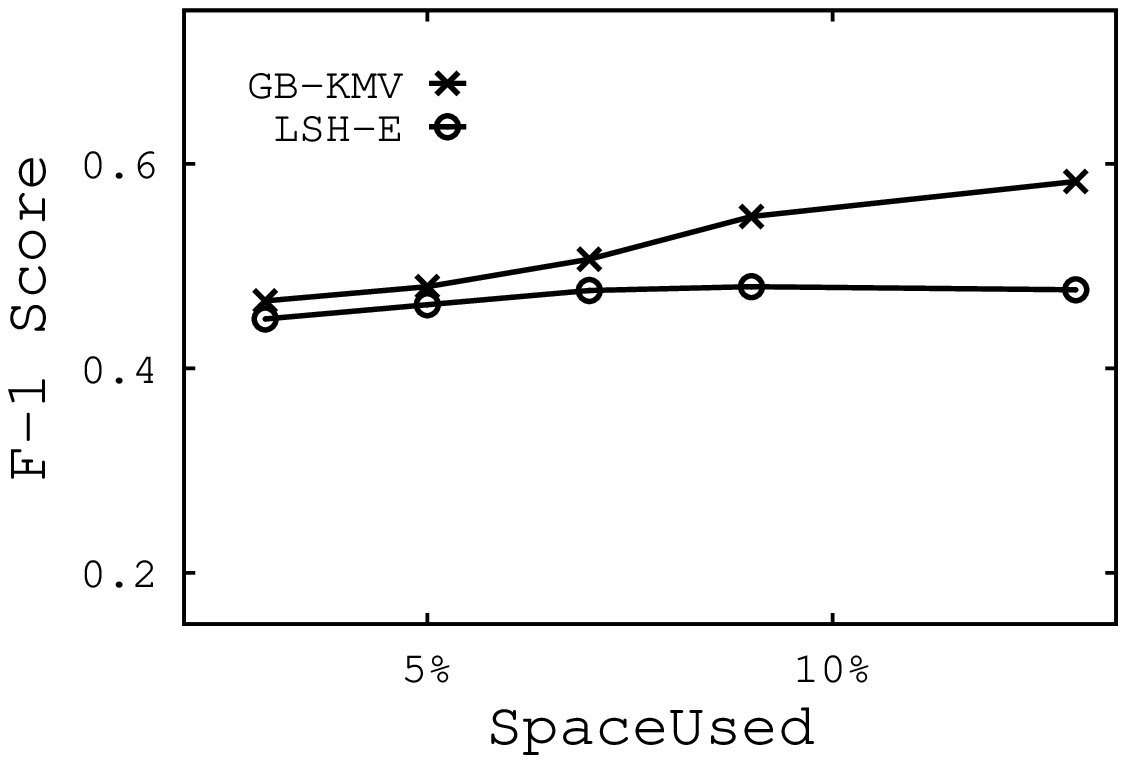}}
\subfigure{\includegraphics[width=0.24\linewidth]{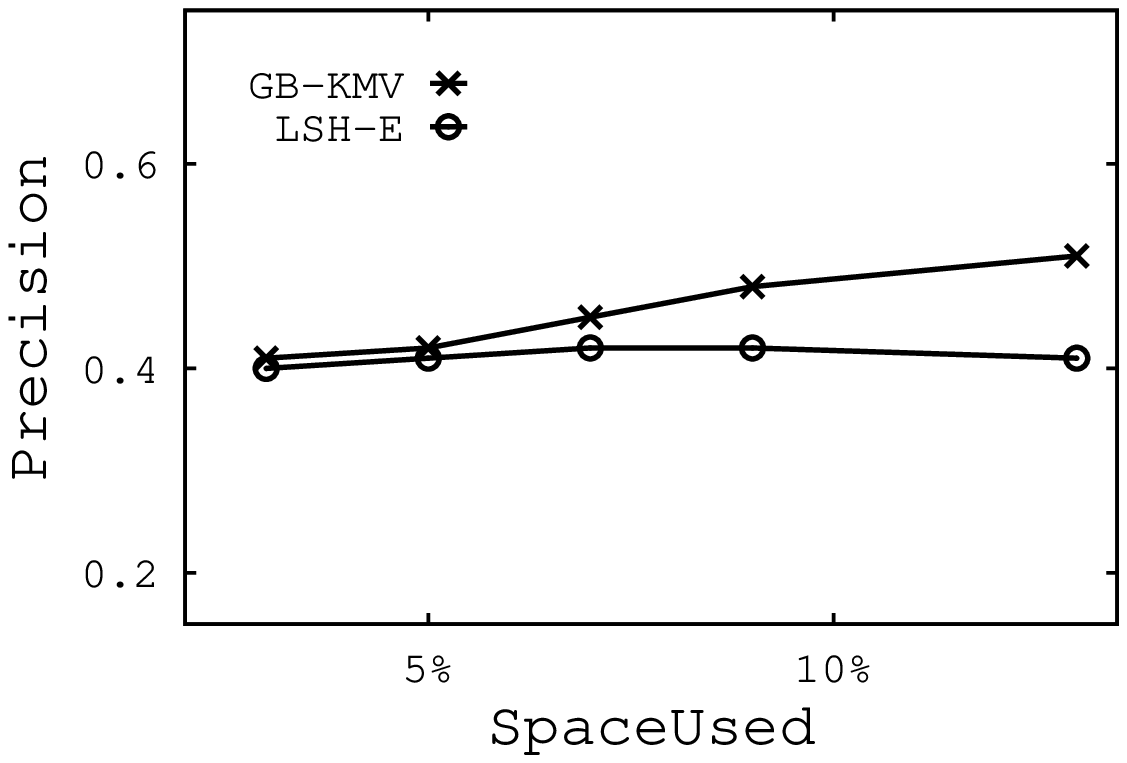}}
\subfigure{\includegraphics[width=0.24\linewidth]{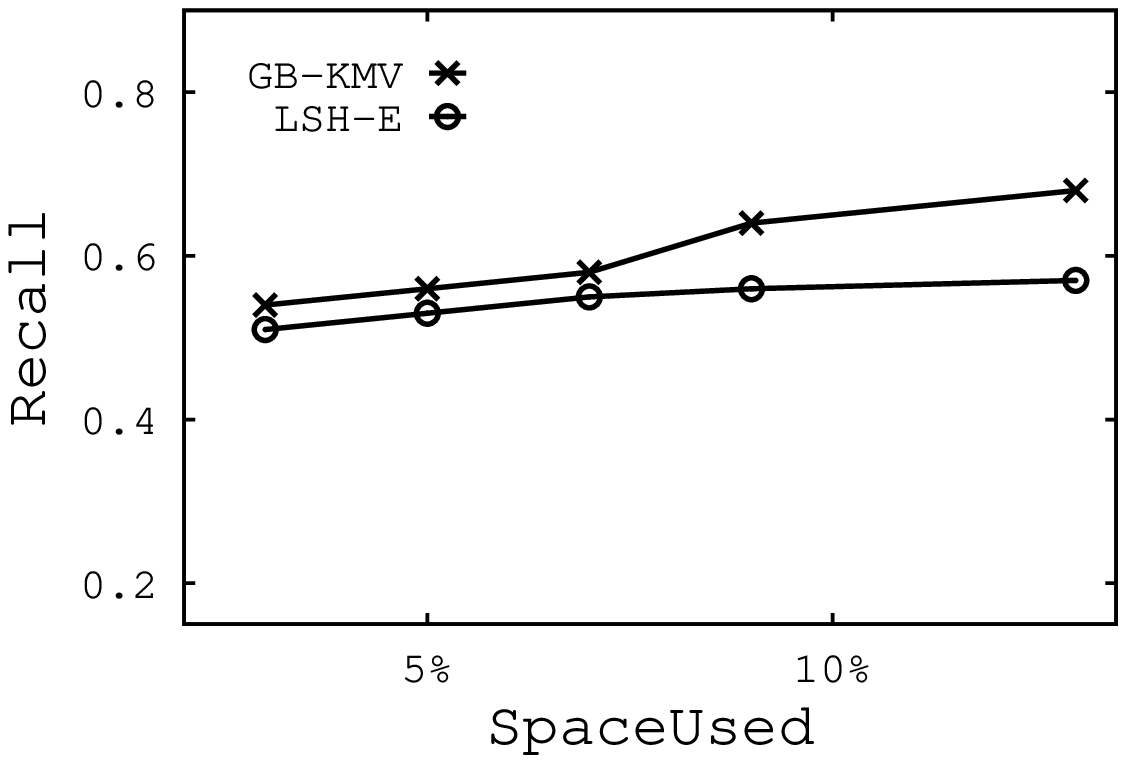}}
\subfigure{\includegraphics[width=0.24\linewidth]{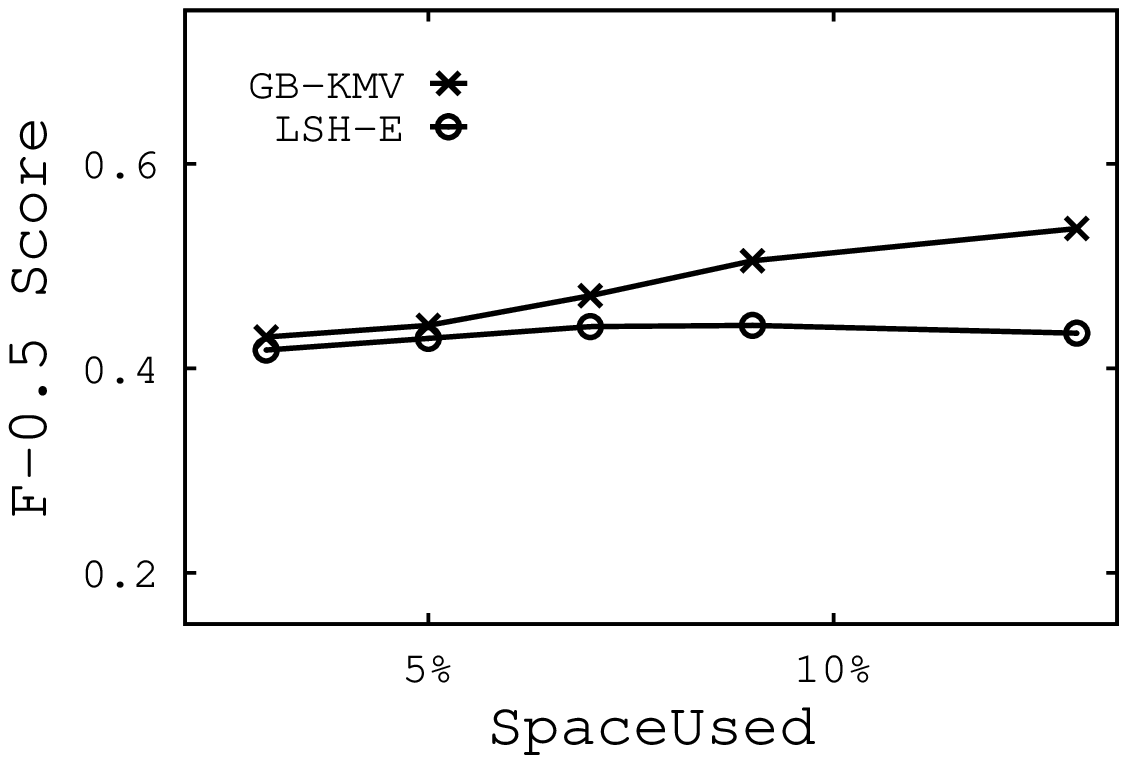}}
\vspace{-0.3cm}
\centering
\caption{\small Accuracy versus Space on WDC}
\vspace{-4mm}
\label{fig:space_accu_webspam}
\end{figure*}

\begin{figure}[hbt]
\centering
\subfigure[NETFLIX]{\includegraphics[width=0.48\linewidth]{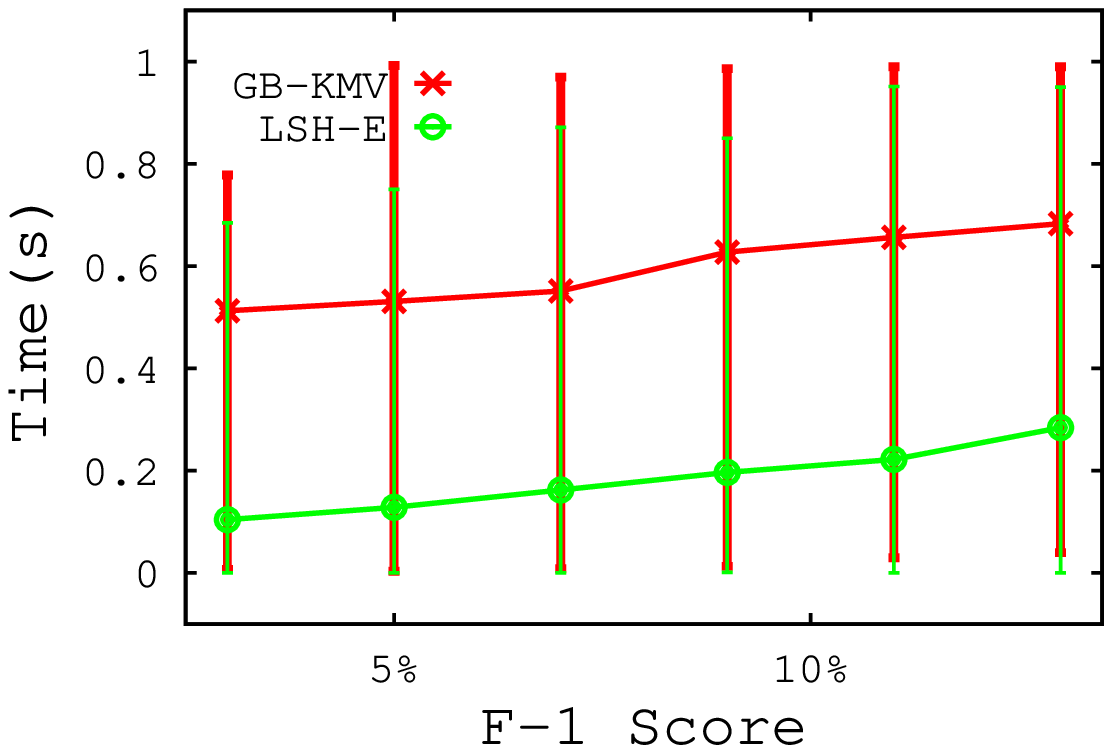}}
\subfigure[DELIC]{\includegraphics[width=0.48\linewidth]{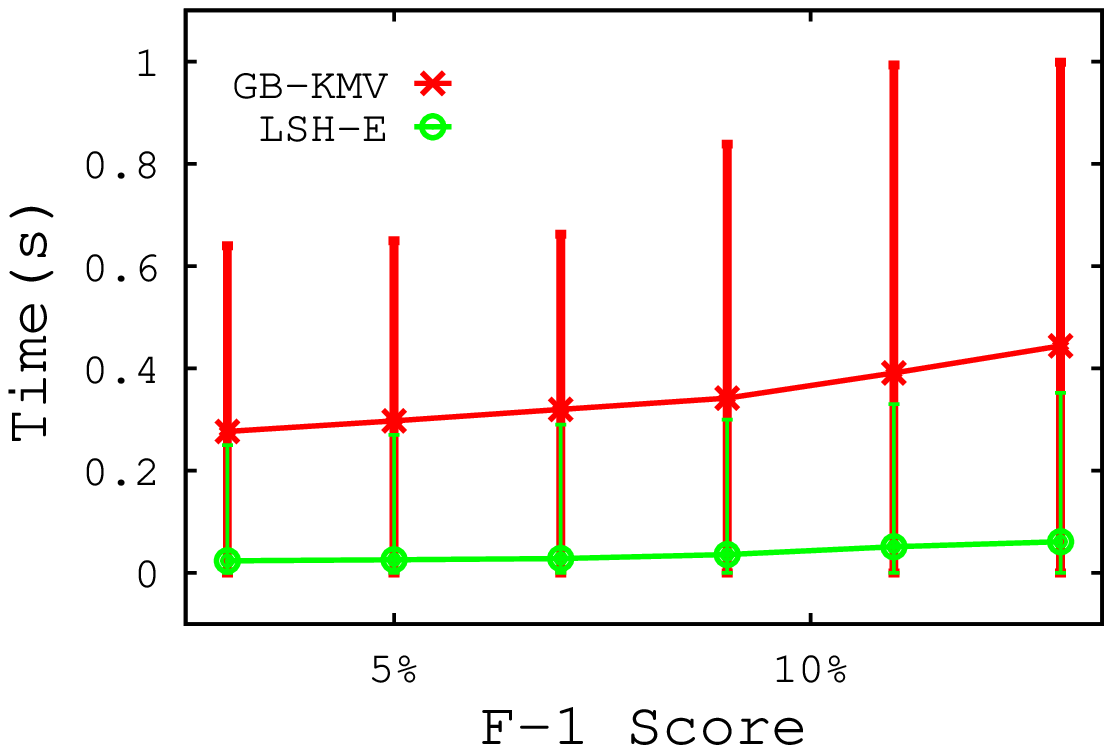}}
\subfigure[COD]{\includegraphics[width=0.48\linewidth]{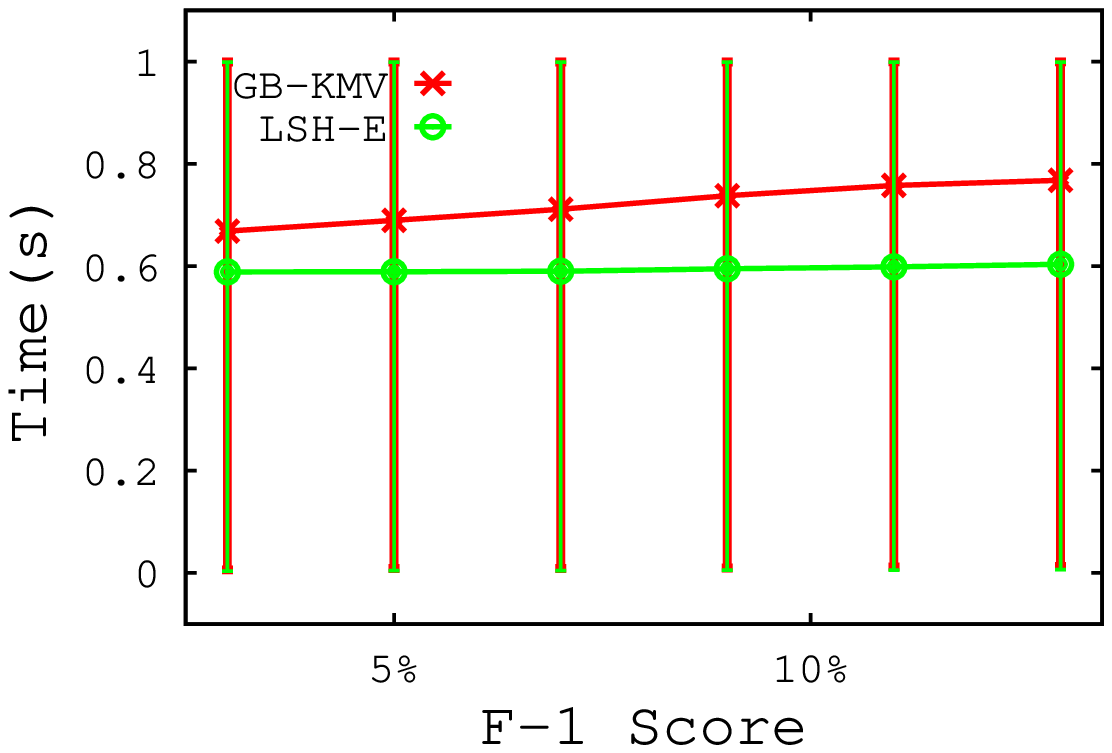}}
\subfigure[ENRON]{\includegraphics[width=0.48\linewidth]{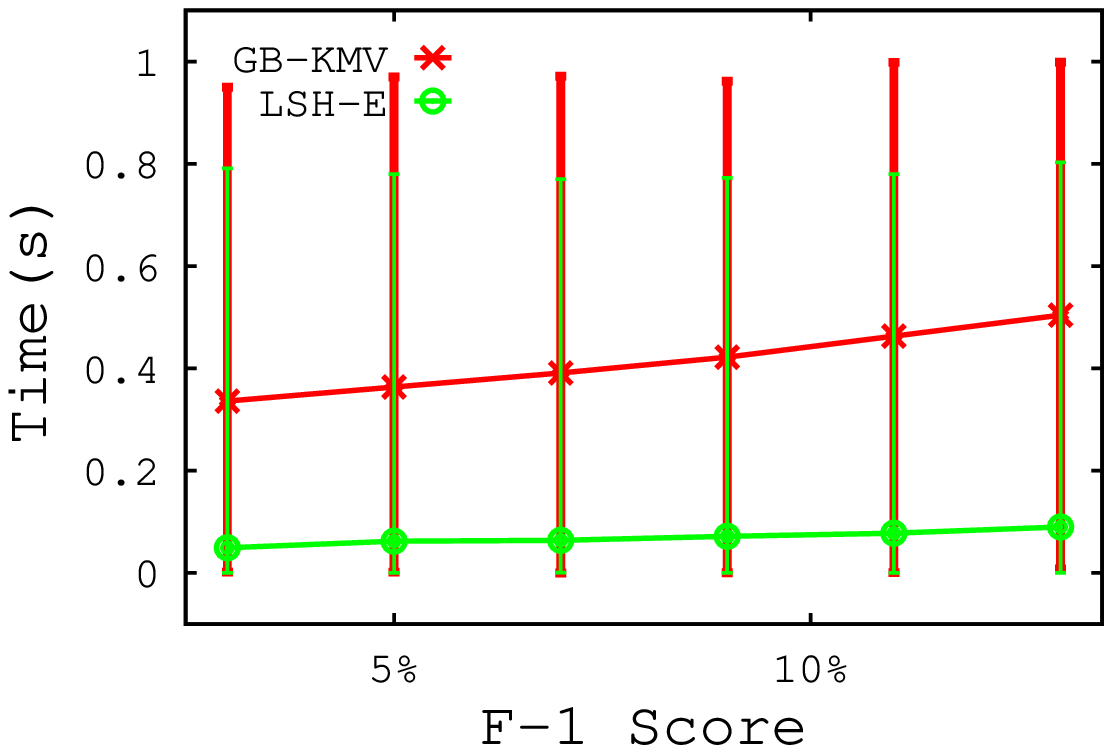}}
\subfigure[REUTERS]{\includegraphics[width=0.48\linewidth]{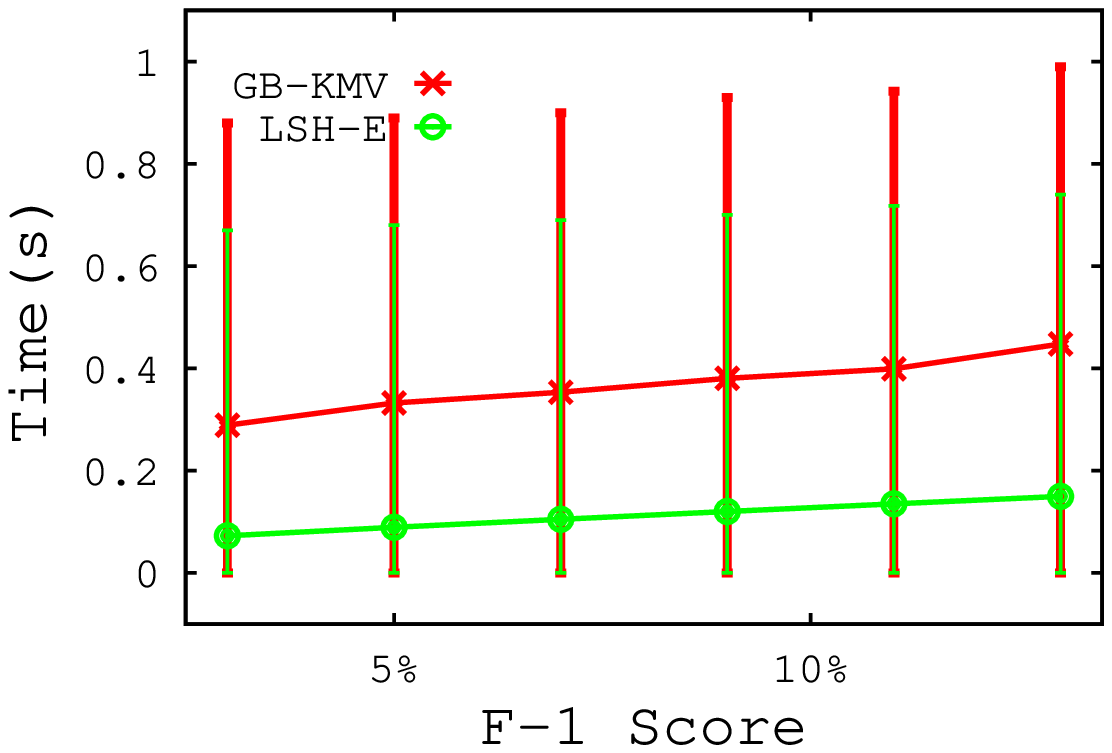}}
\subfigure[WEBSPAM]{\includegraphics[width=0.48\linewidth]{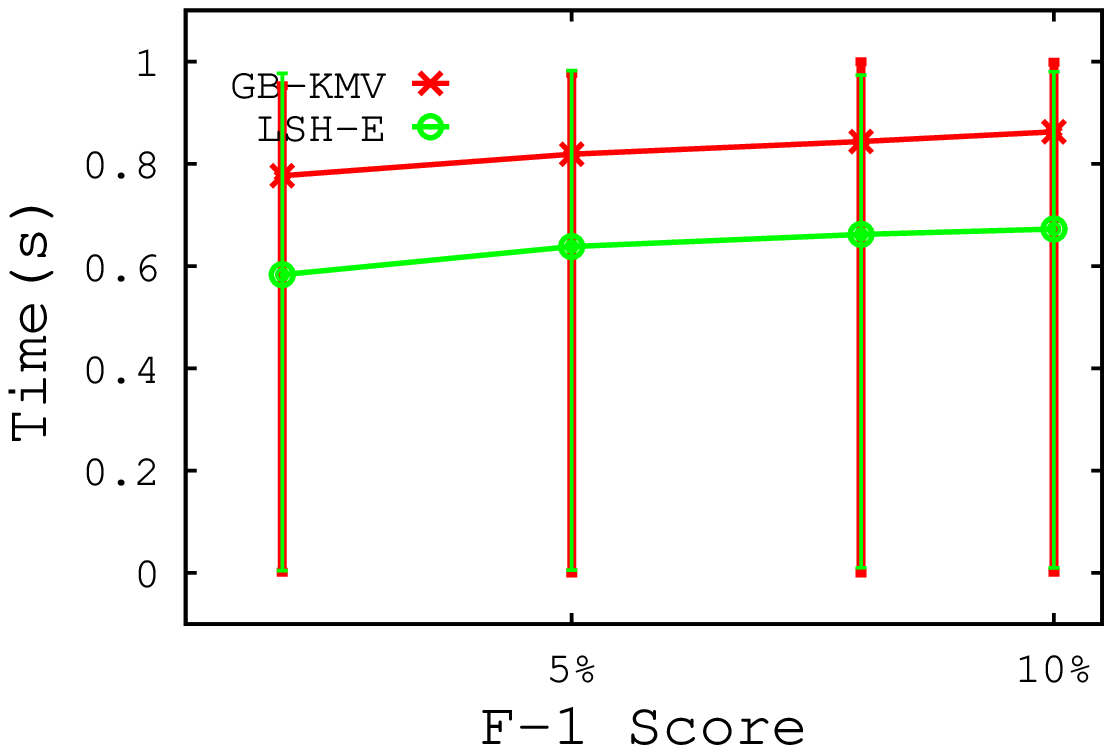}}
\subfigure[WDC]{\includegraphics[width=0.48\linewidth]{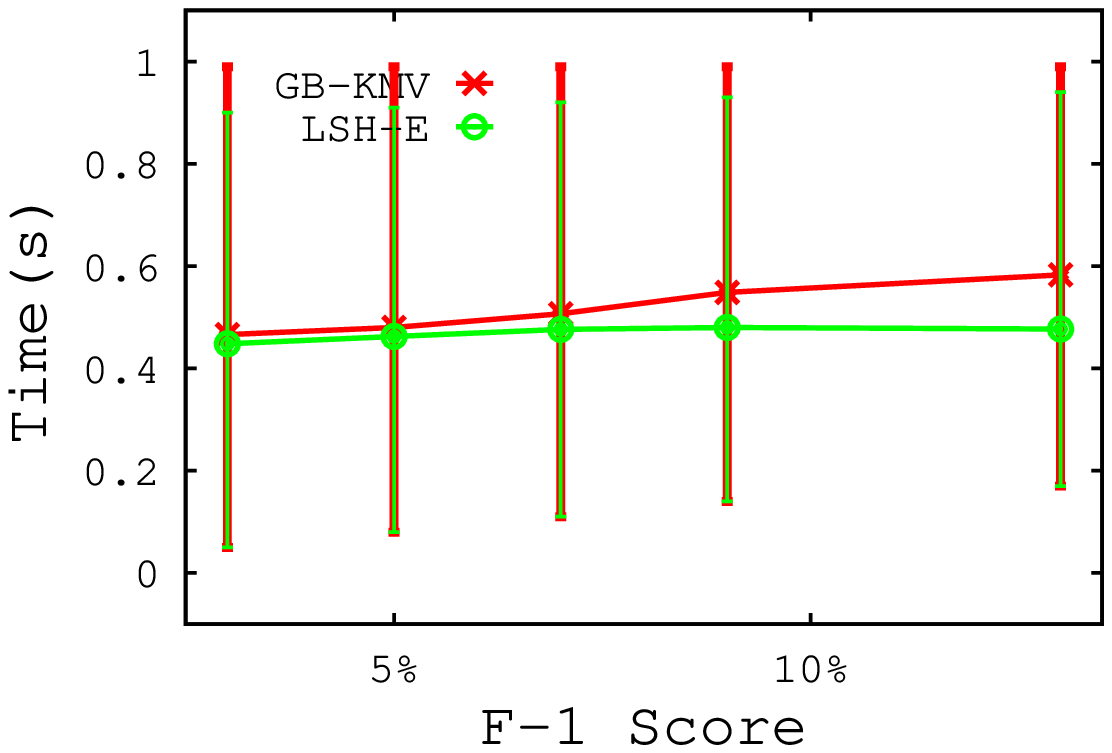}}

\vspace{-0.3cm}
\centering
\caption{\small The distribution of Accuracy}
\vspace{-1mm}
\label{fig:distribution_accuracy}
\end{figure}

\begin{figure}[hbt]
\centering
\subfigure[NETFLIX]{\includegraphics[width=0.48\linewidth]{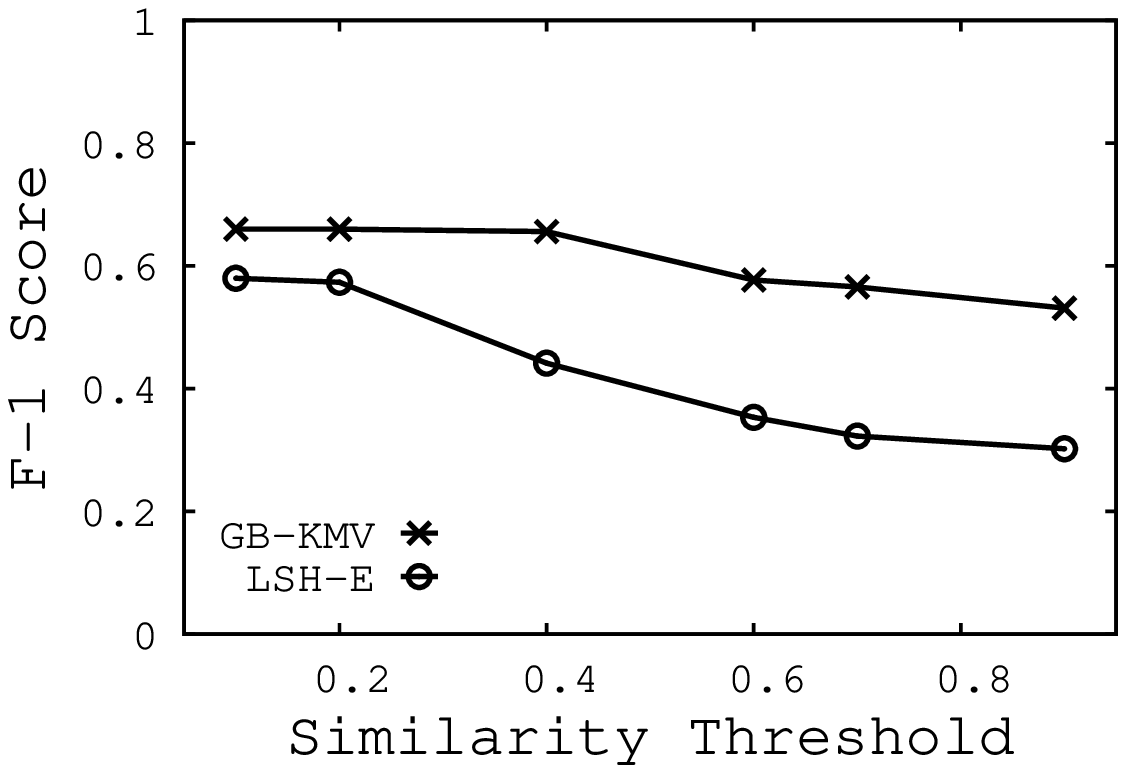}}
\subfigure[DELIC]{\includegraphics[width=0.48\linewidth]{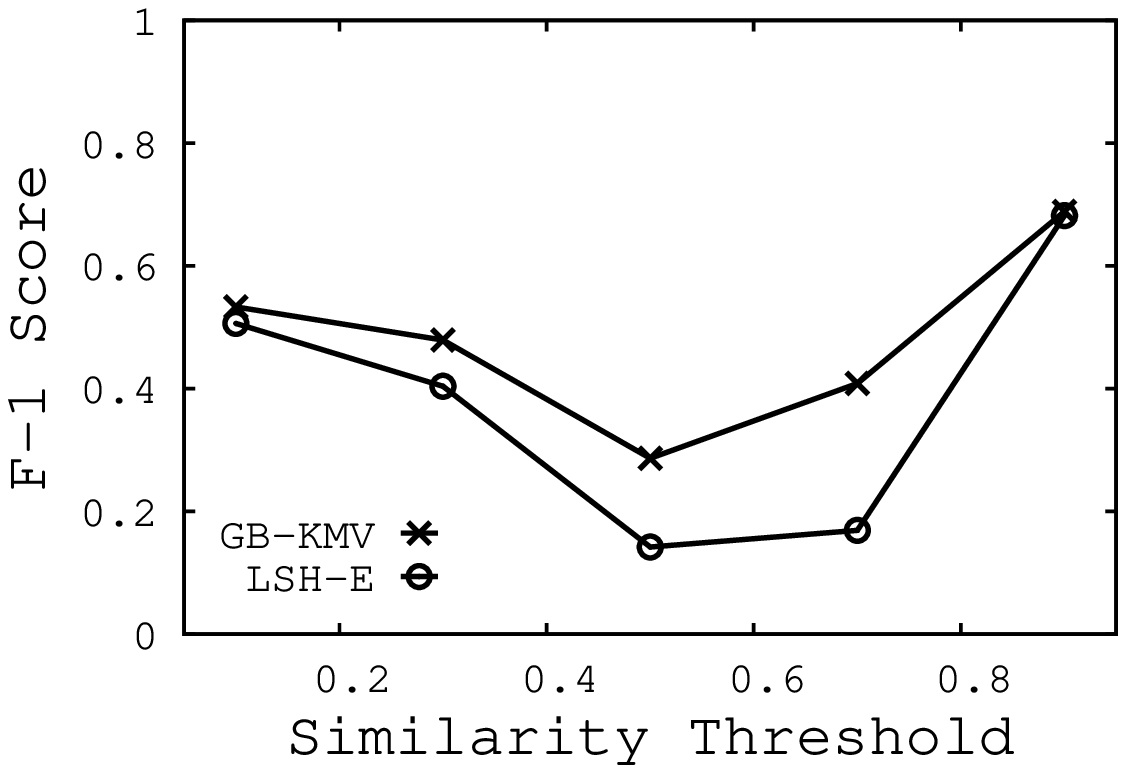}}
\subfigure[COD]{\includegraphics[width=0.48\linewidth]{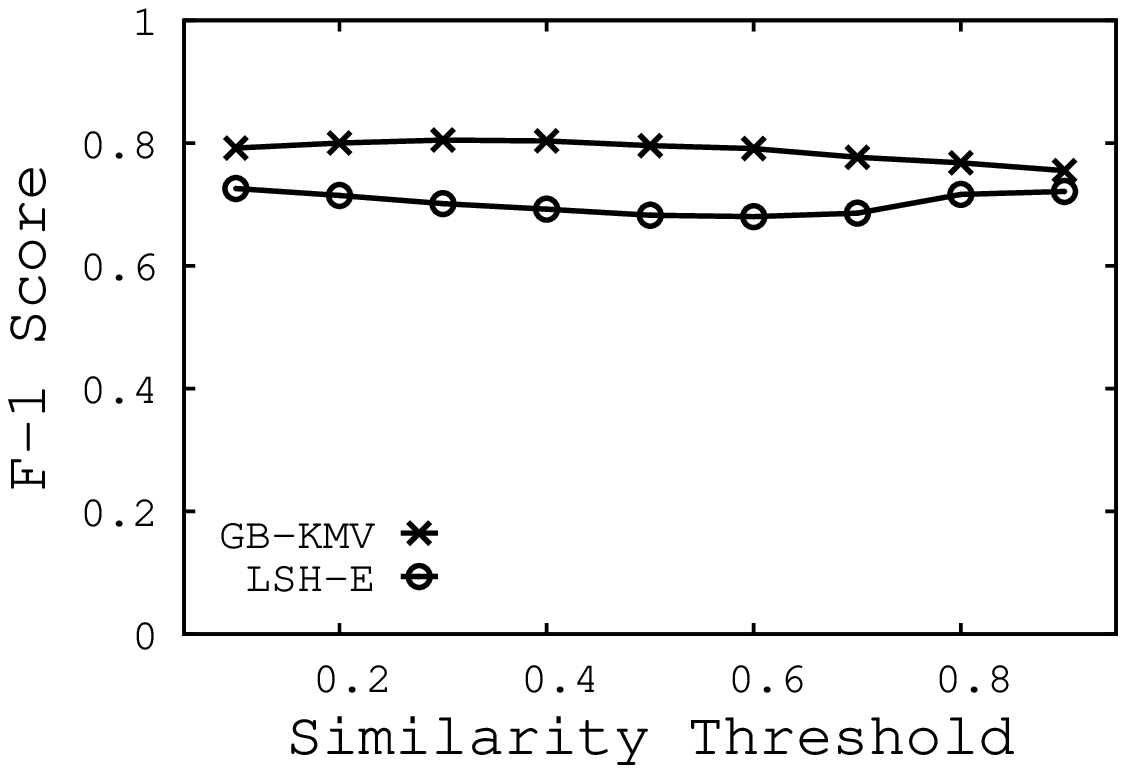}}
\subfigure[ENRON]{\includegraphics[width=0.48\linewidth]{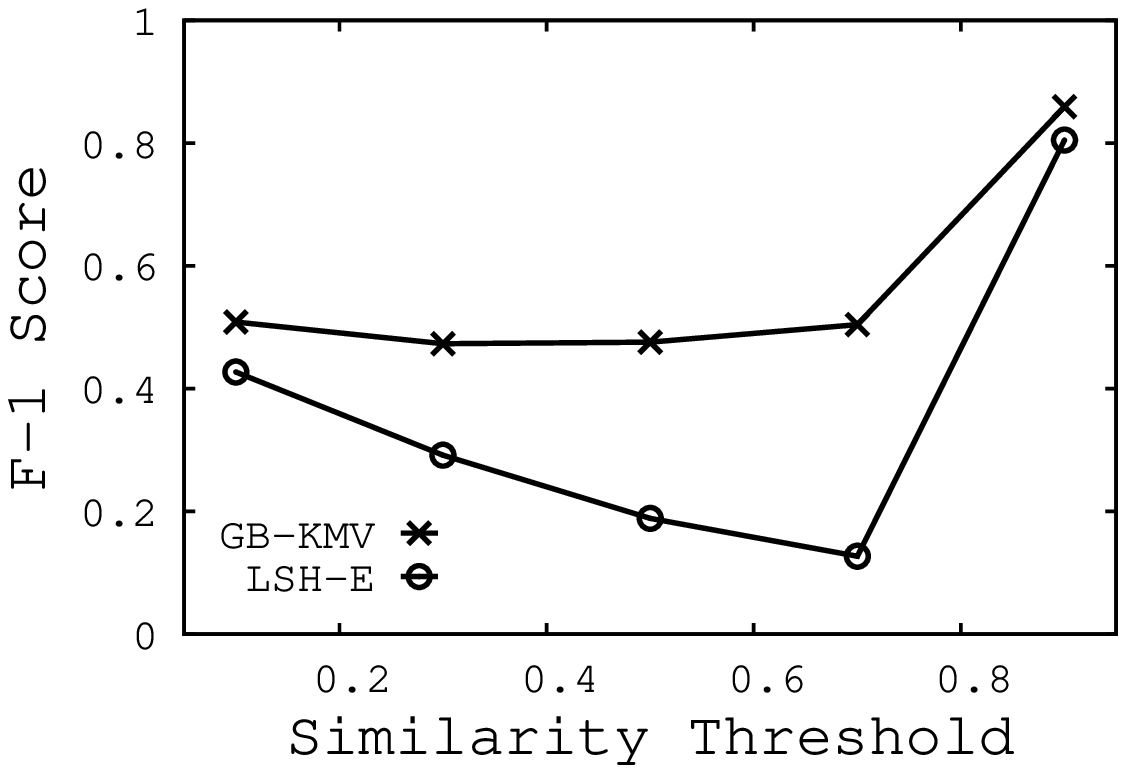}}
\subfigure[REUTERS]{\includegraphics[width=0.48\linewidth]{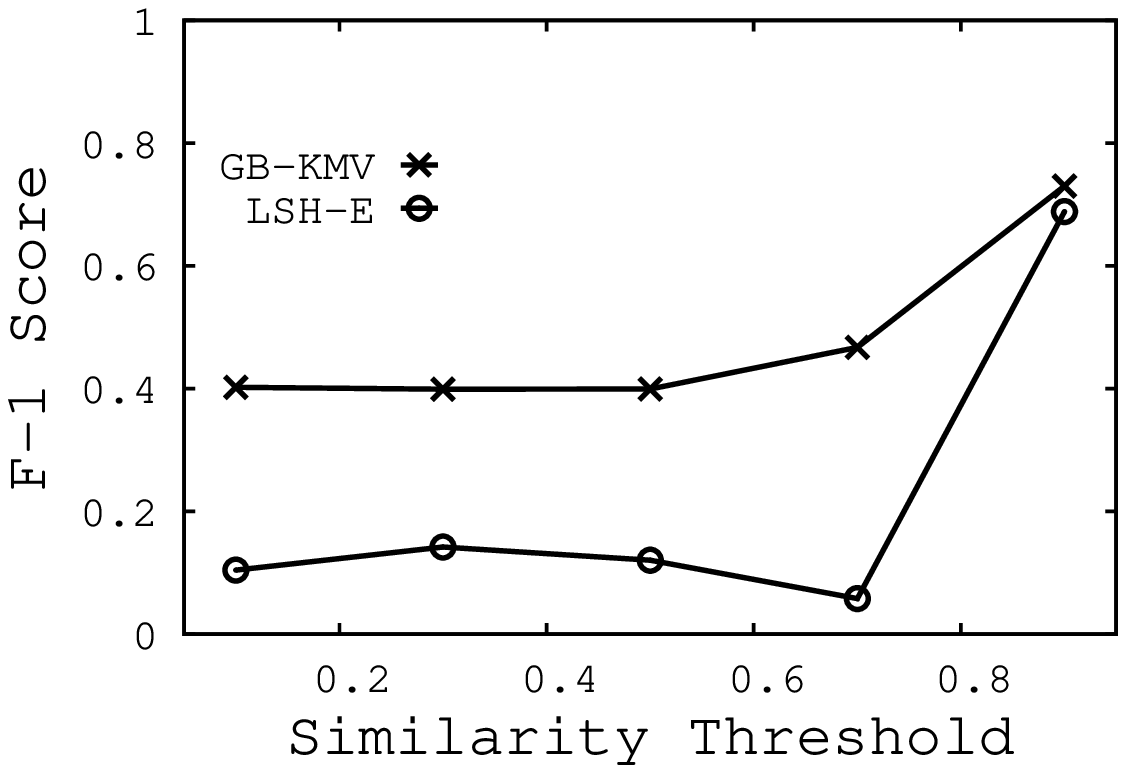}}
\subfigure[WEBSPAM]{\includegraphics[width=0.48\linewidth]{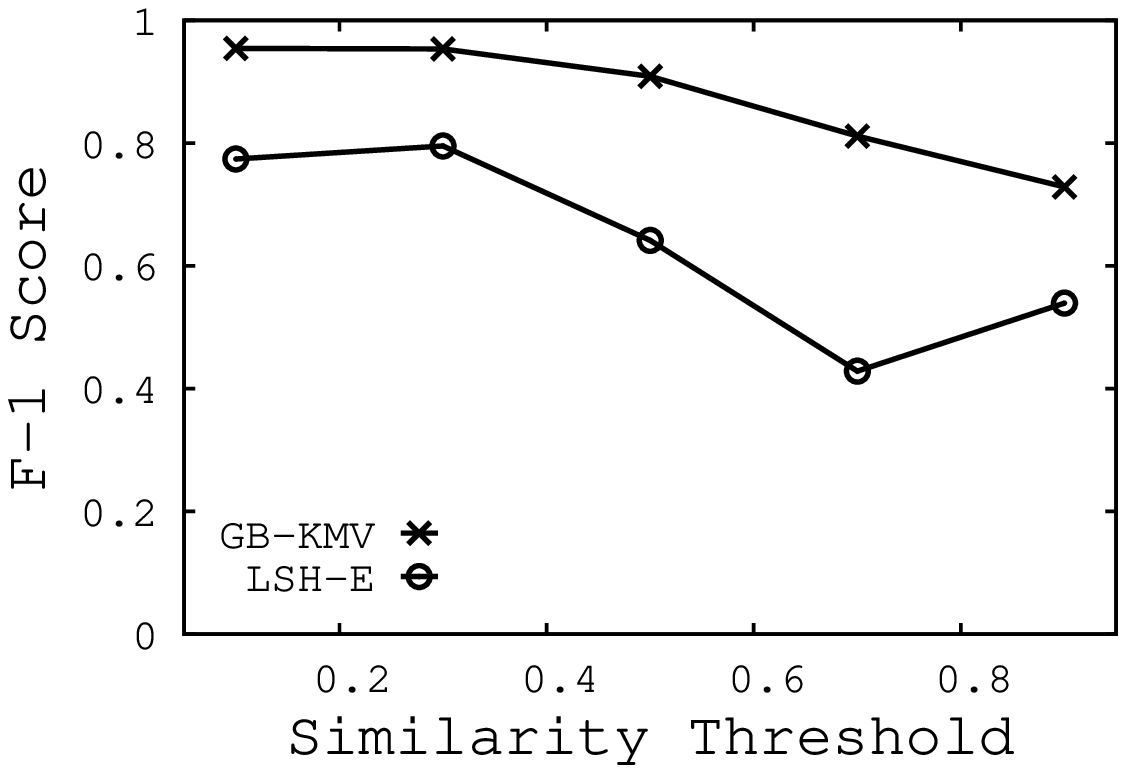}}
\subfigure[WDC]{\includegraphics[width=0.48\linewidth]{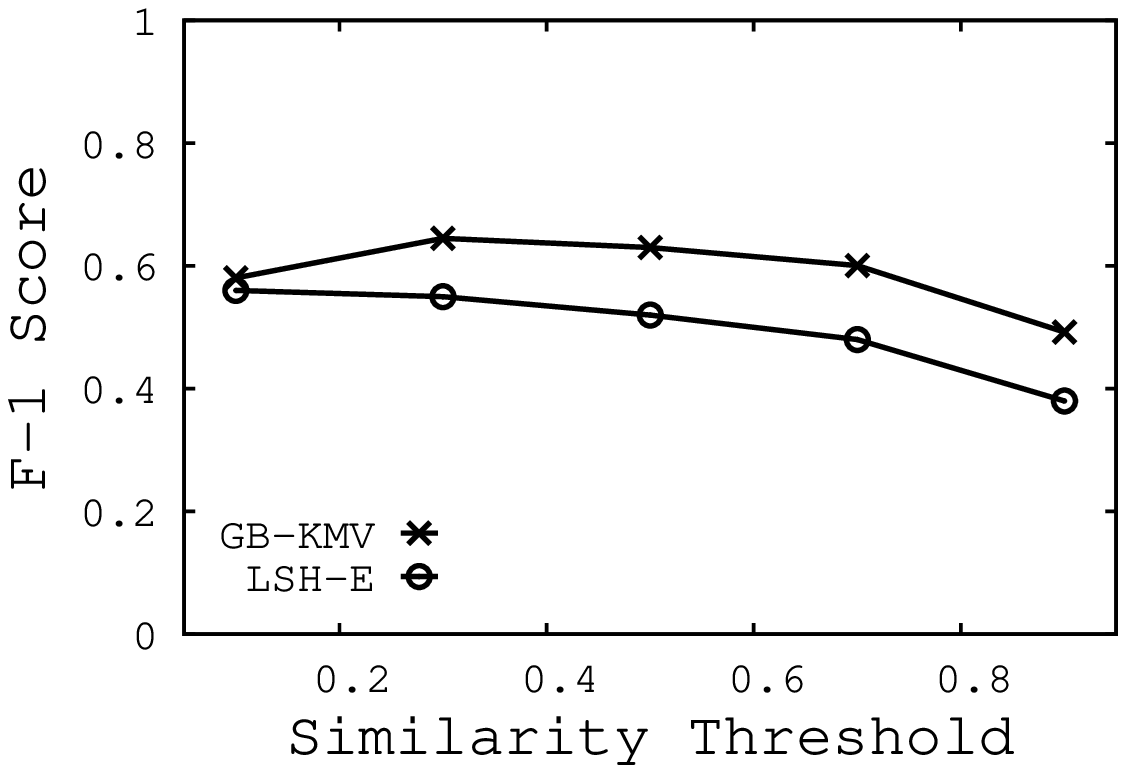}}

\vspace{-0.3cm}
\centering
\caption{\small Accuracy versus Similarity threshold}
\vspace{-1mm}
\label{fig:threshold accu}
\end{figure}

\begin{figure}[hbt]
\centering
\subfigure{\includegraphics[width=0.48\linewidth]{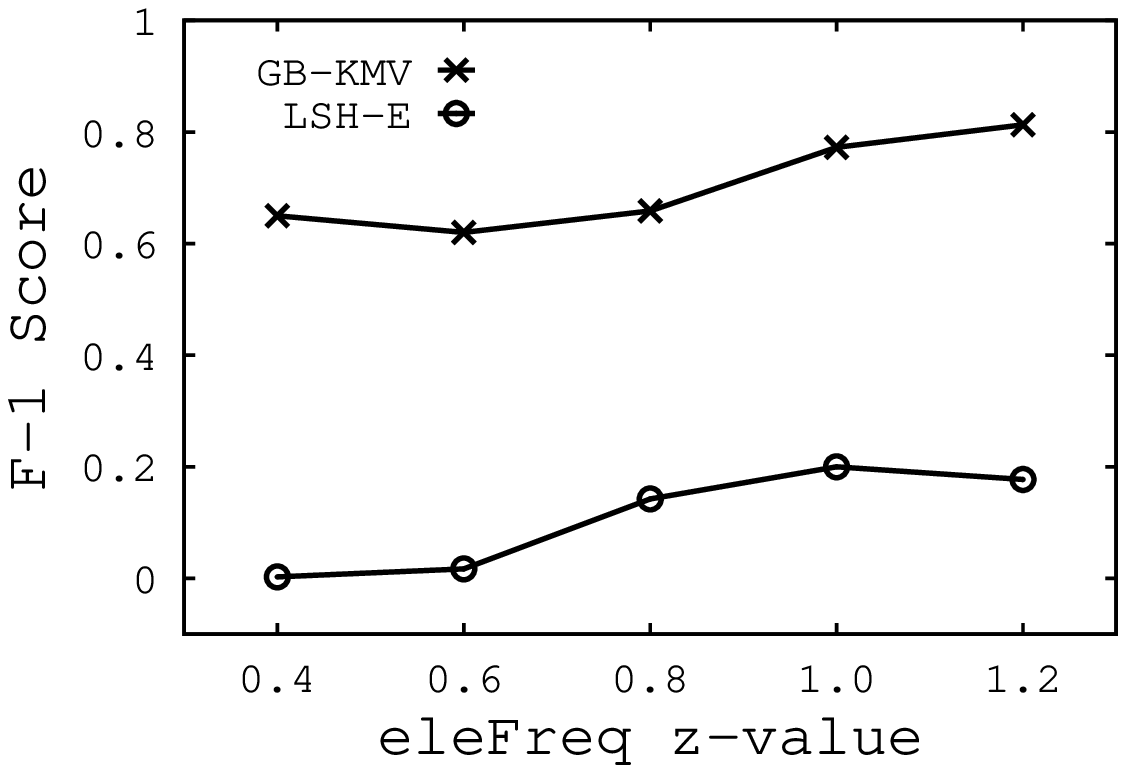}}
\subfigure{\includegraphics[width=0.48\linewidth]{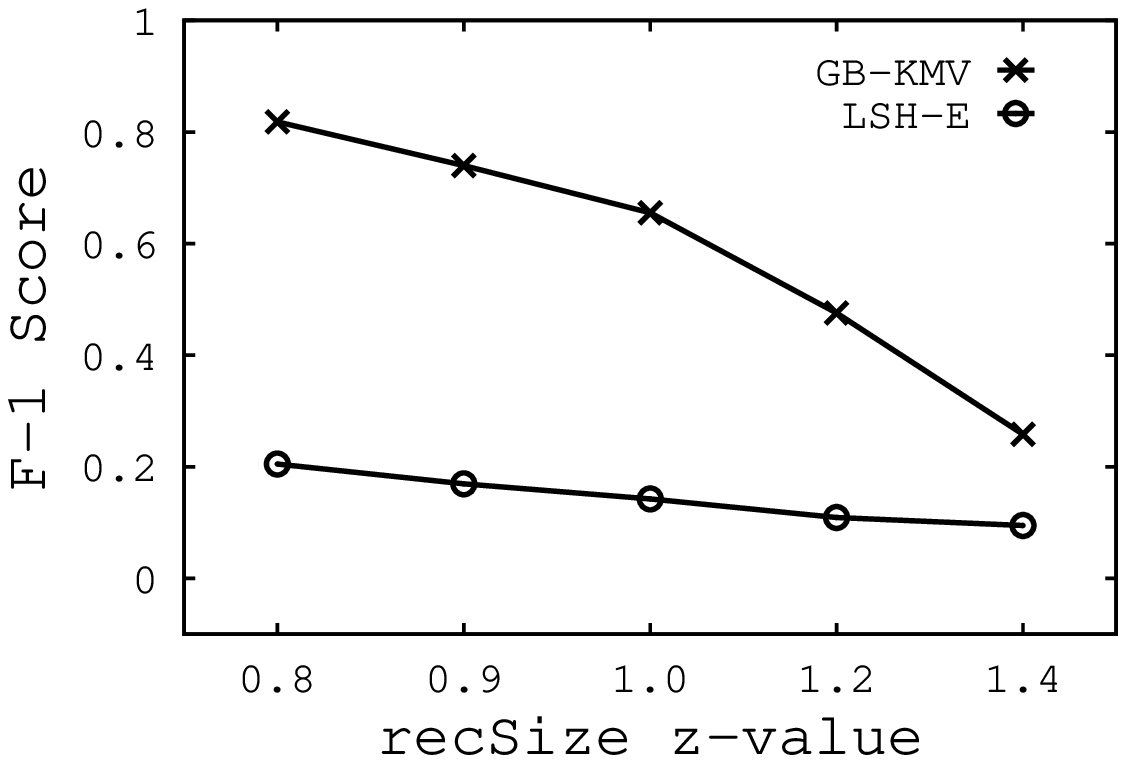}}
\vspace{-0.3cm}
\centering
\caption{\small EleFreq $z$-value varying from 0.4 to 1.2 with recSize $z$-value 1.0; recSize $z$-value varying from 0.8 to 1.4 with eleFreq $z$-value 0.8}
\vspace{-1mm}
\label{fig:syn accu}
\end{figure}
\vspace{-2mm}

\begin{figure}[hbt]
\centering
\subfigure[COD]{\includegraphics[width=0.48\linewidth]{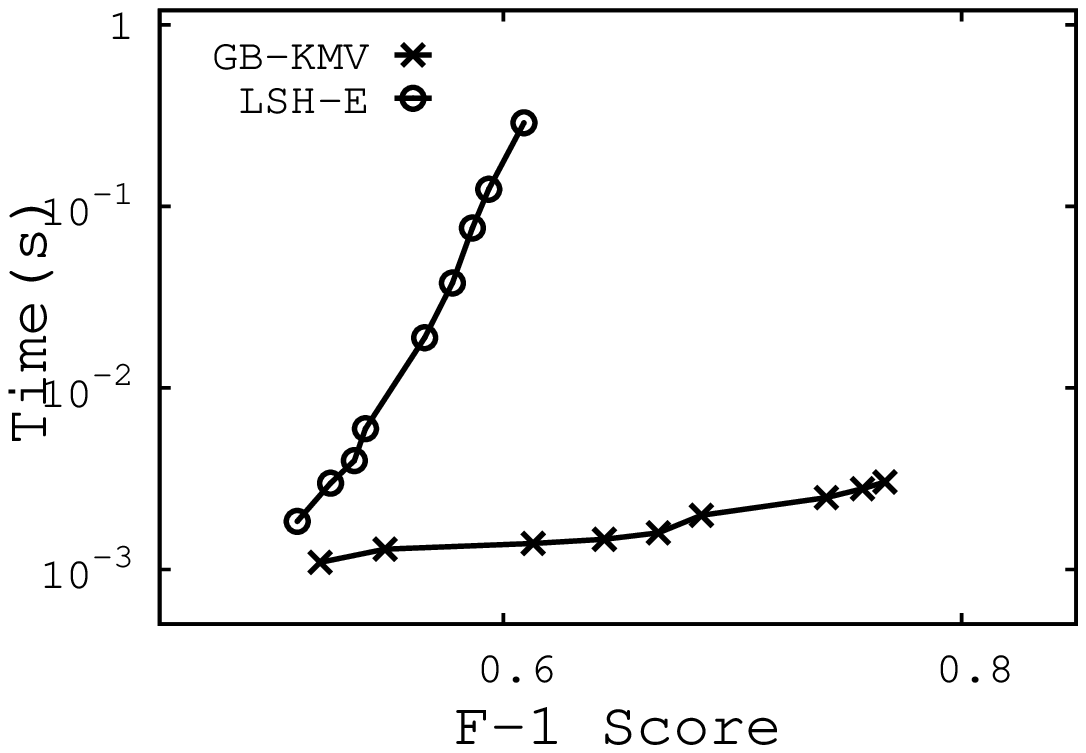}}
\subfigure[DELIC]{\includegraphics[width=0.48\linewidth]{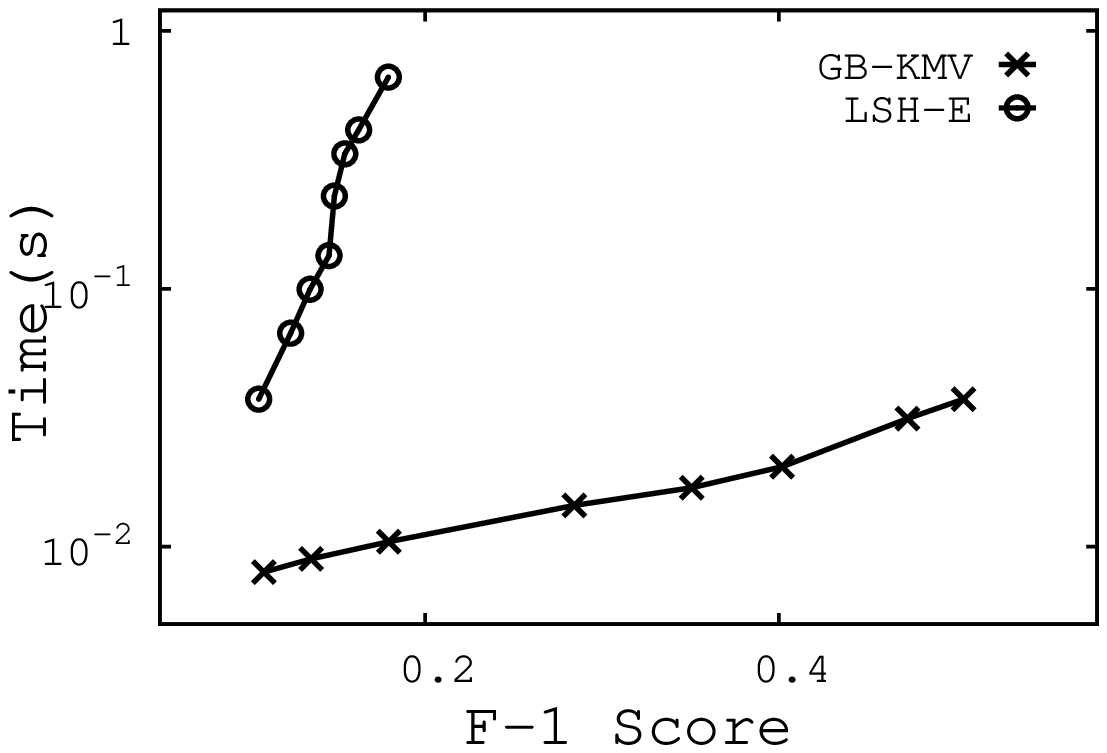}}
\subfigure[ENRON]{\includegraphics[width=0.48\linewidth]{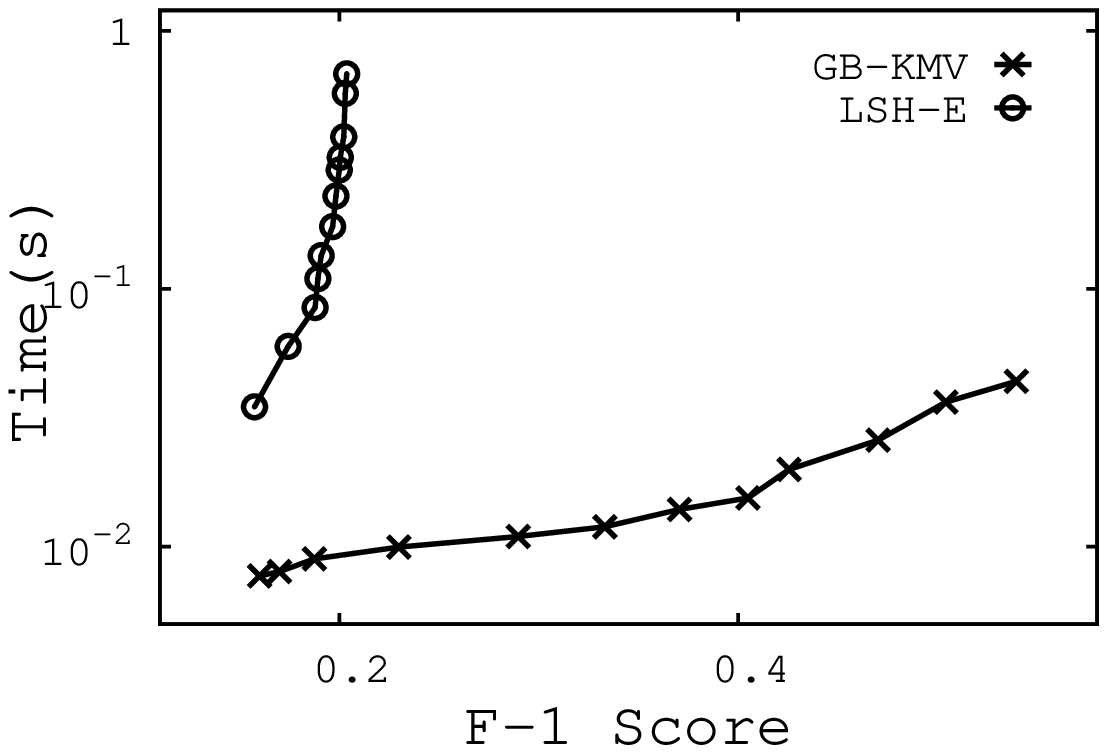}}
\subfigure[NETFLIX]{\includegraphics[width=0.48\linewidth]{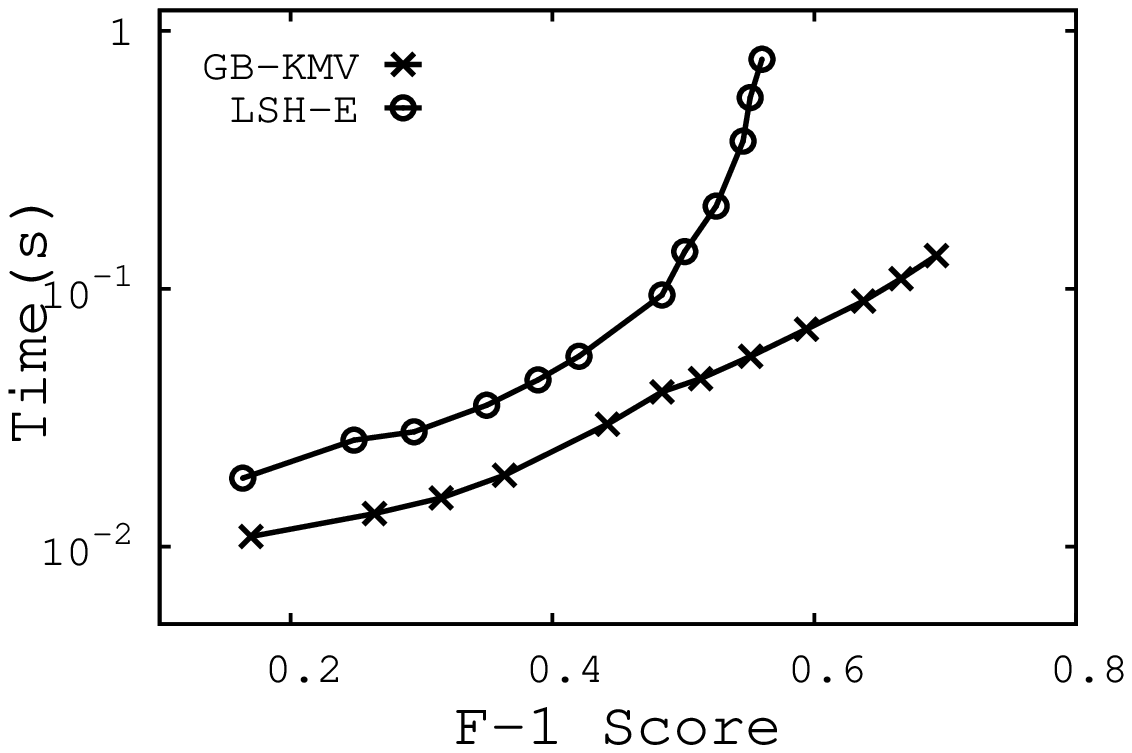}}
\subfigure[REUTERS]{\includegraphics[width=0.48\linewidth]{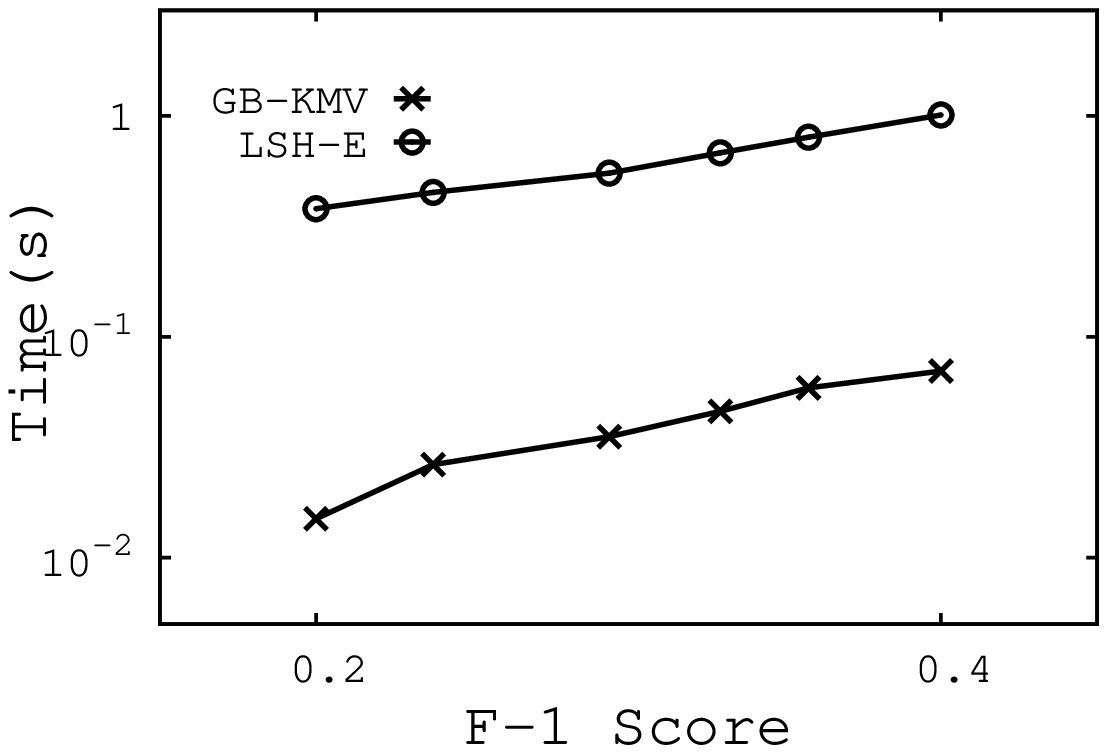}}
\subfigure[WEBSPAM]{\includegraphics[width=0.48\linewidth]{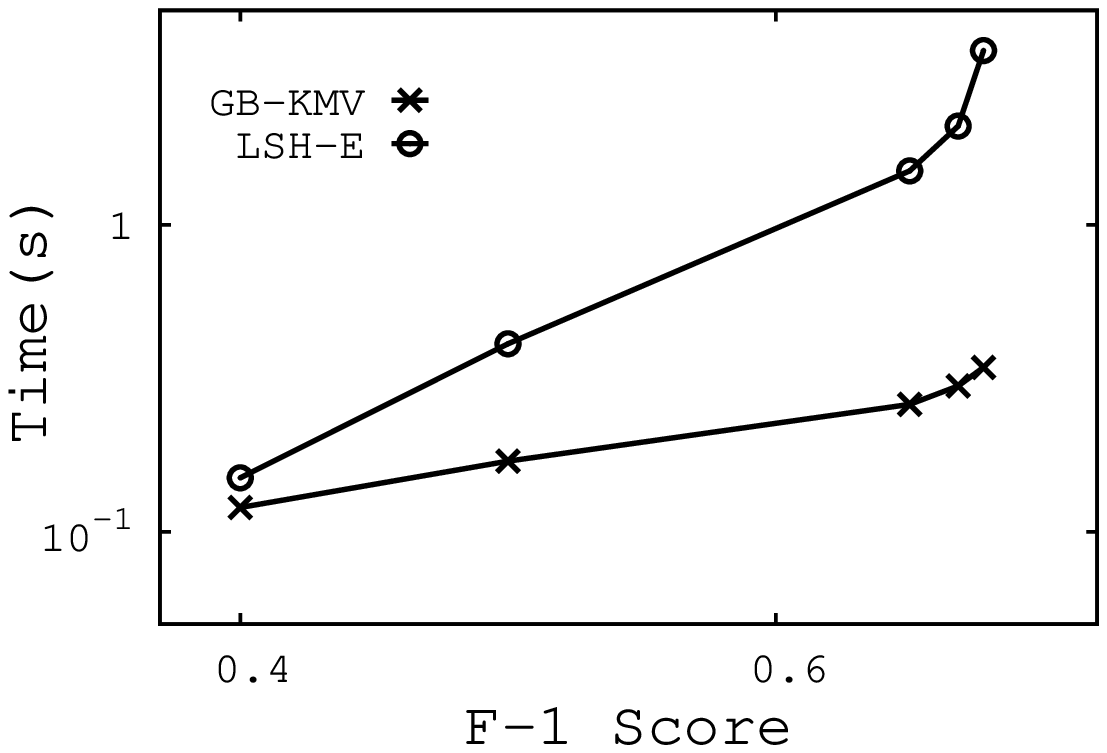}}
\subfigure[WDC]{\includegraphics[width=0.48\linewidth]{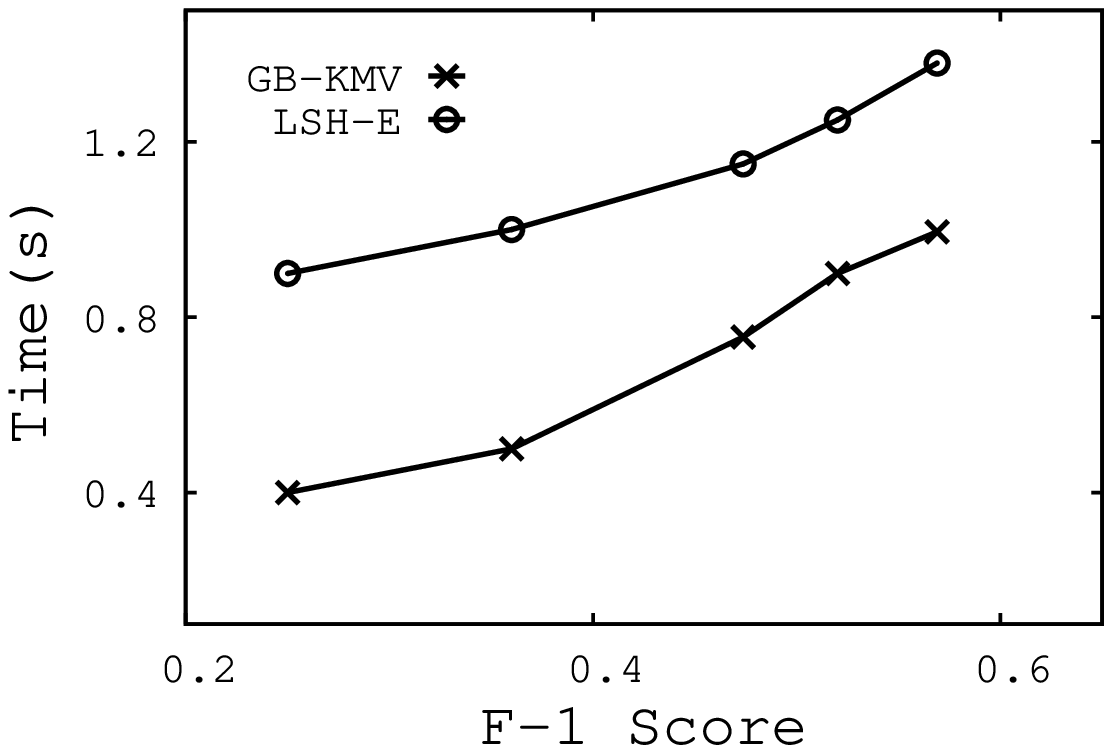}}
\vspace{-0.3cm}
\centering
\caption{\small Time versus Accuracy}
\vspace{-2mm}
\label{fig:time accu}
\end{figure}

\begin{figure}[hbt]
\centering
\subfigure{\includegraphics[width=0.8\linewidth]{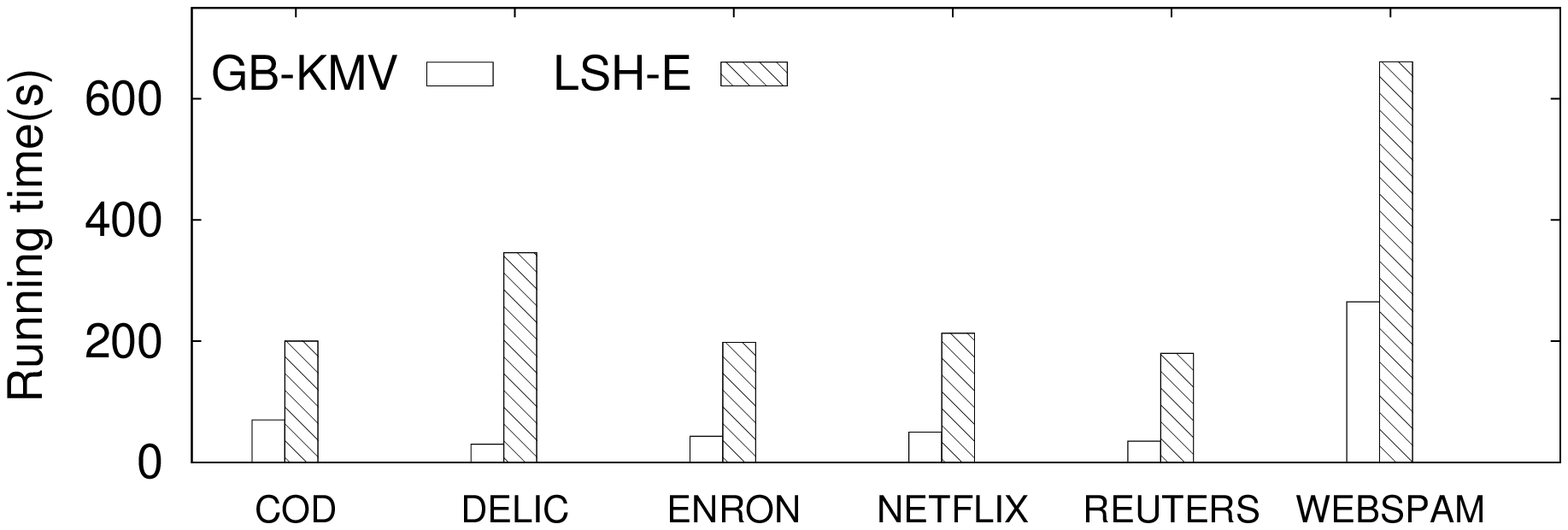}}
\vspace{-4mm}
\caption{\small Sketch Construction Time}
\vspace{-2mm}
\label{fig:index_time}
\end{figure}

\vspace{1mm}
\noindent \textbf{Approximate Algorithms.}
In the experiments, the approximate algorithms evaluated are as follows.
\begin{itemize}
  \item \textbf{\gbkmv}. Our approach proposed in Section~\ref{subsec:alg}.
  \item \textbf{\lshe}. The state-of-the-art approxiamte containment similarity search method proposed in~\cite{zhu2016lsh}.
\end{itemize}
The above two algorithms are implemented in Go programming language.
We get the source code of \lshe from~\cite{zhu2016lsh}. For \lshe, we follow the parameter setting from~\cite{zhu2016lsh}.

\vspace{1mm}
\noindent \textbf{Exact Algorithms.}
To better evaluate the proposed methods, we also compare our approximate method \gbkmv with the following
two exact containment similarity search methods.
\begin{itemize}
  \item \textbf{\ppjoin*}. We extend the prefix-filtering based method from~\cite{xiao2011efficient} to tackle the containment similarity search problem.
  \item \textbf{\freset}. The state-of-the-art exact containment similarity search method proposed in~\cite{agrawal2010indexing}.
\end{itemize}

\begin{remark}
A novel size-aware overlap set similarity join algorithm has been recently proposed in~\cite{deng2018overlap}.
Although the containment similarity search relies on the set overlap, their technique cannot be trivially applied
because we need to construct $c$-subset inverted lists for each possible query size.
In particular, in the size-aware overlap set similarity join algorithm, it is required to build the $c$-subset inverted list 
for the given overlap threshold $c$. 
In our \gbkmv method, the threshold $c$ corresponds to $|Q|*t^{*}$ , where $|Q|$ is the query size and $t^*$ is the similarity threshold, 
thus with different query size $|Q|$, we need to build different $|Q|*t^{*}$-subset inverted lists, which is very inefficient.
\end{remark}


\noindent \textbf{Datasets.}
We deployed $7$ real-life datasets with different data properties. Note that the records with size less than 10 are discarded from dataset. We also remove the stop words (e.g., "the") from the dataset.
Table~\ref{tb:datasets} shows the detailed characteristics of the $7$ datasets.
Each dataset is illustrated with the dataset type, the representations of record,
the number of records in the dataset, the average record length, and the number of distinct
elements in the dataset.
We also report the power-law exponent $\alpha_1$ and $\alpha_2$ (skewness) of the record size and element frequency of the dataset respectively.
Note that we make use of the framework in~\cite{clauset2009power} to quantify the power-law exponent.
The dataset Canadian Open Data appears in the state-of-the-art algorithm \lshe~\cite{zhu2016lsh} .

\noindent \textbf{Settings.}
We borrow the idea from the evaluation of \lshe in ~\cite{zhu2016lsh}  to use $F_{\alpha}$ score ($\alpha$=$1,0.5$) to evaluate the accuracy of the containment similarity search.
Given a query  $Q$ randomly selected from the dataset $\mathcal{S}$ and a containment similarity threshold $t^{*}$,
we define $T = \{X: t(Q, X)\geq t^{*}, X\in \mathcal{S}\}$ as the ground truth set and $A$ as the collection of records returned by some search algorithms.
The precision and recall to evaluate the experiment accuracy are
$Precision= \frac{|T\cap A|}{|A|}$ and $Recall = \frac{|T\cap A|}{|T|}$ respectively.
The $F_\alpha$ score is defined as follows.
\begin{equation}\label{eq:and}
  F_{\alpha} = \frac{(1+\alpha^2)*Precision*Recall}{\alpha^2*Precision + Recall}
\end{equation}
Note that we use $F_{0.5}$ score because \lshe favours recall in~\cite{zhu2016lsh}.
We use the datasets from Table~\ref{tb:datasets} to evaluate the performance of our algorithm,
and we randomly choose 200 queries from the dataset.

As to the default values, the similarity threshold is set as $t^* = 0.5$.
In the experiments, we use the ratio of space budget to the total dataset size to measure
the space used. For our \gbkmv method, it is set to $10\%$.
For \lshe method, we use the same default values in~\cite{zhu2016lsh} where the signature size of each record is $256$ and the number of partition is $32$.
By varying the number of hash functions, we change the space used in \lshe.

%

\vspace{-2mm}
\subsection{Performance Tuning}

As shown in Section~\ref{subsubsec:choose_r}, we can use the variance estimation function to identify a good buffer size $r$
for \gbkmv method based on the skewness of record size and element frequency, as well as the space budget.
In Fig.~\ref{fig:tuning_pgkmv1},
we use NETFLIX and ENRON to evaluate the goodness of the function by comparing
the trend of the variance and the estimation accuracy.
By varying the buffer size $r$,
Fig.~\ref{fig:tuning_pgkmv1} reports the estimated variance (right side y axis) based on the variance function in Section~\ref{subsubsec:choose_r}
as well as the $F_1$ score (left side y axis) of the corresponding \gbkmv sketch with buffer size $r$.
Fig.~\ref{fig:tuning_pgkmv1}(a) shows that the best buffer size for variance estimation (prefer small value) is around $400$,
while the \gbkmv method achieves the best $F_1$ score (prefer large value) with buffer size around $380$.
They respectively become 220 and 230 in Fig.~\ref{fig:tuning_pgkmv1}(b).
This suggests that our variance estimation function is quite reliable to identify a good buffer size.
In the following experiments, \gbkmv method  will use buffer size suggested by this system, instead of manually tuning.


We also compare the performance of \kmv, \gkmv, and \gbkmv methods in Fig.~\ref{fig:tuning_pgkmv2}
to evaluate the effectiveness of using global threshold and the buffer
on $7$ datasets.
It is shown that the use of new KMV estimator with global threshold (i.e., Equation~\ref{eq: gkmv estimator})
can significantly improve the search accuracy.
By using a buffer whose size is suggested by the system, we can further enhance the performance under the same space budget. In the following experiments, we use \gbkmv for the performance comparison with
the state-of-the-art technique \lshe.

%

\vspace{-2mm}
\subsection{Space v.s. Accuracy}
\label{subsec:exp_sa}

An important measurement for sketch technique is the trade-off between the space and accuracy.
We evaluate the space-accuracy trade-offs of \gbkmv method and \lshe method in
Figs.~\ref{fig:space_accu_cod}-\ref{fig:space_accu_webspam} by varying the space usage
on five datasets NETFLIX, DELIC, COD, ENRON, REUTERS, WEBSPAM and WDC.
We use $F_1$ score, $F_{0.5}$ score, precision and recall to measure the accuracy.
By changing the number of hash functions, we tune the space used in \lshe.
It is reported that our \gbkmv can beat the \lshe in terms of space-accuracy trade-off
with a big margin under all settings.

We also plot the distribution of accuracy (i.e., min, max and avgerage value) to 
compare our \gbkmv method with \lshe in Fig.~\ref{fig:distribution_accuracy}.


Meanwhile, by changing the similarity threshold,
$F_1$ score is reported in Fig.~\ref{fig:threshold accu} on dataset NETFLIX and COD.
We can see that with various similarity thresholds, our \gbkmv always outperforms \lshe.

We also evaluate the space-accuracy trade-offs on synthetic datasets with
100K records in Fig.~\ref{fig:syn accu}
where the record size and the element frequency follow the zipf distribution.
We can see that on datasets with different record size and element frequency skewness,
\gbkmv consistently outperforms \lshe in terms of space-accuracy trade-off.

\subsection{Time v.s. Accuracy}

Another important measurement for the sketch technique is the trade-off between time and accuracy.
Hopefully, the sketch should be able to quickly complete the search with a good accuracy.
We tune the index size of \gbkmv to show the trade-off.
As to the \lshe, we tune the number of hash functions.
The time is reported as the average search time per query.
In Fig.~\ref{fig:time accu}, we evaluate the time-accuracy trade-offs for \gbkmv and \lshe
on four datasets COD, NETFLIX, DELIC and ENRON where the accuracy is measured by $F_1$ score.
It is shown that with the similar accuracy ($F_1$ score), \gbkmv is significantly faster than \lshe.
For datasets COD, DELIC and ENRON, \gbkmv can be 100 times faster than \lshe with the same $F_1$ score.
It is observed that the accuracy ($F_1$ score) improvement of \lshe algorithm is very slow compared with \gbkmv method. This is because the \lshe method favours recall and the precision performance is quite poor even for a large number of hash functions, resulting in a poor $F_1$ score which considers both precision and recall.

\subsection{Sketch Construction Time}
In this part, we compare the sketch construction time of \gbkmv and \lshe on different datasets under default settings. As expected, \gbkmv uses much less sketch construction time than that of \lshe
since \gbkmv sketch need only one hash function, while \lshe needs multiple for a decent accuracy.
Note that, for the internet scale dataset WDC, the index construction time for \gbkmv is around $10$ minutes, 
while for \lshe it is above $60$ minutes.
We also give the space usage of the two methods on each dataset in Table~\ref{tb:space_usage}.
The space usage of \gbkmv is $10{\%}$ as mentioned in Settings.
For \lshe in some dataset, the space is over $100{\%}$ because there are many records with size less than the number of hash functions $256$.

\begin{table}
\footnotesize
  \centering

  \vspace{-1mm}
    \begin{tabular}{|l|c|c|}
      \hline
      \cellcolor{gray!25}\textbf{Dataset} & \cellcolor{gray!25}\textbf{\gbkmv}  & \cellcolor{gray!25}\textbf{\lshe}           \\ \hline

      NETFLIX  &  10 & 118 \\ \hline
      DELIC  &  10 & 211 \\ \hline
      COD &  10 & 4 \\ \hline
      ENRON  &  10 & 185 \\ \hline
      REUTERS  &  10 & 329 \\ \hline
      WEBSPAM  &   10 & 7  \\ \hline
      WDC & 10 & 109 \\ \hline

      \hline
    \end{tabular}
\vspace{1mm}
\caption{\small The space usage($\%$)}
\label{tb:space_usage}
\vspace{-4mm}
\end{table}

\subsection{Supplementary Experiment}
\label{sct:runtimeexact}

\noindent \textbf{Evaluation on Uniform Distribution.}
In \textbf{Theorem~\ref{theo:gbkmvBetter}}, we have theoretically shown that when the dataset 
follows uniform distribution (i.e., $\alpha_1=0$ and $\alpha_2=0$), our \gbkmv method can outperform the \lshe method.
In this part, we experimentally illustrate the performance on dataset with uniform distribution.
We generate $100$k records where the record size 
is uniformly distributed between $10$ and $5000$,
and each element is randomly chosen from $100,000$ distinct elements.
Fig.~\ref{fig:last}(a) illustrates the time-accuracy trade-off of \gbkmv and \lshe on the synthetic dataset with $100K$ records. It is reported that, to achieve the same accuracy ($F_1$ score), \gbkmv consumes much less time than \lshe.

\begin{figure}[hbt]
\centering
\subfigure[Time versus Accuracy]{\includegraphics[width=0.48\linewidth]{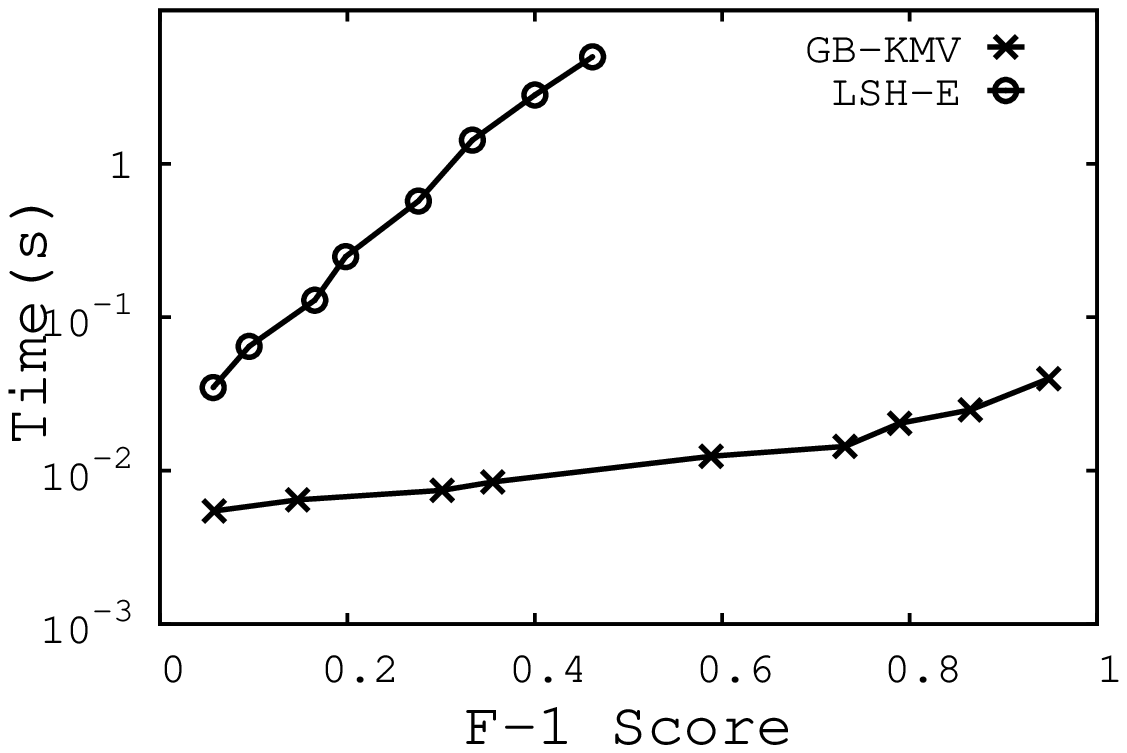}}
\subfigure[Running Time]{\includegraphics[width=0.48\linewidth]{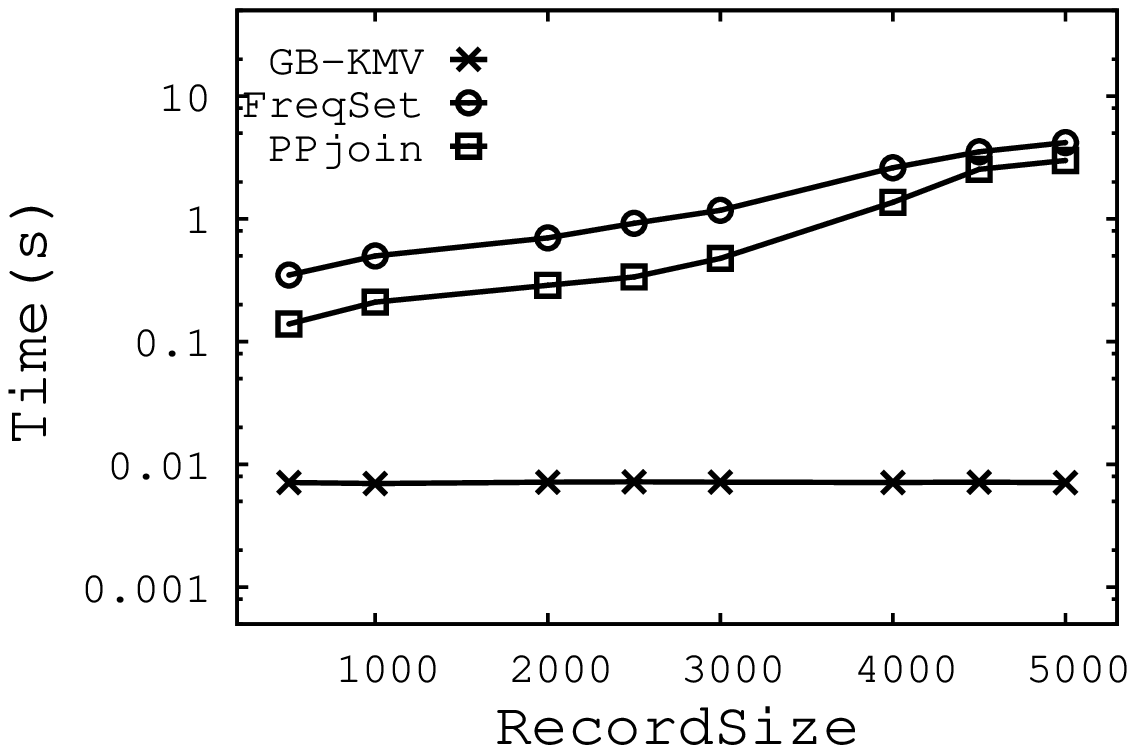}}
\vspace{-0.3cm}
\centering
\caption{\small Supplementary experiments}
\label{fig:last}
\end{figure}

\noindent \textbf{Comparison with Exact Algorithms.}
We also compare the running time of our proposed method \gbkmv with 
two exact containment similarity search methods PPjoin*~\cite{xiao2011efficient} and FreqSet~\cite{agrawal2010indexing}.
Experiments are conducted on the dataset WebSpam, which consists of $350,000$ records and has the average length around $3,700$. We partition the data into $5$ groups based on their record size with boundaries increasing from $1,000$ to $5,000$. As expected, Fig.~\ref{fig:last}(b) shows that the running time of our approximate algorithm is not sensitive to the growth of the record size because a fixed number of samples are used for a given budget.
\gbkmv outperforms two exact algorithm by a big margin, especially when the record size is large,
with a decent accuracy (i.e., with $F_1$ score and recall always larger than 0.8 and 0.9 under all settings).  

\subsection{Discussion Summary}
In the accuracy comparison between \gbkmv and \lshe, it is remarkable to see that 
the accuracy (i.e., $F_1$ score) is very low on some datasets. We give some discussions as follows.

First we should point out that in ~\cite{zhu2016lsh}, the accuracy of \lshe is only evaluated on 
\textbf{only} one dataset COD, in which both our \gbkmv method and \lshe can achieve decent accuracy performance with $F_1$ score above $0.5$. 

As mentioned in ~\ref{subsec:lshe}, the \lshe method first transforms the containment similarity 
to Jaccard similarity, then in order to make use of the efficient index techniques, \lshe partitions the dataset and uses the upper bound to approximate the record size in each partition.
which can favour recall but result in extra false positives as analysed in section ~\ref{subsec:lshe_analysis}. However, the \lshe method does not provide a partition scheme 
associated with different data distribution, and the algorithm setting (e.g., 256 hash functions and 32 partitions) can not perform well in some dataset.


\vspace{-2mm}
\section{Related Work}
\label{sct:related}

In this Section, we review two closely related categories of work on set containment similarity search.

\noindent \textbf{Exact Set Similarity Queries.}
Exact set similarity query has been widely studied in the literature.
Existing solutions are mainly based on the filtering-verification framework
which can be divided into two categories, prefix-filter based method and partition-filter based method. Prefix-filter based method is first introduced by Bayardo {\emph et al.}
in~\cite{bayardo2007scaling}. Xiao {\emph et al.}~\cite{xiao2011efficient} further improve the prefix-filter method by exploring positional filter and suffix filter techniques. In~\cite{mann2016empirical}, Mann {\emph et al.} introduce an efficient candidate verification algorithm which significantly improves the efficiency compared with the other prefix filter algorithms. Wang {\emph et al.}~\cite{wang2017leveraging} consider the relations among records in query processing to improve the performance. Deng {\emph et al.} in ~\cite{deng2017silkmoth} present an efficient similarity search method where each object is a collection of sets.  For partition-based method, in~\cite{arasu2006efficient}, Arasu {\emph et al.} devise a two-level algorithm which uses partition and enumeration techniques to search for exact similar records. Deng {et al.} in~\cite{deng2015efficient} develop a partition-based method which can effectively prune the candidate size at the cost higher filtering cost.
In~\cite{zhang2017efficient}, Zhang {\emph et al.} propose an efficient framework for exact set similarity search based on tree index structure.
In~\cite{deng2018overlap}, Deng {\emph et al.} present a size-aware algorithm which divides all the sets into small and large ones by size and processes them separately.
Regarding exact containment similarity search, Agrawal {\emph et al.} in~\cite{agrawal2010indexing} build the inverted lists on the token-sets and considered the string transformation.


\noindent \textbf{Approximate Set Similarity Queries.}
The approximate set similarity queries mostly adopt the Locality Sensitive Hashing(\lsh)~\cite{indyk1998approximate} techniques. For Jaccard similarity, MinHash~\cite{broder1998min} is used for approximate similarity search.
Asymmetric minwise hashing is a technique for approximate containment similarity search~\cite{shrivastava2015asymmetric}.
This method makes use of vector transformation by padding some values into sets, which makes all sets in the index have same cardinality as the largest set.
After the transformation, the near neighbours with respect to Jaccard similarity of the transformed sets are the same as near neighbours in containment similarity of the original sets.
Thus, MinHash \lsh can be used to index the transformed sets, such that the sets with larger containment similarity scores can be returned with higher probability.
In \cite{shrivastava2015asymmetric}, they show that asymmetric minwise hashing is advantageous in containment similarity search over datasets such as news articles and emails,
while Zhu {\emph et. al} in \cite{zhu2016lsh} finds that for datasets which are very skewed in set size distribution, asymmetric minwise hashing will reduce the recall.

The \kmv sketch technique has been widely used to estimate the cardinality of record size~\cite{zhang2010multi,cormode2012synopses,wang2014selectivity}.
The idea of imposing a global threshold on \kmv sketch is first proposed in~\cite{wang2014selectivity} in the context of term pattern size estimation.
However, there is no theoretical analysis for the estimation performance.
In~\cite{christiani2016set}, Christiani {\emph et al.} give a data structure for approximate similarity search under Braun-Blanquet similarity which has a 1-1 mapping to Jaccard similarity if all the sizes of records are fixed.
In~\cite{cohen2017minimal}, Cohen {\emph et al.} introduce a new estimator for set intersection size, but it is still based on the MinHash technique. In~\cite{dahlgaard2017fast}, Dahlgaard {\emph et al.} develop a new sketch method which has the alignment property and same concentration bounds as MinHash.

\vspace{-4mm}
\section{Conclusion}
\label{sct:conclusion}

In this paper, we study the problem of approximate containment similarity search.
The existing solutions to this problem are based on the MinHash \lsh technique.
We develop an augmented \kmv sketch technique, namely \gbkmv, which 
is data-dependent and can effectively exploit the distributions of record size and element frequency.
We provide thorough theoretical analysis to justify the design of \gbkmv, and show that  
the proposed method can outperform the state-of-the-art technique in terms of space-accuracy trade-off.
Extensive experiments on real-life set-valued datasets from a variety of applications demonstrate
the superior performance of \gbkmv method compared with the state-of-the-art technique.

\bibliographystyle{abbrv}
{\small
\bibliography{paper}
}

\end{document}